%% file: lacasa-report.tex
\documentclass[10pt,numbers,preprint]{sigplanconf}
\usepackage[final]{oopsla2016}
\pdfoutput=1

\usepackage{amsmath}
\usepackage{amsthm}
\usepackage{amssymb}
\usepackage{listings,xspace}
\usepackage{xcolor}
\usepackage{graphicx}
\usepackage{url}

\usepackage[normalem]{ulem}
\usepackage{enumitem}

\usepackage{bcprules}
\usepackage{prooftree}
\usepackage{multicol}

\usepackage{hyperref}

\usepackage{bold-extra}

\newcommand{\highlight}[1]{%
  \colorbox{gray!50}{$\displaystyle#1$}}

\lstdefinelanguage{Scala}%
{morekeywords={abstract,case,catch,char,class,%
    def,else,extends,final,%
    if,import,%
    match,module,new,null,object,override,package,private,protected,%
    public,return,super,this,throw,trait,try,type,val,var,with,implicit,%
    macro,sealed,%
  },%
  sensitive,%
  morecomment=[l]//,%
  morecomment=[s]{/*}{*/},%
  morestring=[b]",%
  morestring=[b]',%
  showstringspaces=false%
}[keywords,comments,strings]%

\newcommand{\comment}[1]{}

\newcommand{\ie}{{\em i.e.,~}}
\newcommand{\eg}{{\em e.g.,~}}

\newcommand{\ok}{{~\textbf{ok}}}
\newcommand{\this}{{\texttt{this}}}
\newcommand{\ocap}{{\texttt{ocap}}}

\newcommand*\stars{\includegraphics[height=0.8em,keepaspectratio]{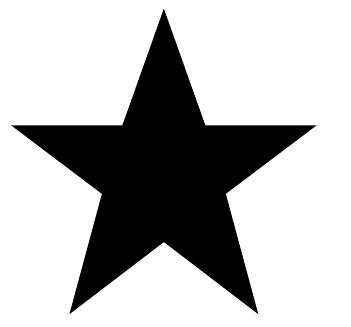}}
\newcommand*\contribs{\includegraphics[height=0.8em,keepaspectratio]{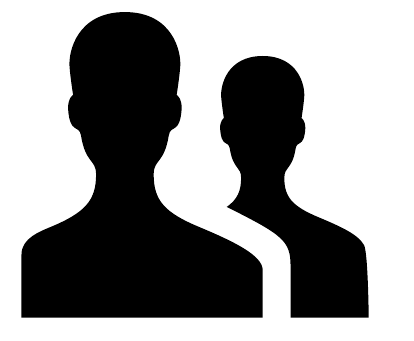}}

\newcommand{\gap}{\quad\quad}

\newcommand{\ba}{\begin{array}}
\newcommand{\ea}{\end{array}}

\newcommand{\ei}{\end{array}}
\newcommand{\bcases}{\left\{\begin{array}{ll}}
\newcommand{\ecases}{\end{array}\right.}

\newcommand{\sub}{<:}
\newcommand{\typ}{:}

\newcommand{\tfun}{\rightarrow}

\newcommand{\seq}[1]{\overline{#1}}

\newcommand{\aframe}[3]{\langle #1, #2 \rangle ^{#3}}

\newcommand{\pframe}[4]{\langle #1, #2, #3 \rangle ^{#4}}

\newcommand{\obj}[2]{\langle #1, #2 \rangle}
\newcommand{\proc}[3]{\langle #1, #2, #3 \rangle}

\newcommand{\fsreduce}[4]{#1, #2 \;\twoheadrightarrow\; #3, #4}
\newcommand{\fsreducebreak}[4]{#1, #2 \\ \;\twoheadrightarrow\; #3, #4}

\newcommand{\freduce}[4]{#1, #2 \;\longrightarrow\; #3, #4}
\newcommand{\freducebreak}[4]{#1, #2 \\ \;\longrightarrow\; #3, #4}

\newcommand{\preduce}[4]{#1, #2 \;\leadsto\; #3, #4}

\newcommand{\lacasa}{\textsc{LaCasa}}
\newcommand{\CLCONE}{\textsc{CLC}$^{1}$}
\newcommand{\CLCTWO}{\textsc{CLC}$^{2}$}
\newcommand{\CLCTHR}{\textsc{CLC}$^{3}$}

\newcommand{\tri}{\triangleright}
\newcommand{\proves}{\vdash}

\newcommand{\BoxT}[1]{\texttt{Box}[#1]}
\newcommand{\BoxV}[1]{\texttt{box}[#1]}
\newcommand{\GuaT}[2]{#1 \tri \texttt{Box}[#2]}
\newcommand{\Perm}[1]{\texttt{Perm}[#1]}
\newcommand{\ProcT}[1]{\texttt{Proc}[#1]}

\theoremstyle{plain}
\newtheorem{theorem}{Theorem}

\newtheorem{lemma}{Lemma}
\newtheorem{corollary}{Corollary}
\newtheorem*{theoremnonum}{Theorem}

\theoremstyle{definition}
\newtheorem{definition}{Definition}

\theoremstyle{remark}

\makeatletter
\def\@copyrightspace{\relax}
\makeatother

\begin{document}

\copyrightdata{978-1-nnnn-nnnn-n/yy/mm} 




\title{Object Capabilities and Lightweight Affinity in Scala}
\subtitle{Implementation, Formalization, and Soundness}

\authorinfo{Philipp Haller}
           {KTH Royal Institute of Technology, Sweden}
           {phaller@kth.se}
\authorinfo{Alexandre Loiko}
           {Google Stockholm, Sweden\titlenote{Work done while at KTH.}}
           {aleloi@google.com}

\maketitle

\begin{abstract}
Aliasing is a known source of challenges in the context of imperative
object-oriented languages, which have led to important advances in
type systems for aliasing control. However, their large-scale adoption
has turned out to be a surprisingly difficult challenge. While new
language designs show promise, they do not address the need of
aliasing control in existing languages.

This paper presents a new approach to isolation and uniqueness in an
existing, widely-used language, Scala. The approach is unique in the
way it addresses some of the most important obstacles to the adoption
of type system extensions for aliasing control. First, adaptation of
existing code requires only a minimal set of annotations. Only a
single bit of information is required per class. Surprisingly, the
paper shows that this information can be provided by the
object-capability discipline, widely-used in program security. We
formalize our approach as a type system and prove key soundness
theorems. The type system is implemented for the full Scala language,
providing, for the first time, a sound integration with Scala's local
type inference. Finally, we empirically evaluate the conformity of
existing Scala open-source code on a corpus of over 75,000 LOC.
\end{abstract}




\section{Introduction}

Uncontrolled aliasing in imperative object-oriented languages
introduces a variety of challenges in large-scale software
development. Among others, aliasing can increase the difficulty of
reasoning about program behavior and software
architecture~\cite{AldrichKC02}, and it can introduce data races in
concurrent programs. These observations have informed the development
of a number of type disciplines aimed at providing static aliasing
properties, such as linear
types~\cite{wadler90,Odersky92,FahndrichD02}, region
inference~\cite{TofteT94,TofteT97}, unique
references~\cite{Hogg91,Minsky96,Boyland01,Clarke03}, and ownership
types~\cite{NobleVP98,ClarkePN98}.

While there have been important advances in the flexibility and
expressiveness of type systems for aliasing control, large-scale
adoption has been shown to be a much greater challenge than
anticipated. Recent efforts in the context of new language designs
like Rust~\cite{Anderson15} are promising, but they do not address the
increasing need for aliasing control in existing, widely-used
languages.

One of the most important obstacles to the adoption of a type system
extension in a widely-used language with a large ecosystem is the
adaptation of existing code, including third-party
libraries. Typically, adaptation consists of adding (type) annotations
required by the type system extension. With a large ecosystem of
existing libraries, this may be prohibitively expensive even for
simple annotations. A second, and almost equally critical obstacle is
robust support for the entirety of an existing language's type system
in a way that satisfies requirements for backward compatibility.

This paper presents a new approach to integrating a flexible type
system for isolation and uniqueness into an existing, full-featured
language, Scala. Our approach minimizes the annotations necessary for
reusing existing code in a context where isolation and uniqueness is
required. In the presented system, a {\em single bit of information}
is enough to decide whether an existing class supports isolation and
uniqueness. A key insight of our approach is that this single bit of
information is provided by the {\em object-capability
  discipline}~\cite{Dennis66,Mill06a}. The object capability model is
an established methodology in the context of program security, and has
been proven in large-scale industrial use for secure sandboxing of
JavaScript applications~\cite{miller2008safe,adsafe,Politz15}.

\comment{
Object
capabilities have also been used in operating systems to enforce
ownership and isolation of resources~\cite{elkaduwe2008verified}.}

This paper makes the following contributions:

\begin{itemize}

\item We present a new approach to separation and uniqueness which
  aims to minimize the annotations necessary to reuse existing code
  (Section~\ref{sec:overview}).  In our system, reusability is based on
  the {\em object capability model}. Thus, when annotating existing
  code bases, only a single bit of information is required per class.

\item We formalize our approach in the context of two object-oriented
  core languages (Section~\ref{sec:formalization}). The first core
  language formalizes a type-based notion of object capabilities. The
  second core language additionally provides external uniqueness via
  flow-insensitive permissions.

\item We provide complete soundness proofs, formally establishing heap
  separation and uniqueness invariants for our two core languages
  (Section~\ref{sec:soundness}). We have also mechanized the operational
  semantics and type system of the first core language in Coq (Section~\ref{sec:coq}).

\item We implement our approach for the full Scala language as a
  compiler plugin (Section~\ref{sec:impl}). To our knowledge, our
  implementation of (external) uniqueness is the first to integrate
  soundly with local type inference in Scala.  Moreover, the
  implementation leverages a unique combination of previous proposals
  for (a) implicit parameters~\cite{Oliveira10,Oliveira12}, and (b)
  closures with capture control~\cite{Epstein11,Miller14}.

\item We empirically evaluate the conformity of existing Scala classes
  to the object capability model on a corpus of over 75,000 LOC of
  popular open-source projects
  (Section~\ref{sec:empirical-eval}). Results show that between 21\% and
  79\% of the classes of a project adhere to a strict object
  capability discipline.
\end{itemize}

\comment{
The complete formal semantics and soundness proof can be found in the
appendix. (To enable double-blind reviewing, a technical report with
all details is forthcoming.)}

In the following we discuss the most closely related work, and defer a
discussion of other related work to Section~\ref{sec:other-related}. In
Section~\ref{sec:conclusion} we conclude.

\paragraph{Selected Related Work.}

Most closely related to our system are approaches based on permissions
or capabilities. Of particular relevance are Haller and Odersky's
capabilities for uniqueness~\cite{Haller10} in Scala
(``Cap4S''). While their work shares our high-level goal of
lightweight unique references in Scala, the two approaches are
significantly different, with important consequences concerning
soundness, robustness, and compatibility.  First, Cap4S is based on
flow-sensitive capabilities which are modeled using Scala's
annotations, similar to the use of extended type annotations in Java 8
for pluggable type systems~\cite{Dietl11}.  However, the interaction
between Scala's local type inference~\cite{Odersky01} and annotation
propagation has been shown to be a source of unsoundness and
implementation complexities for such pluggable type
systems~\cite{Rytz13}; these challenges are exacerbated in
flow-sensitive type systems. In contrast, \lacasa~models capabilities
using Scala's implicits~\cite{Oliveira10}, an intrinsic part of type
inference in Scala. In addition, foundations of implicits have been
studied~\cite{Oliveira12}, whereas Scala's annotations remain poorly
understood. Second, \lacasa~fundamentally simplifies type checking: as
long as a class conforms to the object-capability model, \lacasa's
constructs enable isolation and uniqueness for instances of the class.
This has two important consequences: (a) a minimal set of additional
annotations (a single bit of information per class) enables reusing
existing code, and (b) type checking reusable class declarations is
simple and well-understood, following the object-capability
discipline, which we adapt for Scala.

\comment{
the annotation of existing code becomes trivial:
for each class a single bit of information

problematic, leading to unsoundness of the extended
type system and

extended type annotations of Java
8 used for pluggable type systems in the Checker
Framework~\cite{Dietl11}.

prior work by Haller and
Odersky~\cite{Haller10} on an annotation-based type system extension
for uniqueness in Scala. While the high-level goal is similar, namely,
lightweight unique references in Scala, there are major differences:
first, the interaction between local type inference~\cite{Odersky01}
and annotation propagation in Scala have been shown to be problematic,
leading to soundness issues, implementation complexities, or
both~\cite{Rytz13}. Second, while the proposed annotation system is
fairly lightweight, \lacasa~provides a simpler way to integrate
existing code: if a class follows the object-capability discipline
\lacasa~provides isolation and uniqueness for its instances.
}

\comment{
annotation-based approaches in Scala

es in Scala
suffer from 

an annotation-based approach }

\section{A Brief Overview of \lacasa}\label{sec:overview}

We proceed with an informal overview of \lacasa: its programming model
in Scala, and its type system.

\paragraph{A First Example.}

\begin{figure}
\begin{lstlisting}
class ActorA extends Actor[Any] {
  override def receive(msg: Any): Unit = msg match {
    case s: Start =>
      val newMsg = new Message
      newMsg.arr = Array(1, 2, 3, 4)
      s.next.send(newMsg)
      newMsg.arr(2) = 33
      // ...
    case other => // ...
  }
}
class ActorB extends Actor[Message] {
  override def receive(msg: Message): Unit = {
    println(msg.arr.mkString(","))
  }
}
class Message {
  var arr: Array[Int] = _
  def leak(): Unit = {
    SomeObject.fld = arr
  }
}
object SomeObject {
  var fld: Array[Int] = _
}
class Start {
  var next: ActorRef[Message] = _
}
\end{lstlisting}
\caption{Two communicating actors in Scala.}
\label{fig:actors}
\end{figure}

Consider the case of asynchronous communication between two concurrent
processes. This style of concurrency is well-supported by the actor
model~\cite{Agha86,hewitt77} for which multiple implementations exist
for Scala~\cite{Haller09,Akka,Haller12}. Figure~\ref{fig:actors} shows
the definition of two actor classes.\footnote{In favor of clarity of explanation, Figure~\ref{fig:actors} shows hypothetical Scala code which requires slight changes for compilation with the Akka~\cite{Akka} library.} The behavior of each actor is
implemented by overriding the \verb|receive| method inherited from
superclass \verb|Actor|. The \verb|receive| method is invoked by the
actor runtime system whenever an actor is ready to process an incoming
message. In the example, whenever \verb|ActorA| has received an instance
of class \texttt{Start}, it creates an instance of class \verb|Message|,
initializes the instance with an integer array, and sends the instance
to the \verb|next| actor.

Note that field \texttt{next} of class \texttt{Start} has type
\\ \texttt{ActorRef[Message]} (line 27) instead of
\texttt{Actor[Message]}. An \texttt{ActorRef} serves as an immutable
and serializable handle to an actor. The public interface of
\texttt{ActorRef} is minimal; its only purpose is to provide methods
for \\ asynchronously sending messages to the \texttt{ActorRef}'s
\\ underlying actor (an instance of a subclass of \texttt{Actor}).  The
purpose of \texttt{ActorRef} as a type separate from \texttt{Actor} is
to provide a fault handling model similar to
Erlang~\cite{Arms96a}.\footnote{Scala's original actor
  implementation~\cite{Haller09} only provided an \texttt{Actor} type;
  the distinction between \texttt{Actor} and \texttt{ActorRef} was
  introduced with the adoption of Akka as Scala's standard
  actor implementation.}  In this model, a faulty actor may be
\emph{restarted} in a way where its underlying \texttt{Actor} instance
is replaced with a new instance of the same
class. Importantly, any \texttt{ActorRef} referring to the actor that
is being restarted switches to using the new \texttt{Actor} instance
in a way that is transparent to clients (which only depend on
\texttt{ActorRef}s). This enables introducing fault-handling logic in a
modular way (cf. Erlang's OTP library~\cite{ErlangOTP}).

The shown program suffers from multiple
safety hazards: first, within the \verb|leak| method (line 19), the
array of the current \verb|Message| instance is stored in the global
singleton object \verb|SomeObject| (line 20); thus, subsequently,
multiple actors could access the array through \verb|SomeObject|
concurrently. Second, after sending \verb|newMsg| to \verb|ActorB|
(line 6), \verb|ActorA| mutates the array contained in \verb|newMsg|
(line 7); this could lead to a data race, since \verb|ActorB| may be
accessing \verb|newMsg.arr| at this point.

\begin{figure}
\begin{lstlisting}
class ActorA extends Actor[Any] {
  override def receive(box: Box[Any])
      (implicit acc: CanAccess { type C = box.C }) {
    box.open({
      case s: Start =>
        mkBox[Message] { packed =>
          val access = packed.access
          packed.box.open({ msg =>
            msg.arr = Array(1, 2, 3, 4)
          })$\highlight{(access)}$
          s.next.send(packed.box)({
            // ...
          })$\highlight{(access)}$
        }
      case other => // ...
    })$\highlight{(acc)}$
  }
}
class ActorB extends Actor[Message] {
  override def receive(box: Box[Message])
      (implicit acc: CanAccess { type C = box.C }) {
    box.open({ msg =>
      println(msg.arr.mkString(","))
    })$\highlight{(acc)}$
  }
}
\end{lstlisting}
\caption{Two communicating actors in \lacasa.}
\label{fig:actors-lacasa}
\end{figure}

\lacasa~prevents the two safety hazards of the example using two
complementary mechanisms: {\em object capabilities} and {\em affine
  access permissions}. Figure~\ref{fig:actors-lacasa} shows the same
example written in \lacasa.\footnote{This example is also included in
  the \lacasa~open-source project available at:
  \url{https://github.com/phaller/lacasa/}} \lacasa~introduces two
main changes to the \texttt{Actor} and \texttt{ActorRef} library
classes:
\begin{enumerate}
  \item Actors send and receive \emph{boxes} of type \texttt{Box[T]},
    rather than direct object references. As explained in the
    following, \lacasa's type system enforces strong encapsulation
    properties for boxes.
  \item The type of the \texttt{receive} method is changed to
    additionally include an implicit permission parameter. (We explain
    implicit permissions in detail below.)
\end{enumerate}
\noindent
Due to these changes, \lacasa~provides its own versions of the
\texttt{Actor} and \texttt{ActorRef} library classes;
Figure~\ref{fig:lacasa-actor-classes} shows the main
declarations.\footnote{In ongoing work we are developing adapter
  classes to conveniently integrate \lacasa~and the Akka actor
  library.}

\begin{figure}
\begin{lstlisting}
package lacasa

abstract class Actor[T] {
  def receive(msg: Box[T])
        (implicit acc: CanAccess { type C = msg.C })
         : Unit
  final val self: ActorRef[T] = ...
  ...
}
abstract class ActorRef[T] {
  def send(msg: Box[T])
        (cont: NullarySpore[Unit] {
                 type Excluded = msg.C
               })
        (implicit acc: CanAccess { type C = msg.C })
         : Nothing
  ...
}
\end{lstlisting}
\caption{\lacasa's \texttt{Actor} and \texttt{ActorRef} classes.}
\label{fig:lacasa-actor-classes}
\end{figure}

\paragraph{Boxes.} A box of type \verb|Box[T]| encapsulates a
reference to an object of type \verb|T|. However, this reference is
only accessible using an \texttt{open} method:
\verb|box.open({ x => ... })|; here, \verb|x| is an alias of the
encapsulated reference. For example, on line 22 \verb|ActorB| opens
the received \verb|box| in order to print the array of the
\verb|Message| instance. Note that the use of \texttt{open} on lines
4--16 relies on Scala's syntax for \emph{partial functions:} a block
of \texttt{case} clauses
\begin{lstlisting}
{
  case $pat_1$ => $e_1$
    ...
  case $pat_n$ => $e_n$
}
\end{lstlisting}
\noindent
creates a partial function with the same run-time semantics as the
function
\begin{lstlisting}
x => x match {
  case $pat_1$ => $e_1$
    ...
  case $pat_n$ => $e_n$
}
\end{lstlisting}

In combination with \lacasa's type system, boxes enforce constraints
that directly prevent the first safety hazard in the previous
example. Boxes may only encapsulate instances whose classes follow the
{\em object-capability discipline}. Roughly speaking, the
object-capability discipline prevents an object \verb|obj| from
obtaining references that were not explicitly passed to \verb|obj| via
constructor or method calls; in particular, it is illegal for
\verb|obj| to access shared, global singleton objects like
\verb|SomeObject|. As a result, the problematic leak (line 20 in
Figure~\ref{fig:actors}) causes a compilation error.

\paragraph{Capture Control.}

In general, the requirement of boxes to encapsulate object-capability
safe classes is not sufficient to ensure isolation, as the following
example illustrates:

\begin{lstlisting}
  // box: Box[Message]
  var a: Array[Int] = null
  box.open({ msg =>
    a = msg.arr
    SomeObject.fld = msg.arr
  })(acc)
  next.send(box)({
    a(2) = 33
  })(acc)
\end{lstlisting}
\noindent
In this case, by capturing variable \verb|a| in the body of
\verb|open|, and by making \verb|a| an alias of the array in
\verb|msg|, it would be possible to access the array even after
sending it (inside \verb|msg| inside \verb|box|) to \verb|next|.  To
prevent such problematic leaks, the body of \verb|open| is not allowed
to capture {\em anything} (\ie it must not have free variables).
Furthermore, the body of \verb|open| is not allowed to access global
singleton objects.  Thus, both the access to \verb|a| on line 4 and
the access to \verb|SomeObject| on line 5 cause compilation errors.
Finally, the body of \verb|open| may only create instances of
object-capability safe classes to prevent indirect leaks such as on
line 20 in Figure~\ref{fig:actors}.

The second safety hazard illustrated in Figure~\ref{fig:actors},
namely accessing a box that has been transferred, is prevented using a
combination of boxes, capture control, and access permissions, which
we discuss next.

\paragraph{Access Permissions.}

A box can only be accessed (\eg using \verb|open|) at points in the
program where its corresponding {\em access permission} is in scope.
Box operations take an extra argument which is the permission required
for accessing the corresponding box.  For example, the \texttt{open}
invocation on lines 22--24 in Figure~\ref{fig:actors-lacasa} takes the
\texttt{acc} permission as an argument (highlighted) in addition to
the closure.  Note that \texttt{acc} is passed within a separate
argument list.  The main reason for using an additional \emph{argument
  list} instead of just an additional \emph{argument} is the use of
implicits to reduce the syntactic overhead (see below).

The static types of access permissions are essential for alias
tracking. Importantly, the static types ensure that an access
permission is only compatible with a single \verb|Box[T]| instance.
For example, the \verb|acc| parameter in line 3 has type
\verb|CanAccess { type C = box.C }| where \verb|box| is a parameter of
type \verb|Box[Any]|. Thus, the {\em type member} \verb|C| of the
permission's type is equal to the type member \verb|C| of \verb|box|.

In Scala, \verb|box.C| is a {\em path-dependent
  type}~\cite{Amin14,Amin16}; \verb|box.C| is equivalent to the type
\verb|box.type#C| which selects type \verb|C| from the {\em singleton
  type} \verb|box.type|.  The type \verb|box.type| is only compatible
with singleton types \verb|x.type| where the type checker can prove
that \verb|x| and \verb|box| are always aliases.  (Thus, in a type
\verb|box.type|, \verb|box| may not be re-assignable.) Access
permissions in \lacasa~leverage this aliasing property of singleton
types: since it is impossible to create a box \verb|b| such that
\verb|b.C| is equal to the type member \verb|C| of an existing box, it
follows that an access permission is only compatible with at most one
instance of \verb|Box[T]|.

The only way to create an access permission is by creating a box using
an \verb|mkBox| expression. For example, the \verb|mkBox| expression
on line 6 in Figure~\ref{fig:actors-lacasa} creates a box of type
\verb|Box[Message]| as well as an access permission. Besides a type
argument, \verb|mkBox| also receives a {\em closure} of the form
\verb|{ packed => ... }|. The closure's \verb|packed| parameter
encapsulates both the box and the access permission, since both need
to be available in the scope of the closure.

Certain operations {\em consume} access permissions, causing
associated boxes to become unavailable. For example, the message send
on line 11 consumes the access permission of \verb|packed.box| to
prevent concurrent accesses from the sender and the receiver. As a
result, \verb|packed.box| is no longer accessible in the continuation
of \verb|send|.

Note that permissions in \lacasa~are {\em flow-insensitive} by
design. Therefore, the only way to change the set of available
permissions is by entering scopes that prevent access to consumed
permissions. In \lacasa, this is realized using \emph{continuation
  closures:} each operation that changes the set of available
permissions also takes a closure that is the continuation of the
current computation; the changed set of permissions is only visible in
the continuation closure.  Furthermore, by \emph{discarding the call
  stack} following the execution of a continuation closure,
\lacasa~enforces that scopes where consumed permissions are visible
(and therefore ``accessible'') are never re-entered. The following
\lacasa~operations discard the call stack: \texttt{mkBox},
\texttt{send}, \texttt{swap} (see below). In contrast, \texttt{open}
does not discard the call stack, since it does not change the set of
permissions.

In the example, the \verb|send| operation takes a {\em continuation closure}
(line 11--13) which prevents access to the permission of
\verb|packed.box|; furthermore, the call stack is discarded, making
the code from line 14 unreachable.

\paragraph{Implicit Permissions.}

To make sure access permissions do not have to be explicitly threaded
through the program, they are modeled using {\em
  implicits}~\cite{Oliveira10,Oliveira12}. For example, consider the
\verb|receive| method in line 2--3. In addition to its regular
\verb|box| parameter, the method has an \verb|acc| parameter which is
marked as \verb|implicit|. This means at invocation sites of
\verb|receive| the argument passed to the implicit parameter is {\em
  inferred} (or resolved) by the type checker. Importantly, implicit
resolution fails if no type-compatible implicit value is in scope, or
if multiple ambiguous type-compatible implicit values are in
scope.\footnote{See the Scala language specification~\cite{Odersky14}
  for details of implicit resolution.} The benefit of marking
\verb|acc| as \verb|implicit| is that within the body of
\verb|receive|, \verb|acc| does not have to be passed explicitly to
methods requiring access permissions, including \lacasa~expressions
like \verb|open|. Figure~\ref{fig:actors-lacasa} makes all uses of
implicits explicit (shaded). This explicit style also requires making
parameter lists explicit; for example, consider lines 22--24:
\verb|box.open({ x => ... })(acc)|. In contrast, passing the access
permission \texttt{acc} implicitly enables the more lightweight Scala
syntax \verb|box open { x => ... }|.

\paragraph{Stack Locality.}

It is important to note that the above safety measures with respect to
object capabilities, capture control, and access permissions could be
circumvented by creating heap aliases of boxes and permissions.
Therefore, boxes and permissions are confined to the stack by
default. This means, without additional annotations they cannot be
stored in fields of heap objects or passed as arguments to
constructors.

\paragraph{Unique Fields.}

Strict stack confinement of boxes would be too restrictive in
practice. For example, an actor might have to store a box in the heap
to maintain access across several invocations of its message handler
while enabling a subsequent ownership transfer. To support such
patterns \lacasa~enables actors and boxes to have {\em unique fields}
which store boxes. Access is restricted to maintain {\em external
  uniqueness}~\cite{Clarke03} of unique fields (see
Section~\ref{sec:formalization} for a formalization of the uniqueness
and aliasing guarantees of unique fields).

\begin{figure}
\begin{lstlisting}
class ActorA(next: ActorRef[C])
    extends Actor[Container] {
  def receive(msg: Box[Container])
      (implicit acc: CanAccess { type C = msg.C }) {
    mkBox[C] { packed =>
      // ...
      msg.swap(_.part1)(_.part1 = _, b)(
        spore { pack =>
          implicit val acc = pack.access

          pack.box.open({ part1Obj =>
            println(part1Obj.arr.mkString(","))
            part1Obj.arr(0) = 1000
          })(acc)

          next.send(pack.box)({
            // ...
          })(acc)
        }
      )
    }
  }
}
class Container {
  var part1: Box[C] = _
  var part2: Box[C] = _
}
class C {
  var arr: Array[Int] = _
}
\end{lstlisting}
\caption{An actor accessing a unique field via \texttt{swap}.}
\label{fig:actors-unique}
\end{figure}

Figure~\ref{fig:actors-unique} shows an actor receiving a box of type
\\ \verb|Box[Container]| where class \verb|Container| declares two
unique fields, \verb|part1| and \verb|part2| (line 23--24).  In
\lacasa, these fields, which are identified as unique fields through
their box types, have to be accessed using a \verb|swap| expression.
\verb|swap| ``removes'' the box of a unique field and replaces it with
another box. For example, in line 7, \verb|swap| extracts the
\verb|part1| field of the \verb|msg| box and replaces it with some
other box \verb|b|. The extracted box is accessible via the
\verb|pack| parameter of the subsequent ``spore.''

\comment{
As explained in
Section~\ref{sec:back}, a spore is a closure which allows
``excluding'' specific types from occurring in its body.)
}

\comment{
In the current version of \lacasa, \verb|swap| requires specifying the
target field using two well-formed accessor functions (\ie it is
illegal to get one field and set another, say). A more ``invasive''
approach which adds direct support for \verb|swap| has been proposed
before~\cite{Haller10}, but it has the disadvantage of breaking
compatibility with existing Scala tools (\eg IDEs) which do not
support this extension.}

\comment{ In order

The example in Figure~\ref{fig:actors-lacasa} shows how 

Within the body of the \verb|open| expression
access to \verb|x| is unrestricted.}

\comment{
To make
both the box and the access permission available to the argument
closure of \verb|mkBox|.

of the form \verb|mkBox[T] { p => ... }|.

Both
the box and the access permission}

\comment{
At certain program points, access permissions are introduced, \eg, when creating 
a new box; at other program points, access permissions go out of scope.
To avoid

Scala's type checker infers 

implicitly passes a compatible
value in scope, or fails if no compatible implicit value is in scope.

not only takes a regular parameter of type \verb|Box[Any]|, 
}

\section{Formalization}\label{sec:formalization}

In this section we formalize the main concepts of \lacasa~in the
context of typed object-oriented core languages. Our approach,
however, extends to the whole of Scala (see Section~\ref{sec:impl}).

Our technical development proceeds in two steps. In the first step, we
formalize simple object capabilities. In combination with \lacasa's
boxes and its \verb|open| construct object capabilities enforce an
essential heap separation invariant (Section~\ref{sec:core1-heap-sep}).

In the second step, we extend our first core language with lightweight
affinity based on flow-insensitive permissions. The extended core
language combines permissions and continuation terms to enable
expressing (external) uniqueness, ownership transfer, and unique
fields.

Soundness, isolation, and uniqueness invariants are established based
on small-step operational semantics and syntax-directed type rules.

\subsection{Object Capabilities}

This section introduces \textsc{CoreLaCasa}$^{1}$ (\CLCONE), a
typed, object-oriented core language with object capabilities.

\paragraph{Syntax}

Figure~\ref{fig:core1-syntax} and Figure~\ref{fig:core1-syntax2} show
the syntax of \CLCONE. A program consists of a sequence of class
definitions, $\seq{cd}$, a sequence of global variable declarations,
$\seq{vd}$, and a ``main'' term $t$. Global variables model top-level,
stateful singleton objects of our realization in Scala (see
Section~\ref{sec:impl}). A class $C$ has exactly one superclass
$D$ and a (possibly empty) sequence of fields, $\seq{vd}$, and
methods, $\seq{md}$. The superclass may be \verb|AnyRef|, the
superclass of all classes. To simplify the presentation, methods have
exactly one parameter $x$; their body is a term $t$. There are three
kinds of types: class types $C$, box types $\texttt{Box}[C]$, and the
\verb|Null| type. \verb|Null| is a subtype of all class
types; it is used to assign a type to \verb|null|.


\begin{figure}[t]
  \centering
$\ba[t]{l@{\hspace{2mm}}l}
p    ::=  \seq{cd}~\seq{vd}~t                                       & \mbox{program}  \\
cd   ::=  \texttt{class}~C~\texttt{extends}~D~\{\seq{vd}~\seq{md}\} & \mbox{class}    \\
vd   ::=  \texttt{var}~f \typ C                                     & \mbox{variable} \\
md   ::=  \texttt{def}~m(x \typ \sigma) \typ \tau = t               & \mbox{method}   \\
\sigma,\tau ::=                                                     & \mbox{type}     \\
\gap C, D                                                           & \gap\mbox{class type} \\
\gap ~|~  \texttt{Box}[C]                                           & \gap\mbox{box type}   \\
\gap ~|~  \texttt{Null}                                             & \gap\mbox{null type}  \\
\ea$
\caption{\CLCONE~syntax. $C$, $D$ range over class names, $f$, $m$,
  $x$ range over term names.}
  \label{fig:core1-syntax}
\end{figure}

\begin{figure}[t]
  \centering

$\ba[t]{l@{\hspace{2mm}}l}
t    ::=                                                   & \mbox{terms}               \\
\gap     x                                                 & \gap\mbox{variable}        \\
\gap ~|~ \texttt{let}~x = e~\texttt{in}~t                  & \gap\mbox{let binding}     \\
~ & ~ \\
e    ::=                                                 & \mbox{expressions}           \\
\gap     \texttt{null}                                   & \gap\mbox{null reference}    \\
\gap ~|~ x                                               & \gap\mbox{variable}          \\
\gap ~|~ x.f                                             & \gap\mbox{selection}         \\
\gap ~|~ x.f = y                                         & \gap\mbox{assignment}        \\
\gap ~|~ \texttt{new}~C                                  & \gap\mbox{instance creation} \\
\gap ~|~ x.m(y)                                          & \gap\mbox{invocation}        \\
\gap ~|~ \texttt{box}[C]                                 & \gap\mbox{box creation}      \\
\gap ~|~ x.\texttt{open}~\{ y \Rightarrow t \}           & \gap\mbox{open box}          \\
\ea$
  \caption{\CLCONE~terms and expressions.}
  \label{fig:core1-syntax2}
\end{figure}

In order to simplify the presentation of the operational semantics,
programs are written in {\em A-normal form}~\cite{FlanaganSDF93} (ANF)
which requires all subexpressions to be named. We enforce ANF by
introducing two separate syntactic categories for {\em terms} and {\em
  expressions}, shown in Figure~\ref{fig:core1-syntax2}. Terms are
either variables or let bindings. Let bindings introduce names for
intermediate results. Most expressions are standard, except that the
usual object-based expressions, namely field selections, field
assignments, and method invocations, have only variables as trivial
subexpressions. The instance creation expression (\verb|new|) does not
take arguments: all fields of newly-created objects are initialized to
\verb|null|.

Two kinds of expressions are unique to our core language:
$\texttt{box}[C]$ creates a box containing a new instance of class
$C$. The expression $\texttt{box}[C]$ has type $\texttt{Box}[C]$. The
expression $x.\texttt{open}~\{ y \Rightarrow t \}$ provides temporary
access to box $x$.

\subsubsection{Dynamic Semantics}

We formalize the dynamic semantics as a small-step operational
semantics based on two reduction relations, $\freduce H F {H'} {F'}$,
and $\fsreduce H {FS} {H'} {FS'}$. The first relation reduces single
(stack) frames $F$ in heap $H$, whereas the second relation reduces
entire frame stacks $FS$ in heap $H$.

A heap $H$ maps references $o \in dom(H)$ to run-time objects $\obj C
{FM}$ where $C$ is a class type and $FM$ is a field map that maps
field names to values in $dom(H) \cup \{\texttt{null}\}$. $FS$ is a
sequence of stack frames $F$; we use the notation $FS = F \circ FS'$
to indicate that in stack $FS$ frame $F$ is the top-most frame which
is followed by frame stack $FS'$.

A single frame $F = \aframe L t l$ consists of a variable environment
$L = env(F)$, a term $t$, and an annotation $l$. The variable environment $L$
maps variable names $x$ to values $v \in dom(H) \cup \{\texttt{null}\}
\cup \{ b(o) ~|~ o \in dom(H) \}$. A value $b(o)$ is a {\em box
  reference} created using \CLCONE's $\texttt{box}[C]$ expression. A
box reference $b(o)$ prevents accessing the members of $o$ using
regular selection, assignment, and invocation expressions; instead,
accessing $o$'s members requires the use of an \verb|open| expression
to temporarily ``borrow'' the encapsulated reference. As is common,
$L' = L[x \mapsto v]$ denotes the updated mapping where $L'(y) = L(y)$
if $y \neq x$ and $L'(y) = v$ if $y = x$. A frame annotation $l$ is
either empty (or non-existant), expressed as $l = \epsilon$, or equal
to a variable name $x$. In the latter case, $x$ is the name of a
variable in the next frame which is to be assigned the return value of
the current frame.

As is common~\cite{Igarashi01}, $fields(C)$ denotes the fields of
class $C$, and $mbody(C, f) = x \rightarrow t$ denotes the body of a
method $\texttt{def}~m(x \typ \sigma) \typ \tau = t$.

Reduction of a program $p = \seq{cd}~\seq{vd}~t$ begins in an initial
environment $H_0, F_0 \circ \epsilon$ such that \\ $H_0 = \{ o_g \mapsto \obj{C_g}{FM_g} \}$ (initial heap), $F_0 = \aframe{L_0}{t}{\epsilon}$ (initial frame), $L_0 = \{ \texttt{global} \mapsto o_g \}$, $FM_g = \{ x \mapsto \texttt{null} ~|~ \texttt{var}~x \typ C \in \seq{vd} \}$, and $o_g$ a fresh object identifier; $C_g$ is a synthetic class defined as: \\ $\texttt{class}~C_g~\texttt{extends}~\texttt{AnyRef}~\{\seq{vd}\}$. Thus, a global variable $\texttt{var}~x \typ C$ is accessed using $\texttt{global}.x$; we treat \texttt{global} as a reserved variable name.

\paragraph{Single Frame Reduction.}

\begin{figure}
  \centering

  \infax[\textsc{E-Null}] {
    \freducebreak H {\aframe L {\texttt{let}~x = \texttt{null}~\texttt{in}~t} l} H {\aframe {L[x \mapsto \texttt{null}]} t l}
  }

  \vspace{0.3cm}

  \infax[\textsc{E-Var}] {
    \freducebreak H {\aframe L {\texttt{let}~x = y~\texttt{in}~t} l} H {\aframe {L[x \mapsto L(y)]} t l}
  }

  \vspace{0.3cm}

  \infrule[\textsc{E-Select}] {
    H(L(y)) = \obj C {FM} \andalso f \in dom(FM)
  } {
    \freducebreak H {\aframe L {\texttt{let}~x = y.f~\texttt{in}~t} l} H {\aframe {L[x \mapsto FM(f)]} t l}
  }

  \vspace{0.3cm}

  \infrule[\textsc{E-Assign}] {
    L(y) = o \andalso H(o) = \obj C {FM} \\
    H' = H[o \mapsto \obj C {FM[f \mapsto L(z)]}] \\
  } {
    \freducebreak H {\aframe L {\texttt{let}~x = y.f = z~\texttt{in}~t} l} {H'} {\aframe L {\texttt{let}~x = z~\texttt{in}~t} l}
  }

  \vspace{0.3cm}

  \infrule[\textsc{E-New}] {
    o \notin dom(H) \andalso fields(C) = \seq{f} \\
    H' = H[o \mapsto \obj C {\seq{f \mapsto \texttt{null}}}] \\
  } {
    \freducebreak H {\aframe L {\texttt{let}~x = \texttt{new}~C~\texttt{in}~t} l} {H'} {\aframe {L[x \mapsto o]} t l}
  }

  \vspace{0.3cm}

  \infrule[\textsc{E-Box}] {
    o \notin dom(H) \andalso fields(C) = \seq{f} \\
    H' = H[o \mapsto \obj C {\seq{f \mapsto \texttt{null}}}] \\
  } {
    \freducebreak H {\aframe L {\texttt{let}~x = \texttt{box}[C]~\texttt{in}~t} l} {H'} {\aframe {L[x \mapsto b(o)]} t l}
  }

  \caption{\CLCONE~frame transition rules.}
  \label{fig:core1-frame-rules}
\end{figure}

Figure~\ref{fig:core1-frame-rules} shows single frame transition
rules. Thanks to the fact that terms are in ANF in our core language,
the reduced term is a let binding in each case. This means reduction
results can be stored immediately in the variable environment,
avoiding the introduction of locations or references in the core
language syntax. Rule \textsc{E-Box} is analogous to rule
\textsc{E-New}, except that variable $x$ is bound to a box reference
$b(o)$. As a result, fields of the encapsulated object $o$ are not
accessible using regular field selection and assignment, since rules
\textsc{E-Select} and \textsc{E-Assign} would not be applicable. Apart
from \textsc{E-Box} the transition rules are similar to previous
stack-based formalizations of class-based core languages with
objects~\cite{UCAM-CL-TR-563,OstlundW10,FormalizingAsync}.

\paragraph{Frame Stack Reduction.}

\begin{figure}
  \centering

  \infrule[\textsc{E-Invoke}] {
    H(L(y)) = \obj C {FM} \\
    mbody(C, m) = x \rightarrow t' \\
    L' = L_0[\texttt{this} \mapsto L(y), x \mapsto L(z)]
  } {
    \fsreducebreak H {\aframe L {\texttt{let}~x = y.m(z)~\texttt{in}~t} l \circ FS} H {\aframe {L'} {t'} x \circ \aframe L t l \circ FS}
  }

  \vspace{0.3cm}

  \infax[\textsc{E-Return1}] {
    \fsreducebreak H {\aframe L x y \circ \aframe {L'} {t'} l \circ FS} H {\aframe {L'[y \mapsto L(x)]} {t'} l \circ FS}
  }

  \vspace{0.3cm}

  \infax[\textsc{E-Return2}] {
    \fsreducebreak H {\aframe L x {\epsilon} \circ \aframe {L'} {t'} l \circ FS} H {\aframe {L'} {t'} l \circ FS}
  }

  \vspace{0.3cm}

  \infrule[\textsc{E-Open}] {
    L(y) = b(o) \andalso L' = [z \mapsto o]
  } {
    \fsreducebreak H {\aframe L {\texttt{let}~x = y.\texttt{open}~\{ z \Rightarrow t' \}~\texttt{in}~t} l \circ FS} H {\aframe {L'} {t'} {\epsilon} \circ \aframe {L[x \mapsto L(y)]} t l \circ FS}
  }

\comment{
\vspace{0.3cm}

\infrule[\textsc{E-Frame}]
{ \freduce H F {H'} {F'}
}
{ \fsreduce H {F \circ FS} {H'} {F' \circ FS}
}}

  \caption{\CLCONE~frame stack transition rules.}
  \label{fig:core1-frame-stack-rules}
\end{figure}

Figure~\ref{fig:core1-frame-stack-rules} shows the frame stack
transition rules. Rule \textsc{E-Invoke} creates a new frame,
annotated with $x$, that evaluates the body of the called method.
Rule \textsc{E-Return1} uses the annotation $y$ of the top-most frame
to return the value of $x$ to its caller's frame. Rule
\textsc{E-Return2} enables returning from an $\epsilon$-annotated
frame. Rule \textsc{E-Open} creates such an $\epsilon$-annotated
frame. In the new frame, the object encapsulated by box $y$ is
accessible under alias $z$. In contrast to \textsc{E-Invoke}, the new
frame does not include the global environment $L_0$; instead, $z$ is
the only variable in the (domain of the) new environment.

\subsubsection{Static Semantics}

\begin{figure}
  \centering

  \infrule[\textsc{WF-Program}] {
    p \vdash \seq{cd} \andalso p \vdash \Gamma_0 \andalso \Gamma_0~;~\epsilon \vdash t \typ \sigma
  } {
    p \vdash \seq{cd}~\seq{vd}~t
  }

  \vspace{0.3cm}

  \infrule[\textsc{WF-Class}] {
    C \vdash \seq{md} \andalso D = \texttt{AnyRef} \lor p \vdash \texttt{class}~D~\ldots \\
    \forall~(\texttt{def}~m~\ldots) \in \seq{md}.~override(m, C, D) \\
    \forall~\texttt{var}~f \typ \sigma \in \seq{fd}.~f \notin fields(D)
  } {
    p \vdash \texttt{class}~C~\texttt{extends}~D~\{\seq{fd}~\seq{md}\}
  }

  \vspace{0.3cm}

  \infrule[\textsc{WF-Override}] {
    mtype(m, D)~\text{not defined} \lor mtype(m, D) = mtype(m, C)
  } {
    override(m, C, D)
  }

  \vspace{0.3cm}

  \infrule[\textsc{WF-Method}] {
    \Gamma_0 , \texttt{this} \typ C, x \typ \sigma ~;~ \epsilon \vdash t \typ \tau' \\
    \tau' \sub \tau
  } {
    C \vdash \texttt{def}~m(x \typ \sigma) \typ \tau = t
  }

  \caption{Well-formed \CLCONE~programs.}
  \label{fig:wf-programs}
\end{figure}

\paragraph{Type Assignment.}

A judgement of the form $\Gamma ~;~ a \vdash t \typ \sigma$ assigns
type $\sigma$ to term $t$ in type environment $\Gamma$ under effect
$a$. When assigning a type to the top-level term of a program the
effect $a$ is $\epsilon$ which is the unrestricted effect. In
contrast, the body of an \verb|open| expression must be well-typed
under effect \verb|ocap| which requires instantiated classes to be
$ocap$.

\paragraph{Well-Formed Programs.}

Figure~\ref{fig:wf-programs} shows the rules for well-formed
programs. (We write $\ldots$ to omit unimportant parts of a program.)
A program is well-formed if all its class definitions are well-formed
and its top-level term is well-typed in type environment $\Gamma_0 =
\{ \texttt{global} \typ C_g \}$
(\textsc{WF-Program}); class $C_g$ is defined as: $\texttt{class}~C_g~\texttt{extends}~\texttt{AnyRef}~\{\seq{vd}\}$.  Rule \textsc{WF-Class} defines well-formed
class definitions. In a well-formed class definition (a) all methods
are well-formed, (b) the superclass is either \verb|AnyRef| or a
well-formed class in the same program, and (c) method overriding (if
any) is well-formed; fields may not be overridden. We use a standard
function $fields(D)$~\cite{Igarashi01} to obtain the fields in class
$D$ and superclasses of $D$. Rule \textsc{WF-Override} defines
well-formed method overriding: overriding of a method $m$ in class $C$
with superclass $D$ is well-formed if $D$ (transitively) does not
define a method $m$ or the type of $m$ in $D$ is the same as the type
of $m$ in $C$. A method $m$ is well-typed in class $C$ if its body is
well-typed with type $\tau'$ under effect $\epsilon$ in environment
$\Gamma_0, \texttt{this} \typ C, x \typ \sigma$, where $\sigma$ is the
type of m's parameter $x$, such that $\tau'$ is a subtype of m's
declared result type $\tau$ (\textsc{WF-Method}).

\paragraph{Object Capabilities.}

\begin{figure}
  \centering

  \infax[\textsc{Ocap-AnyRef}] { ocap(\texttt{AnyRef}) }

  \infrule[\textsc{Ocap-Class}] {
    p \vdash \texttt{class}~C~\texttt{extends}~D~\{\seq{fd}~\seq{md}\} \\
    C \vdash_{ocap} \seq{md} \andalso ocap(D) \\
    \forall~\texttt{var}~f \typ E \in \seq{fd}.~ocap(E)
  } {
    ocap(C)
  }
  \vspace{0.3cm}
  \infrule[\textsc{Ocap-Method}] {
    \texttt{this} \typ C, x \typ \sigma ~;~ \texttt{ocap} \vdash t \typ \tau' \\
    \tau' \sub \tau
  } {
    C \vdash_{ocap} \texttt{def}~m(x \typ \sigma) \typ \tau = t
  }

  \caption{Object capability rules.}
  \label{fig:ocap-rules}
\end{figure}

For a class $C$ to satisfy the constraints of the object-capability
discipline, written $ocap(C)$, it must be well-formed according to the
rules shown in Figure~\ref{fig:ocap-rules}. Essentially, for a class
$C$ we have $ocap(C)$ if its superclass is ocap, the types of its
fields are ocap, and its methods are well-formed according to
$\vdash_{ocap}$. Rule \textsc{Ocap-Method} looks a lot like rule
\textsc{WF-Method}, but there are two essential differences: first,
the method body must be well-typed in a type environment that does not
contain the global environment $\Gamma_0$; thus, global variables are
inaccessible. Second, the method body must be well-typed under effect
\verb|ocap|; this means that within the method body only ocap classes
may be instantiated.

\paragraph{Subclassing and Subtypes.}

In \CLCONE, the subtyping relation $\sub$, defined by the class table,
is identical to that of FJ~\cite{Igarashi01} except for two additional
rules:

\setlength{\columnsep}{0pt}
\begin{multicols}{2}
  \infrule[\textsc{$\sub$-Box}]
          { C \sub D }
          { \texttt{Box}[C] \sub \texttt{Box}[D] }
  \infax[\textsc{$\sub$-Null}]
        { \texttt{Null} \sub \sigma }
\end{multicols}

\begin{figure}[t]
  \centering

  \begin{multicols}{2}
  \infax[\textsc{T-Null}] {
    \Gamma ~;~ a \vdash \texttt{null} \typ \texttt{Null}
  }

  \infrule[\textsc{T-Var}] {
    x \in dom(\Gamma)
  } {
    \Gamma ~;~ a \vdash x \typ \Gamma(x)
  }
  \end{multicols}

  \vspace{0.3cm}

  \infrule[\textsc{T-Let}] {
    \Gamma ~;~ a \vdash e \typ \tau \andalso \Gamma , x \typ \tau ~;~ a \vdash t \typ \sigma
  } {
    \Gamma ~;~ a \vdash \texttt{let}~x = e~\texttt{in}~t \typ \sigma
  }

  \vspace{0.3cm}

  \infrule[\textsc{T-Select}] {
    \Gamma ~;~ a \vdash x \typ C \andalso ftype(C, f) = D
  } {
    \Gamma ~;~ a \vdash x.f \typ D
  }

  \vspace{0.3cm}

  \infrule[\textsc{T-Assign}] {
    \Gamma ~;~ a \vdash x \typ C \andalso ftype(C, f) = D \\
    \Gamma ~;~ a \vdash y \typ D' \andalso D' \sub D
  } {
    \Gamma ~;~ a \vdash x.f = y \typ D
  }

  \vspace{0.3cm}

  \infrule[\textsc{T-New}] {
    a = \texttt{ocap} \implies ocap(C)
  } {
    \Gamma ~;~ a \vdash \texttt{new}~C \typ C
  }

  \vspace{0.3cm}

  \infrule[\textsc{T-Invoke}] {
    \Gamma ~;~ a \vdash x \typ C \andalso mtype(C, m) = \sigma \tfun \tau \\
    \Gamma ~;~ a \vdash y \typ \sigma' \andalso \sigma' \sub \sigma
  } {
    \Gamma ~;~ a \vdash x.m(y) \typ \tau
  }

  \vspace{0.3cm}

  \infrule[\textsc{T-Box}] {
    ocap(C)
  } {
    \Gamma ~;~ a \vdash \BoxV{C} \typ \BoxT{C}
  }

  \vspace{0.3cm}

  \infrule[\textsc{T-Open}] {
    \Gamma ~;~ a \vdash x \typ \BoxT{C} \andalso y \typ C ~;~ \texttt{ocap} \vdash t \typ \sigma
  } {
    \Gamma ~;~ a \vdash x.\texttt{open}~\{ y \Rightarrow t \} \typ \BoxT{C}
  }

  \caption{\CLCONE~term and expression typing.}
  \label{fig:core1-type-rules}
\end{figure}

\paragraph{Term and Expression Typing.}

Figure~\ref{fig:core1-type-rules} shows the inference rules for typing
terms and expressions. The type rules are standard except for
\textsc{T-New}, \textsc{T-Box}, and \textsc{T-Open}. Under effect
\ocap~\textsc{T-New} requires the instantiated class to be ocap. An
expression $\texttt{box}[C]$ has type $\texttt{Box}[C]$ provided
$ocap(C)$ holds (\textsc{T-Box}). Finally, \textsc{T-Open} requires
the body $t$ of \verb|open| to be well-typed under effect \ocap~and in
a type environment consisting only of $y \typ C$. The type of the
\verb|open| expression itself is $\BoxT{C}$ (it simply returns box
$x$).

\begin{figure}[t]
  \centering
  \infrule[\textsc{WF-Var}] {
    L(x) = \texttt{null}~ \lor \\
    typeof(H, L(x)) \sub \Gamma(x)
  } {
    H \vdash \Gamma ; L ; x
  }

  \vspace{0.3cm}

  \infrule[\textsc{WF-Env}] {
    dom(\Gamma) \subseteq dom(L) \\
    \forall x \in dom(\Gamma).~H \vdash \Gamma ; L ; x
  } {
    H \vdash \Gamma ; L
  }

  \vspace{0.3cm}

  \begin{multicols}{2}
  \infrule[\textsc{T-Frame1}] {
    \Gamma ~;~ a \vdash t \typ \sigma \\
    H \vdash \Gamma ; L
  } {
    H \vdash \aframe L t l \typ \sigma
  }
  \infax[\textsc{T-EmpFS}] {
    H \vdash \epsilon
  }
  \end{multicols}

  \vspace{0.3cm}

  \infrule[\textsc{T-Frame2}] {
    \Gamma, x \typ \tau ~;~ a \vdash t \typ \sigma \andalso H \vdash \Gamma ; L
  } {
    H \vdash^{\tau}_x \aframe L t l \typ \sigma
  }

  \vspace{0.3cm}

  \begin{multicols}{2}
  \infrule[\textsc{T-FS-NA}] {
    H \vdash F^{\epsilon} \typ \sigma \\
    H \vdash FS
  } {
    H \vdash F^{\epsilon} \circ FS
  }
  \infrule[\textsc{T-FS-NA2}] {
    H \vdash^{\tau}_x F^{\epsilon} \typ \sigma \\
    H \vdash FS
  } {
    H \vdash^{\tau}_x F^{\epsilon} \circ FS
  }
  \end{multicols}

  \vspace{0.3cm}

  \begin{multicols}{2}
  \infrule[\textsc{T-FS-A}] {
    H \vdash F^x \typ \tau \\
    H \vdash^{\tau}_x FS
  } {
    H \vdash F^x \circ FS
  }
  \infrule[\textsc{T-FS-A2}] {
    H \vdash^{\sigma}_y F^x \typ \tau \\
    H \vdash^{\tau}_x FS
  } {
    H \vdash^{\sigma}_y F^x \circ FS
  }
  \end{multicols}

  \vspace{0.3cm}

  \caption{Well-formed environments, frames, and frame stacks.}
  \label{fig:wf-env-frames}
\end{figure}

\paragraph{Well-Formedness.}

Frames, frame stacks, and heaps must be
well-formed. Figure~\ref{fig:wf-env-frames} shows the well-formedness
rules for environments, frames, and frame stacks. Essentially, $\Gamma
, L$ are well-formed in heap $H$ if for all variables $x \in
dom(\Gamma)$ the type of $L(x)$ in $H$ is a subtype of the static type
of $x$ in $\Gamma$ (\textsc{WF-Var}, \textsc{WF-Env}). A frame
$\aframe L t l$ is well-typed in $H$ if its term $t$ is well-typed in
some environment $\Gamma$ such that $\Gamma , L$ are well-formed in
$H$ (\textsc{T-Frame1}, \textsc{T-Frame2}). A frame stack is
well-formed if all its frames are well-typed. Rules \textsc{T-FS-NA}
and \textsc{T-FS-NA2} are required for $\epsilon$-annotated frames.
Well-typed heaps are defined as follows.

\begin{definition}[Object Type]\label{def:typeof}
  For an object identifier $o \in dom(H)$ where $H(o) = \obj C {FM}$,
  $typeof(H, o) := C$
\end{definition}

\begin{definition}[Well-typed Heap]\label{def:well-typed-heap}
  A heap $H$ is well-typed, written $\vdash H \typ \star$ iff
  \begin{align*}
    \forall o \in dom(H).~ & H(o) = \obj C {FM} \implies \\
                           & (dom(FM) = fields(C)~ \land \\
                           & \forall f \in dom(FM).~FM(f) = \texttt{null}~ \lor \\
                           & typeof(H, FM(f)) \sub ftype(C, f)) \\
  \end{align*}
\end{definition}
\noindent
To formalize the heap structure enforced by \CLCONE~we use the
following definitions.

\begin{definition}[Separation]\label{def:separation}
  Two object identifiers $o$ and $o'$ are separate in heap $H$,
  written $sep(H, o, o')$, iff

  $\forall q, q' \in dom(H).~ \\ reach(H, o, q) \land reach(H, o', q')
  \implies q \neq q'$.
\end{definition}

\begin{definition}[Box Separation]\label{def:box-sep-frame1}
  For heap $H$ and frame $F$, $boxSep(H, F)$ holds iff

  $F = \aframe L t l \land \forall x \mapsto b(o), y \mapsto b(o') \in
  L.~ \\ \quad o \neq o' \implies sep(H, o, o')$
\end{definition}

\newpage
\begin{definition}[Box-Object Separation]\label{def:box-obj-sep1}
  For heap $H$ and frame $F$, $boxObjSep(H, F)$ holds iff

  $F = \aframe L t l \land \forall x \mapsto b(o), y \mapsto o' \in
  L.~sep(H, o, o')$
\end{definition}

\begin{definition}[Box Ocap Invariant]\label{def:box-ocap1}
  For heap $H$ and frame $F$, $boxOcap(H, F)$ holds iff

  $F = \aframe L t l \land \forall x \mapsto b(o) \in L, o' \in
  dom(H).~ \\ reach(H, o, o') \implies ocap(typeof(H, o'))$
\end{definition}
\noindent
In a well-formed frame, (a) two box references that are not aliases
are disjoint (Def.~\ref{def:box-sep-frame1}), (b) box references and
non-box references are disjoint (Def.~\ref{def:box-obj-sep1}), and (c)
all types reachable from box references are ocap
(Def.~\ref{def:box-ocap1}).

\begin{definition}[Global Ocap Separation]\label{def:global-ocap-sep1}
  For heap $H$ and frame $F$, $globalOcapSep(H, F)$ holds iff

  $F = \aframe L t l \land \forall x \mapsto o \in L, y \mapsto o' \in
  L_0.~ \\ ocap(typeof(H, o)) \land sep(H, o, o')$
\end{definition}
\noindent
In addition, in a well-formed frame that is well-typed under effect
\ocap, non-box references have ocap types, and they are disjoint from
the global variables in $L_0$ (Def.~\ref{def:global-ocap-sep1}).

\begin{figure}
  \centering
  \infrule[\textsc{F-ok}] {
    boxSep(H, F) \andalso boxObjSep(H, F) \\
    boxOcap(H, F) \\
    a = \texttt{ocap} \implies globalOcapSep(H, F)
  } {
    H ~;~ a \vdash F \ok
  }

  \vspace{0.3cm}

  \infrule[\textsc{SingFS-ok}] {
    H ~;~ a \vdash F \ok
  } {
    H ~;~ a \vdash F \circ \epsilon \ok
  }

  \vspace{0.3cm}

  \infrule[\textsc{FS-ok}] {
    H ~;~ b \vdash F^l \ok \andalso H ~;~ a \vdash FS \ok \\
    b = \begin{cases}
      \texttt{ocap} & \text{if } a = \texttt{ocap} \lor l = \epsilon \\
      \epsilon      & \text{otherwise}
    \end{cases} \\
    boxSeparation(H, F, FS) \\
    uniqueOpenBox(H, F, FS) \\
    openBoxPropagation(H, F^l, FS)
  } {
    H ~;~ b \vdash F^l \circ FS \ok
  }
  \caption{Separation invariants of frames and frame stacks.}
  \label{fig:f-ok1}
\end{figure}

The judgement $H ~;~ a \vdash F \ok$ combines these invariants as
shown in Figure~\ref{fig:f-ok1}; the corresponding judgement for frame
stacks uses the following additional invariants.

\begin{definition}[Box Separation]\label{def:box-separation1}
  For heap $H$, frame $F$, and frame stack $FS$,
  $boxSeparation(H, F, FS)$ holds iff

  $\forall o, o' \in dom(H).~boxRoot(o, F) \land boxRoot(o', FS) \land
  o \neq o' \implies sep(H, o, o')$
\end{definition}
\noindent
Def.~\ref{def:box-separation1} uses auxiliary predicate $boxRoot$
shown in Figure~\ref{fig:core1-aux-predicates}. $boxRoot(o, F)$ holds
iff there is a box reference to $o$ in frame $F$; $boxRoot(o, FS)$
holds iff there is a box reference to $o$ in one of the frames $FS$.
Informally, $boxSeparation(H, F, FS)$ holds iff non-aliased boxes are
disjoint.

\begin{definition}[Unique Open Box]\label{def:unique-open-box1}
  For heap $H$, frame $F$, and frame stack $FS$,
  $uniqueOpenBox(H, F, FS)$ holds iff

  $\forall o, o' \in dom(H).~openbox(H, o, F, FS) \land \\
    openbox(H, o', F, FS) \implies o = o'$
\end{definition}
\noindent
Def.~\ref{def:unique-open-box1} uses auxiliary predicate $openbox$
shown in Figure~\ref{fig:core1-aux-predicates}. $openbox(H, o, F, FS)$
holds iff $boxRoot(o, FS)$ and there is a local variable in frame $F$
which points to an object reachable from $o$ (box $o$ is ``open'' in
frame $F$). Informally, $uniqueOpenBox(H, F, FS)$ holds iff at most
one box is open (\ie accessible via non-box references) in frame $F$.

\begin{definition}[Open Box Propagation]\label{def:open-box-propag1}
  For heap $H$, frame $F^l$, and frame stack $FS$,
  $openBoxPropagation(H, F^l, FS)$ holds iff

  $l \neq \epsilon \land FS = G \circ GS \land openbox(H, o, F, FS) \implies
   openbox(H, o, G, GS)$
\end{definition}
\noindent
Informally, $openBoxPropagation(H, F^l, FS)$ holds iff frame $F^l$
preserves the open boxes in the top-most frame of frame stack $FS$.

According to rule \textsc{FS-ok} shown in Figure~\ref{fig:f-ok1},
well-formed frame stacks ensure (a) non-aliased boxes are disjoint
(Def.~\ref{def:box-separation1}), (b) at most one box is open (\ie
accessible via non-box references) per frame
(Def.~\ref{def:unique-open-box1}), and (c) method calls preserve open
boxes (Def.~\ref{def:open-box-propag1}).

\begin{figure}
  \centering
  \infrule {
    x \mapsto b(o) \in L
  } {
    boxRoot(o, \aframe{L}{t}{l})
  }
  \vspace{0.3cm}

  \infrule {
    boxRoot(o, F)
  } {
    boxRoot(o, F \circ \epsilon)
  }
  \vspace{0.3cm}

  \infrule {
    boxRoot(o, F) \lor boxRoot(o, FS)
  } {
    boxRoot(o, F \circ FS)
  }
  \vspace{0.3cm}

  \infrule {
    boxRoot(o, FS) \andalso x \mapsto o' \in L \andalso reach(H, o, o')
  } {
    openbox(H, o, \aframe{L}{t}{l}, FS)
  }
  \caption{Auxiliary predicates.}
  \label{fig:core1-aux-predicates}
\end{figure}

\subsection{Soundness and Heap Separation}\label{sec:core1-heap-sep}

Type soundness of \CLCONE~follows from the following preservation and
progress theorems. Instead of proving these theorems directly, we
prove corresponding theorems for an extended core language
(Section~\ref{sec:soundness}).

\begin{theorem}[Preservation]\label{thm:core1-preservation}
  If $\vdash H \typ \star$ then:
  \begin{enumerate}
  \item If $H \vdash F \typ \sigma$, $H ~;~ a \vdash F \ok$, and
    $\freduce H F {H'} {F'}$ then $\vdash H' \typ \star$, $H' \vdash
    F' \typ \sigma$, and $H' ~;~ a \vdash F' \ok$.

  \item If $H \vdash FS$, $H ~;~ a \vdash FS \ok$, and $\fsreduce H
    {FS} {H'} {FS'}$ then $\vdash H' \typ \star$, $H' \vdash FS'$, and
    $H' ~;~ b \vdash FS' \ok$.
  \end{enumerate}
\end{theorem}

\newpage
\begin{theorem}[Progress]\label{thm:core1-progress}
  If $\vdash H \typ \star$ then:

  If $H \vdash FS$ and $H ~;~ a \vdash FS \ok$ then either
  $\fsreduce H {FS} {H'} {FS'}$ or $FS = \aframe L x l \circ \epsilon$ or $FS = F \circ GS$
  where $F = \aframe L {\texttt{let}~x = t~\texttt{in}~t'} l$,
  $t \in \{ y.f, y.f = z, y.m(z), \\ y.\texttt{open}~\{ z \Rightarrow t'' \} \}$,
  and $L(y) = \texttt{null}$.

\end{theorem}

The following corollary expresses an essential heap separation
invariant enforced by \CLCONE. Informally, the corollary states
that objects ``within a box'' (reachable from a box reference) are
never mutated unless their box is ``open'' (a reference to the box
entry object is on the stack).

\begin{corollary}[Heap Separation]\label{coroll:core1-heap-sep}
  If $\vdash H \typ \star$ then:

  If $H \vdash FS$, $H ~;~ a \vdash FS \ok$,
  $\fsreduce{H}{FS}{H'}{FS'}$, \\ $FS = F \circ GS$, $F =
  \aframe{L}{\texttt{let}~x = y.f = z~\texttt{in}~t}{l}$, $L(y) = o'$,
  $boxRoot(o, FS)$, and $reach(H, o, o')$, then $w \mapsto o \in
  env(G)$ where $G \in FS$.
\end{corollary}

\begin{proof}[Proof sketch]
  First, by $H ~;~ a \vdash FS \ok$, \textsc{FS-ok}, and $FS = F \circ
  GS$ we have $H ~;~ a \vdash F \ok$. By \textsc{F-ok} we have
  $boxObjSep(H, F)$. By def.~\ref{def:box-obj-sep1} this means $u
  \mapsto b(o) \notin L$ and therefore $\lnot boxRoot(o, F)$. Given
  that $L(y) = o'$, $reach(H, o, o')$, and $boxRoot(o, FS)$ it must be
  that \\ $boxRoot(o, GS)$. Given that $boxRoot(o, GS)$ by
  def.~$boxRoot$ (Figure~\ref{fig:core1-aux-predicates}), there is a
  frame $G' \in GS$ such that $u \mapsto b(o) \in
  env(G')$. Well-formedness of $FS$ implies well-formedness of all its
  frames (\textsc{FS-ok}). Therefore, $G'$ is well-formed and by
  \textsc{F-ok}, $o$ is disjoint from other boxes
  (def.~\ref{def:box-sep-frame1}) and other objects
  (def.~\ref{def:box-obj-sep1}) reachable in $G'$, including the
  \texttt{global} variable. By the transition rules, box reference $u
  \mapsto b(o)$ prevents field selection; as a result, between frames
  $F$ and $G'$ there must be a frame created by opening $b(o)$. By
  \textsc{E-Open}, this means there is a frame $G \in FS$ such that $w
  \mapsto o \in env(G)$.
\end{proof}

\subsection{Lightweight Affinity}\label{sec:core2}


\comment{
Through the consumption of boxes, patterns such as merging two
boxes, ownership transfer of boxes, and unique fields are supported.}

This section introduces the \textsc{CoreLaCasa}$^{2}$
language \newline (\textsc{CLC}$^{2}$) which extends
\textsc{CLC}$^{1}$ with {\em affinity}, such that boxes may be
consumed at most once. Access to boxes is controlled using {\em
  permissions}. Permissions themselves are neither flow-sensitive nor
affine.  Consequently, they can be maintained in the type environment
$\Gamma$. Our notion of affinity is based on {\em continuation terms:}
consumption of permissions, and, thus, boxes, is only possible in
contexts where an explicit continuation is provided.  The
consumed permission is then no longer available in the continuation.

\begin{figure}[t]
  \centering
$\ba[t]{l@{\hspace{2mm}}l}
p    ::=  \seq{cd}~\seq{vd}~t                                       & \mbox{program}              \\
cd   ::=  \texttt{class}~C~\texttt{extends}~D~\{\highlight{\seq{fd}}~\seq{md}\} & \mbox{class}    \\
vd   ::=  \texttt{var}~f \typ C                                     & \mbox{variable} \\
ud   ::=  \texttt{var}~f \typ \highlight{\texttt{Box}[C]}           & \mbox{unique field} \\
fd   ::=  \highlight{vd ~|~ ud}                                     & \mbox{field}    \\
md   ::=  \texttt{def}~m(x \typ \sigma) \typ \highlight{C} = t      & \mbox{method}   \\
\sigma,\tau ::=                                                     & \mbox{surface type}         \\
\gap C, D                                                           & \gap\mbox{class type}       \\
\gap ~|~  \texttt{Box}[C]                                           & \gap\mbox{box type}         \\
\gap ~|~  \texttt{Null}                                             & \gap\mbox{null type}  \\
\pi ::=                                                             & \mbox{type}                 \\
\gap \sigma, \tau                                                   & \gap\mbox{surface type}     \\
\gap \highlight{Q \tri \texttt{Box}[C]}                             & \gap\mbox{guarded type} \\
\gap \highlight{\bot}                                               & \gap\mbox{bottom type} \\
\ea$
\caption{\textsc{CLC}$^{2}$ syntax. $C$, $D$ range over class names,
  $f$, $m$, $x$ range over term names. $Q$ ranges over abstract
  types.}
  \label{fig:core2-syntax}
\end{figure}

\paragraph{Syntax}

Figure~\ref{fig:core2-syntax} and Figure~\ref{fig:core2-syntax2} show the
syntactic differences between \textsc{CLC}$^{2}$ and
\textsc{CLC}$^{1}$: first, field types are either class types $C$ or
box types $\texttt{Box}[C]$; second, we introduce a bottom type $\bot$
and {\em guarded types} $\GuaT{Q}{C}$ where $Q$ ranges over a
countably infinite supply of abstract types; third, we introduce {\em
  continuation terms}.

In \textsc{CLC}$^{2}$ types are divided into {\em surface types} which
can occur in the surface syntax, and general {\em types}, including
guarded types, which cannot occur in the surface syntax; guarded types
are only introduced by type inference (see
Section~\ref{sec:static-semantics}). The bottom type $\bot$ is the
type of continuation terms $t^c$; these terms come in three forms:
\begin{enumerate}

\item A term $\texttt{box}[C]~\{ x \Rightarrow t \}$ creates a box
  containing a new instance of class type $C$, and makes that box
  accessible as $x$ in the continuation $t$. Note that in
  \textsc{CLC}$^{1}$ boxes are created using expressions of the form
  $\texttt{box}[C]$. In \textsc{CLC}$^{2}$ we require the continuation
  term $t$, because creating a box in addition creates a {\em
    permission} only available in $t$.

\item A term $\texttt{capture}(x.f, y)~\{ z \Rightarrow t \}$ merges
  two boxes $x$ and $y$ by assigning the value of $y$ to the field $f$
  of the value of $x$. In the continuation $t$ (a) $y$'s permission is
  no longer available, and (b) $z$ refers to box $x$.

\item A term $\texttt{swap}(x.f, y)~\{ z \Rightarrow t \}$ extracts
  the value of the unique field $x.f$ and makes it available as the
  value of a box $z$ in the continuation $t$; in addition, the value
  of box $y$ replaces the previous value of $x.f$. Finally, $y$'s
  permission is consumed.

\end{enumerate}

Given that only continuation terms can create boxes in \CLCTWO, method
invocations cannot return boxes unknown to the caller. As a result,
any box returned by a method invocation must have been passed as the
single argument in the invocation. However, a method that takes a box
as an argument, and returns the same box can be expressed using a
combination of \verb|open| and a method that takes the contents of the
box as an argument. Therefore, method return types are always class
types in \CLCTWO, simplifying the meta-theory.

\begin{figure}[t]
  \centering

$\ba[t]{l@{\hspace{2mm}}l}
t    ::=                                                   & \mbox{terms}                 \\
\gap     x                                                 & \gap\mbox{variable}          \\
\gap ~|~ \texttt{let}~x = e~\texttt{in}~t                  & \gap\mbox{let binding}       \\
\gap ~|~ \highlight{t^c}                                   & \gap\mbox{continuation term} \\
~ & ~ \\
e    ::=                                                 & \mbox{expressions}           \\
\gap     \texttt{null}                                   & \gap\mbox{null reference}    \\
\gap ~|~ x                                               & \gap\mbox{variable}          \\
\gap ~|~ x.f                                             & \gap\mbox{selection}         \\
\gap ~|~ x.f = y                                         & \gap\mbox{assignment}        \\
\gap ~|~ \texttt{new}~C                                  & \gap\mbox{instance creation} \\
\gap ~|~ x.m(y)                                          & \gap\mbox{invocation}        \\
\gap ~|~ x.\texttt{open}~\{ y \Rightarrow t \}           & \gap\mbox{open box}          \\
~ & ~ \\
t^c  ::=                                                   & \mbox{continuation term}     \\
\gap     \highlight{\texttt{box}[C]~\{ x \Rightarrow t \}} & \gap\mbox{box creation}      \\
\gap ~|~ \highlight{\texttt{capture}(x.f, y)~\{ z \Rightarrow t \}} & \gap\mbox{capture}           \\
\gap ~|~ \highlight{\texttt{swap}(x.f, y)~\{ z \Rightarrow t \}}    & \gap\mbox{swap}              \\
\ea$
  \caption{\textsc{CLC}$^{2}$ terms and expressions.}
  \label{fig:core2-syntax2}
\end{figure}

\subsubsection{Dynamic Semantics}

\CLCTWO~extends the dynamic semantics compared to \CLCONE \newline
with dynamically changing {\em permissions}. A dynamic access to a box
requires its associated permission to be available. For this, we
extend the reduction relations compared to \CLCONE~with {\em
  permission sets} $P$. Thus, a frame $\pframe L t P l$ combines a
variable environment $L$ and term $t$ with a set of permissions
$P$. (As before, the label $l$ is used for transferring return values
from method invocations.)

The transition rules of \CLCTWO~for single frames are identical to the
corresponding transition rules of \CLCONE; the permission sets do not
change.\footnote{Therefore, the transition rules (trivially) extended
  with permission sets are only shown in
  Appendix~\ref{app:core2-preservation-proof}.} In contrast, the
transition rules for {\em frame stacks} affect the permission sets of
frames.

The extended transition rules of \CLCTWO~are shown in
Figure~\ref{fig:core2-frame-stack-rules}. Rule \textsc{E-Invoke}
additionally requires permission $p$ to be available in $P$ in case
the argument of the invocation is a box protected by $p$; in this case
permission $p$ is also transferred to the new frame (the ``activation
record'').  Reduction gets stuck if permission $p$ is not
available. Rules \textsc{E-Return1} and \textsc{E-Return2} do not
affect permission sets and are otherwise identical to the
corresponding rules of \CLCONE. Rule \textsc{E-Open} requires that
permission $p$ of the box-to-open $b(o, p)$ is one of the currently
available permissions $P$. The permission set of the new frame is
empty. Rule \textsc{E-Box} creates a box $b(o, p)$ accessible in
continuation $t$ using fresh permission $p$. Note that rule
\textsc{E-Box} discards frame stack $FS$ in favor of the continuation
$t$.

\begin{figure}[t]
  \centering

\infrule[\textsc{E-Invoke}]
{ H(L(y)) = \obj C {FM} \andalso mbody(C, m) = x \rightarrow t' \\
  L' = L_0[\texttt{this} \mapsto L(y), x \mapsto L(z)] \\
  P' = \emptyset \lor (L(z) = b(o, p) \land p \in P \land P' = \{p\})
}
{ \fsreducebreak H {\pframe L {\texttt{let}~x = y.m(z)~\texttt{in}~t} P l \circ FS} H {\pframe {L'} {t'} {P'} x \circ \pframe L t P l \circ FS}
}

\vspace{0.3cm}

\infax[\textsc{E-Return1}]
{ \fsreducebreak H {\pframe L x P y \circ \pframe {L'} {t'} {P'} l \circ FS} H {\pframe {L'[y \mapsto L(x)]} {t'} {P'} l \circ FS}
}

\vspace{0.3cm}

\infax[\textsc{E-Return2}]
{ \fsreducebreak H {\pframe L x P {\epsilon} \circ \pframe {L'} {t'} {P'} l \circ FS} H {\pframe {L'} {t'} {P'} l \circ FS}
}

\vspace{0.3cm}

\infrule[\textsc{E-Open}]
{ L(y) = b(o, p) \andalso p \in P \andalso L' = [z \mapsto o]
}
{ \fsreducebreak H {\pframe L {\texttt{let}~x = y.\texttt{open}~\{ z \Rightarrow t' \}~\texttt{in}~t} P l \circ FS} H {\pframe {L'} {t'} {\emptyset} {\epsilon} \circ \pframe {L[x \mapsto L(y)]} t P l \circ FS}
}

\vspace{0.3cm}

\infrule[\textsc{E-Box}]
{ o \notin dom(H) \andalso fields(C) = \seq{f} \\
  H' = H[o \mapsto \obj C {\seq{f \mapsto \texttt{null}}}] \andalso p~\text{fresh} \\
}
{ \fsreducebreak H {\pframe L {\texttt{box}[C]~\{ x \Rightarrow t \}} P l \circ FS} {H'} {\pframe {L[x \mapsto b(o, p)]} t {P \cup \{p\}} {\epsilon} \circ \epsilon}
}

\comment{
\vspace{0.3cm}

\infrule[\textsc{E-Frame}]
{ \freduce H F {H'} {F'}
}
{ \fsreduce H {F \circ FS} {H'} {F' \circ FS}
}}

  \caption{\CLCTWO~frame stack transition rules.}
  \label{fig:core2-frame-stack-rules}
\end{figure}

Figure~\ref{fig:core2-frame-stack-rules-capture-swap} shows \CLCTWO's
two new transition rules. \textsc{E-Capture} merges box $b(o, p)$ and
box $b(o', p')$ by assigning $o'$ to field $f$ of object $H(o)$. The
semantics of \texttt{capture} is thus similar to that of a regular
field assignment. However, \texttt{capture} additionally requires both
permissions $p$ and $p'$ to be available; moreover, in continuation
$t$ permission $p'$ is no longer available, effectively {\em
  consuming} box $b(o', p')$. Like \textsc{E-Box}, \textsc{E-Capture}
discards frame stack $FS$.  Finally, \textsc{E-Swap} provides access
to a unique field $f$ of an object in box $b(o, p)$: in continuation
$t$ variable $z$ refers to the previous object $o''$ in $f$; the
object $o'$ in box $b(o', p')$ replaces $o''$. Like
\textsc{E-Capture}, \textsc{E-Swap} requires both permissions $p$ and
$p'$ to be available, and in continuation $t$ permission $p'$ is no
longer available, consuming box $b(o', p')$.

\begin{figure}
  \centering

\infrule[\textsc{E-Capture}]
{ L(x) = b(o, p) \andalso L(y) = b(o', p') \andalso \{ p, p' \} \subseteq P \\
  H(o) = \obj C {FM} \andalso H' = H[o \mapsto \obj C {FM[f \mapsto o']}] \\
}
{ \fsreducebreak H {\pframe L {\texttt{capture}(x.f, y)~\{ z \Rightarrow t \}} P l \circ FS} {H'} {\pframe {L[z \mapsto L(x)]} t {P \setminus \{ p' \}} {\epsilon} \circ \epsilon}
}

  \vspace{0.3cm}

  \infrule[\textsc{E-Swap}] {
    L(x) = b(o, p) \andalso L(y) = b(o', p') \andalso \{ p, p' \} \subseteq P \\
    H(o) = \obj C {FM} \andalso FM(f) = o'' \andalso p''~\text{fresh} \\
    H' = H[o \mapsto \obj C {FM[f \mapsto o']}] \\
  } {
    \fsreducebreak H {\pframe L {\texttt{swap}(x.f, y)~\{ z \Rightarrow t \}} P l \circ FS} {H'} {\pframe {L[z \mapsto b(o'', p'')]} t {(P \setminus \{ p' \}) \cup \{ p'' \}} {\epsilon} \circ \epsilon}
  }

  \caption{Transition rules for \texttt{capture} and \texttt{swap}.}
  \label{fig:core2-frame-stack-rules-capture-swap}
\end{figure}

\subsubsection{Static Semantics}\label{sec:static-semantics}

\begin{figure}
  \centering
  \infrule[\textsc{WF-Method1}] {
    \Gamma_0 , \texttt{this} \typ C, x \typ D ~;~ \epsilon \vdash t \typ E' \\
    E' \sub E
  } {
    C \vdash \texttt{def}~m(x \typ D) \typ E = t
  }
  \vspace{0.3cm}
  \infrule[\textsc{WF-Method2}] {
    \Gamma = \Gamma_0 , \texttt{this} \typ C, x \typ \highlight{\GuaT Q D, \Perm Q} \\
    Q~\text{fresh} \andalso \Gamma ~;~ \epsilon \vdash t \typ E' \andalso E' \sub E
  } {
    C \vdash \texttt{def}~m(x \typ \BoxT D) \typ E = t
  }
  \caption{Well-formed \CLCTWO~methods.}
  \label{fig:clc2-wf-programs}
\end{figure}

\paragraph{Well-Formed Programs.}

\CLCTWO~adapts method \\ well-formedness for the case where static
permissions are propagated to the callee context: the body of a method
with a parameter of type $\BoxT D$ is type-checked in an environment
$\Gamma$ which includes a static permission $\Perm Q$ where $Q$ is a
fresh abstract type; furthermore, the parameter has a {\em guarded
  type} $\GuaT Q D$. This environment $\Gamma$ ensures the method body
has full access to the argument box.

The \textsc{Ocap-*} rules for \CLCTWO~treat box-typed method
parameters analogously, and are left to the appendix.

\paragraph{Subclassing and Subtypes.}

In \CLCTWO, the subtyping relation $\sub$ is identical to that of
\CLCONE except for one additional rule for the $\bot$ type:

\setlength{\columnsep}{0pt}
\begin{multicols}{2}
  \infax[\textsc{$\sub$-Bot}]
        { \bot \sub \pi }
\end{multicols}
\noindent
The \textsc{$\sub$-Bot} rule says that $\bot$ is a subtype of any
type $\pi$. (Note that type \texttt{Null} is a subtype of any
{\em surface type}, whereas $\pi$ ranges over {\em all} types
including guarded types.)

\begin{figure}
  \centering
\infrule[\textsc{T-Invoke}] {
  \Gamma ~;~ a \vdash x \typ C \andalso mtype(C, m) = \sigma \rightarrow \tau \\
  \Gamma ~;~ a \vdash y \typ \sigma' \andalso \sigma' \sub \sigma~\lor \\
  \highlight{(\sigma = \BoxT D \land \sigma' = \GuaT Q D \land \Perm Q \in \Gamma)}
} {
  \Gamma ~;~ a \vdash x.m(y) \typ \tau
}

\comment{
\vspace{0.3cm}

\infrule[\textsc{T-Invoke2}] {
  \Gamma ~;~ a \vdash x \typ C \andalso \Gamma ~;~ a \vdash y \typ \highlight{\GuaT Q D} \\
  mtype(C, m) = \BoxT D \rightarrow E \andalso \highlight{\Perm Q \in \Gamma}
} {
  \Gamma ~;~ a \vdash x.m(y) \typ E
}}

\vspace{0.3cm}

\infrule[\textsc{T-New}] {
  a = \texttt{ocap} \implies ocap(C) \\
  \highlight{\forall~\texttt{var}~f \typ \sigma \in \seq{fd}.~\exists~D.~\sigma = D}
} {
  \Gamma ~;~ a \vdash \texttt{new}~C \typ C
}

\vspace{0.3cm}

\infrule[\textsc{T-Open}] {
  \Gamma ~;~ a \vdash x \typ \highlight{\GuaT Q C} \andalso \highlight{\Perm Q \in \Gamma} \\
  y \typ C ~;~ \texttt{ocap} \vdash t \typ \sigma
} {
  \Gamma ~;~ a \vdash x.\texttt{open}~\{ y \Rightarrow t \} \typ \highlight{\GuaT Q C}
}

\vspace{0.3cm}

\infrule[\textsc{T-Box}] {
  ocap(C) \andalso \highlight{Q~\text{fresh}} \\
  \Gamma, x \typ \highlight{\GuaT Q C}, \highlight{\Perm Q} ~;~ a \vdash t \typ \sigma
} {
  \Gamma ~;~ a \vdash \texttt{box}[C]~\{ x \Rightarrow t \} \typ \highlight{\bot}
}

\vspace{0.3cm}
  \caption{\CLCTWO~term and expression typing.}
  \label{fig:core2-type-rules}
\end{figure}

\paragraph{Term and Expression Typing.}

Figure~\ref{fig:core2-type-rules} shows the changes in the type
rules. In \textsc{T-Invoke}, if the method parameter has a type $\BoxT
D$ then the argument $y$ must have a type $\GuaT Q D$ such that the
static permission $\Perm Q$ is available in $\Gamma$. (Otherwise, box
$y$ has been consumed.) \textsc{T-New} checks that none of the field
types are box types. This makes sure classes with box-typed fields are
only created using \verb|box| expressions. (Box-typed fields are then
accessible using \verb|swap|.) \textsc{T-Open} requires the static
permission $\Perm Q$ corresponding to the guarded type $\GuaT Q C$ of
the opened box $x$ to be available in $\Gamma$; this ensures consumed
boxes are never opened. Finally, \textsc{T-Box} assigns a guarded type
$\GuaT Q C$ to the newly created box $x$ where $Q$ is a fresh abstract
type; the permission $\Perm Q$ is available in the type context of the
continuation term $t$.  The \verb|box| expression itself has type
$\bot$, since reduction never ``returns''; $t$ is the (only)
continuation.

\begin{figure}
  \centering
  \infrule[\textsc{T-Capture}] {
    \Gamma ~;~ a \vdash x \typ \GuaT Q C \andalso \Gamma ~;~ a \vdash y \typ \GuaT {Q'} D \\
    \{ \Perm Q, \Perm {Q'} \} \subseteq \Gamma \andalso D \sub ftype(C, f) \\
    \Gamma \setminus \{ \Perm {Q'} \}, z \typ \GuaT Q C ~;~ a \vdash t \typ \sigma
  } {
    \Gamma ~;~ a \vdash \texttt{capture}(x.f, y)~\{ z \Rightarrow t \} \typ \bot
  }
  \vspace{0.3cm}
  \infrule[\textsc{T-Swap}] {
    \Gamma ~;~ a \vdash x \typ \GuaT Q C \andalso \Gamma ~;~ a \vdash y \typ \GuaT {Q'} {D'} \\
    \{ \Perm Q, \Perm {Q'} \} \subseteq \Gamma \andalso ftype(C, f) = \texttt{Box}[D] \\
    D' \sub D \andalso R~\text{fresh} \\
    \Gamma \setminus \{ \Perm {Q'} \}, z \typ \GuaT R D, \Perm R ~;~ a \vdash t \typ \sigma
  } {
    \Gamma ~;~ a \vdash \texttt{swap}(x.f, y)~\{ z \Rightarrow t \} \typ \bot
  }
  \caption{Typing \CLCTWO's \texttt{capture} and \texttt{swap}.}
  \label{fig:core2-capture-swap}
\end{figure}

Figure~\ref{fig:core2-capture-swap} shows the type rules for \CLCTWO's
two new expressions. Both rules require $x$ and $y$ to have guarded
types such that the corresponding permissions are available in
$\Gamma$. In both cases the permission of $y$ is removed from the
environment used to type-check the continuation $t$; thus, box $y$ is
consumed in each case. In its continuation, \verb|capture| provides
access to box $x$ under alias $z$; thus, $z$'s type is equal to $x$'s
type. In contrast, \verb|swap| extracts the value of a {\em unique
  field} and provides access to it under alias $z$ in its
continuation. \CLCTWO~ensures the value extracted from the unique
field is externally unique. Therefore, the type of $z$ is a guarded
type $\GuaT R D$ where $R$ is fresh; permission $\Perm R$ is created
for use in continuation $t$.

\paragraph{Well-Formedness.}

\CLCTWO~extends \CLCONE~with unique fields of type $\texttt{Box}[C]$
(see Figure~\ref{fig:core2-syntax}); the following refined definition of
well-typed heaps in \CLCTWO~reflects this extension:

\begin{definition}[Well-typed Heap]\label{def:well-typed-heap2}
  A heap $H$ is well-typed, written $\vdash H \typ \star$ iff
  \begin{align*}
    \forall o \in dom(H).~ & H(o) = \obj{C}{FM} \implies \\
                           & (dom(FM) = fields(C)~ \land \\
                           & \forall f \in dom(FM).~FM(f) = \texttt{null}~ \lor \\
                           & (typeof(H, FM(f)) \sub D~ \land \\
                           & \quad ftype(C, f) \in \{D, \texttt{Box}[D]\})) \\
  \end{align*}
\end{definition}

The most interesting additions of \CLCTWO with respect to
well-formedness concern (a) separation invariants and (b) field
uniqueness. In \CLCONE, two boxes $x$ and $y$ are separate as long as
$x$ is not an alias of $y$. In \CLCTWO, the separation invariant is
more complex, because \verb|capture| merges two boxes, and \verb|swap|
replaces the value of a unique field. The key idea is to make {\em
  separation conditional on the availability of permissions}.

Box separation for frames in \CLCTWO~is defined as follows:
\begin{definition}[Box Separation]\label{def:box-sep-frame}
  For heap $H$ and frame $F$, $boxSep(H, F)$ holds iff

  $F = \pframe L t P l \land \forall x \mapsto b(o, p), y \mapsto
  b(o', p') \in L.~ p \neq p' \land \{ p, p' \} \subseteq P \implies
  sep(H, o, o')$
\end{definition}
\noindent
Two box references are disjoint if they are guarded by two different
permissions which are both available. As soon as a box is consumed,
\eg via \verb|capture|, box separation no longer holds, as expected.
In other invariants like $boxObjSep$, box permissions are not
required. Similarly, the differences in $boxOcap$ and $globalOcapSep$
are minor, and therefore left to the appendix.

\comment{
The main difference between box separation in \CLCONE~and in \CLCTWO~

\CLCTWO~extends frame and frame stack typing with permissions.

The
details are unsurprising and therefore left to the appendix.}

\CLCTWO's addition of unique fields requires a new {\em field
  uniqueness} invariant for well-formed frames:
\begin{definition}[Field Uniqueness]\label{def:field-uniqueness}
  For heap $H$ and frame $F = \pframe L t P l$, $fieldUniqueness(H,
  F)$ holds iff

  $\forall x \mapsto b(o, p) \in L, o', \hat{o} \in dom(H).~\\ p \in P
  ~\land~ reach(H, o, \hat{o}) ~\land~ H(\hat{o}) = \obj C {FM}
  ~\land~ \\ ftype(C, f) = \BoxT D ~\land~ reach(H, FM(f), o') \implies
  domedge(H, \hat{o}, f, o, o')$
\end{definition}
\noindent
This invariant expresses the fact that all reference paths from a box
$b(o, p)$ to an object $o'$ reachable from a unique field $f$ of
object $\hat{o}$ must ``go through'' that unique field. In other
words, in all reference paths from $o$ to $o'$, the edge $(\hat{o},
f)$ is a {\em dominating edge}. (A precise definition of $domedge$
appears in the appendix.)

\paragraph{Frame Stack Invariants.}

The frame stack invariants of \CLCTWO~are extended to take the
availability of permissions into account. For example, box separation
is now only preserved for boxes (a) that are not controlled by the
same permission, and (b) whose permissions are available:

\begin{definition}[Box Separation]\label{def:box-separation}
  Frame $F$ and frame stack $FS$ satisfy the box separation property
  in $H$, written \\ $boxSep(H, F, FS)$ iff

  $\forall o, o' \in dom(H).~boxRoot(o, F, p) \land \\ boxRoot(o', FS, p') \land p \neq p'
  \implies sep(H, o, o')$
\end{definition}
\noindent

Note that the availability of permissions is required indirectly by
the $boxRoot$ predicate (its other details are uninteresting, and
therefore omitted).

\comment{
\begin{definition}[Unique Open Box]\label{def:unique-open-box}
  Frame $F$ and frame stack $FS$ satisfy the unique open box property
  in $H$, written $uniqueOpenBox(H, F, FS)$ iff

  $\forall o, o' \in dom(H).~openbox(H, o, F, FS) \land \\ openbox(H, o', F, FS) \implies o = o'$
\end{definition}

\begin{definition}[Open Box Propagation]\label{def:open-box-propag}
  Frame $F^l$ and frame stack $FS$ satisfy the open box propagation property
  in $H$, written $openBoxPropagation(H, F^l, FS)$ iff

  $l \neq \epsilon \land FS = G \circ GS \land openbox(H, o, F, FS) \implies openbox(H, o, G, GS)$
\end{definition}
}

\section{Soundness}\label{sec:soundness}

\begin{theorem}[Preservation]\label{thm:core2-preservation}
  If $\vdash H \typ \star$ then:
  \begin{enumerate}

  \item If $H \vdash F \typ \sigma$, $H ~;~ a \vdash F \ok$, and
    $\freduce H F {H'} {F'}$ then $\vdash H' \typ \star$, $H' \vdash
    F' \typ \sigma$, and $H' ~;~ a \vdash F' \ok$.

  \item If $H \vdash FS$, $H ~;~ a \vdash FS \ok$, and $\fsreduce H
    {FS} {H'} {FS'}$ then $\vdash H' \typ \star$, $H' \vdash FS'$, and
    $H' ~;~ b \vdash FS' \ok$.

  \end{enumerate}
\end{theorem}

\begin{proof}

  Part (1) is proved by induction on the derivation of $\freduce H F
  {H'} {F'}$. Part (2) is proved by induction on the derivation of
  $\fsreduce H {FS} {H'} {FS'}$ and part (1). (See
  Appendix~\ref{app:core2-preservation-proof} for the full proof.)

\end{proof}

\begin{theorem}[Progress]\label{thm:core2-progress}
  If $\vdash H \typ \star$ then:

  If $H \vdash FS$ and $H ~;~ a \vdash FS \ok$ then either
  $\fsreduce H {FS} {H'} {FS'}$ or $FS = \pframe L x P l \circ \epsilon$ or $FS = F \circ GS$
  where
  \begin{itemize}
  \item $F = \pframe L {\texttt{let}~x = t~\texttt{in}~t'} P l$,
        $t \in \{ y.f, y.f = z, y.m(z), y.\texttt{open}~\{ z \Rightarrow t'' \} \}$,
        and $L(y) = \texttt{null}$; or
  \item $F = \pframe L {\texttt{capture}(x.f, y)~\{ z \Rightarrow t \}} P l$ where $L(x) = \texttt{null} \wedge L(y) = \texttt{null}$; or
  \item $F = \pframe L {\texttt{swap}(x.f, y)~\{ z \Rightarrow t \}} P l$ where $L(x) = \texttt{null} \wedge L(y) = \texttt{null}$.
  \end{itemize}

\end{theorem}

\begin{proof}

  By induction on the derivation of $H \vdash FS$. (See
  Appendix~\ref{app:core2-progress-proof} for the full proof.)

\end{proof}

Importantly, for a well-formed frame configuration, \\ \CLCTWO~ensures
that all required permissions are dynamically available; thus,
reduction is never stuck due to missing permissions.

\subsection{Isolation}

In order to state an essential isolation theorem, in the following we
extend \CLCTWO~with a simple form of message-passing concurrency.
This extension enables the statement of Theorem~\ref{thm:isolation}
which expresses the fact that the type system of \CLCTWO~enforces
data-race freedom in the presence of a shared heap and efficient,
by-reference message passing.

\begin{figure}
  \centering
$\ba[t]{l@{\hspace{2mm}}l}
\sigma,\tau ::=                                     & \mbox{surface type}         \\
\gap     \ldots                                     & ~    \\
\gap ~|~ \ProcT{C}                                  & \gap\mbox{process type}     \\
~ & ~ \\
e    ::=                                            & \mbox{expression}           \\
\gap     \ldots                                     & ~    \\
\gap ~|~ \texttt{proc}~\{ (x \typ \BoxT C) \Rightarrow t \}  & \gap\mbox{process creation}  \\
~ & ~ \\
t^c  ::=                                            & \mbox{continuation term}    \\
\gap     \ldots                                     & ~    \\
\gap ~|~ \texttt{send}(x, y)~\{ z \Rightarrow t \}  & \gap\mbox{message send}     \\
\ea$
  \caption{Syntax extensions for concurrency.}
  \label{fig:core3-syntax}
\end{figure}

Figure~\ref{fig:core3-syntax} summarizes the syntax extensions. An
expression of the form $\texttt{proc}~\{ (x \typ \BoxT C) \Rightarrow
t \}$ creates a concurrent process which applies the function $\{ (x
\typ \BoxT C) \Rightarrow t \}$ to each received message. A
continuation term of the form $\texttt{send}(x, y)~\{ z \Rightarrow t
\}$ asynchronously sends box $y$ to process $x$ and then applies the
continuation closure $\{ z \Rightarrow t \}$ to $x$.

\paragraph{Dynamic semantics.}

\CLCTHR~extends the dynamic semantics of \CLCTWO~such that the
configuration of a program consists of a shared heap $H$ and a set of
processes $\mathcal{P}$. Each process $FS^o$ is a frame stack $FS$
labelled with an object identifier $o$. A heap $H$ maps the object
identifier $o$ of a process $FS^o$ to a \emph{process record}
$\proc{\BoxT{C}}{M}{x \rightarrow t}$ where $\BoxT{C}$ is the type of
messages the process can receive, $M$ is a set of object identifiers
representing (buffered) incoming messages, and $x \rightarrow t$ is
the message handler function. \CLCTHR~introduces a third reduction
relation $\preduce{H}{\mathcal{P}}{H'}{\mathcal{P'}}$ which reduces a
set of processes $\mathcal{P}$ in heap $H$.

\begin{figure}
  \centering

  \infrule[\textsc{E-Proc}] {
    F = \pframe{L}{\texttt{let}~x = \texttt{proc}~\{ (y \typ \BoxT{C}) \Rightarrow
t' \}~\texttt{in}~t}{P}{l} \\
    F' = \pframe{L[x \mapsto o']}{t}{P}{l} \quad o'~\text{fresh} \\
    H' = H[o' \mapsto \proc{\BoxT{C}}{\emptyset}{y \rightarrow t'}]
  } {
    \preduce{H}{\{ (F \circ FS)^o \} \cup \mathcal{P}}{H'}{\{ (F' \circ FS)^o , \epsilon^{o'} \} \cup \mathcal{P}}
  }

  \vspace{0.3cm}

  \infrule[\textsc{E-Send}] {
    F = \pframe{L}{\texttt{send}(x, y)~\{ z \Rightarrow t \}}{P}{l} \quad L(x) = o \\
    H(o) = \proc{\BoxT{C}}{M}{f} \quad L(y) = b(o', p') \quad p' \in P \\
    F' = \pframe{L[z \mapsto o]}{t}{P \setminus \{p'\}}{\epsilon} \\
    H' = H[o \mapsto \proc{\BoxT{C}}{M \cup \{o'\}}{f}]
  } {
    \fsreduce{H}{F \circ FS}{H'}{F' \circ \epsilon}
  }

  \vspace{0.3cm}

  \infrule[\textsc{E-Receive}] {
    H(o) = \proc{\BoxT{C}}{M}{x \rightarrow t} \quad M = M' \uplus \{ o' \} \\
    F = \pframe{L}{y}{P}{l} \quad
    F' = \pframe{\emptyset[x \mapsto b(o', p)]}{t}{\{p\}}{\epsilon} \quad p~\text{fresh} \\
    H' = H[o \mapsto \proc{\BoxT{C}}{M'}{x \rightarrow t}]
  } {
    \preduce{H}{\{ (F \circ \epsilon)^o \} \cup \mathcal{P}}{H'}{\{ (F' \circ \epsilon)^o \} \cup \mathcal{P}}
  }
  \caption{\CLCTHR~process transition rules.}
  \label{fig:core3-proc-rules}
\end{figure}

Figure~\ref{fig:core3-proc-rules} shows the process transition
rules. Rule \textsc{E-Proc} creates a new process by allocating a
process record with an empty received message set and the message type
and handler function as specified in the \texttt{proc} expression. The
new process $\epsilon^{o'}$ starts out with an empty frame stack,
since it is initially idle. Rule \textsc{E-Send} sends the object
identifier in box $y$ to process $x$. The required permission $p'$ of
box $y$ is consumed in the resulting frame $F'$. The call stack is
discarded, since \texttt{send} is a continuation term. In rule
\textsc{E-Receive} process $o$ is ready to process a message from its
non-empty set of incoming messages $M$, since (the term in) frame $F$
cannot be reduced further and there are no other frames on the frame
stack. (The $\uplus$ operator denotes disjoint set union.) Frame $F'$
starts message processing with the parameter bound to a box reference
with a fresh permission.

\paragraph{Static semantics.}
Figure~\ref{fig:core3-wf-rules} shows the well-formedness rules that
\CLCTHR~introduces for (sets of) processes. A set of processes is
well-formed if each process is well-formed (\textsc{WF-Soup}). A
process is well-formed if its frame stack is well-formed
(\textsc{WF-Proc}).

\begin{figure}
  \centering
  \infrule[\textsc{WF-Soup}] {
    H \vdash FS^o \quad H \vdash \mathcal{P}
  } {
    H \vdash \{FS^o\} \cup \mathcal{P}
  }
  \vspace{0.3cm}
  \infrule[\textsc{WF-Proc}] {
    H \vdash FS \quad H ~;~ a \vdash FS \ok
  } {
    H \vdash FS^o
  }
  \caption{\CLCTHR~well-formedness rules.}
  \label{fig:core3-wf-rules}
\end{figure}

Figure~\ref{fig:core3-typing} shows the typing of process creation and
message sending. Rule \textsc{T-Proc} requires the body of a new
process to be well-typed in an environment that only contains the
parameter of the message handler and a matching access permission
which is fresh. Importantly, body term $t$ is type-checked under
effect \ocap. This means that $t$ may only instantiate ocap
classes. As a result, it is impossible to access global variables from
within the newly created process. Rule \textsc{T-Send} requires the
permission $\Perm{Q}$ of sent box $y$ to be available in context
$\Gamma$. The body $t$ of the continuation closure must be well-typed
in a context where $\Perm{Q}$ is no longer available. As with all
continuation terms, the type of a \texttt{send} term is $\bot$.

\begin{figure}
  \centering
  \infrule[\textsc{T-Proc}] {
    x \typ \GuaT Q C, \Perm Q ~;~ \texttt{ocap} \vdash t \typ \pi \\
    Q~\text{fresh}
  } {
    \Gamma ~;~ a \vdash \texttt{proc}~\{ (x \typ \BoxT C) \Rightarrow t \} \typ \ProcT{C}
  }
  \vspace{0.3cm}
  \infrule[\textsc{T-Send}] {
    \Gamma ~;~ a \vdash x \typ \ProcT{C} \andalso \Gamma ~;~ a \vdash y \typ \GuaT{Q}{D} \\
    \Perm{Q} \in \Gamma \andalso D \sub C \\
    \Gamma \setminus \{ \Perm{Q} \}, z \typ \ProcT{C} ~;~ a \vdash t \typ \pi
  } {
    \Gamma ~;~ a \vdash \texttt{send}(x, y)~\{ z \Rightarrow t \} \typ \bot
  }
  \caption{\CLCTHR~typing rules.}
  \label{fig:core3-typing}
\end{figure}

\begin{figure}
  \centering
  \infrule[\textsc{Acc-F}] {
    x \mapsto o \in L \lor (x \mapsto b(o, p) \in L \land p \in P)
  } {
    accRoot(o, \pframe{L}{t}{P}{l})
  }
  \vspace{0.3cm}
  \infrule[\textsc{Acc-FS}] {
    accRoot(o, F) \lor accRoot(o, FS)
  } {
    accRoot(o, F \circ FS)
  }
  \vspace{0.3cm}
  \infrule[\textsc{Iso-FS}] {
    \forall o, o' \in dom(H).~(accRoot(o, FS) \\ \land accRoot(o', FS')) \Rightarrow sep(H, o, o')
  } {
    isolated(H, FS, FS')
  }
  \vspace{0.3cm}
  \infrule[\textsc{Iso-Proc}] {
    H(o)  = \proc{\BoxT{C}}{M}{f} \\
    H(o') = \proc{\BoxT{D}}{M'}{g} \\
    \forall q \in M, q' \in M'.~sep(H, q, q') \\
    isolated(H, FS, GS)
  } {
    isolated(H, FS^o, GS^{o'})
  }
  \caption{\CLCTHR~process and frame stack isolation.}
  \label{fig:core3-isolation}
\end{figure}

Figure~\ref{fig:core3-isolation} defines a predicate $isolated$ to
express isolation of frame stacks and processes. Isolation of frame
stacks builds on an $accRoot$ predicate: identifier $o$ is an
accessible root in frame $F$, written $accRoot(o, F)$, iff $env(F)$
contains a binding $x \mapsto o$ or $x \mapsto b(o, p)$ where
permission $p$ is available in $F$ (\textsc{Acc-F}); $o$ is an
accessible root in frame stack $FS$ iff $accRoot(o, F)$ holds for any
frame $F \in FS$ (\textsc{Acc-FS}). Two frame stacks are then isolated
in $H$ iff all their accessible roots are disjoint in $H$
(\textsc{Iso-FS}). Finally, two processes are isolated iff their
message queues and frame stacks are disjoint (\textsc{Iso-Proc}).

\begin{theorem}[Isolation]\label{thm:isolation}
  If $\vdash H \typ \star$ then:

  If $H \vdash \mathcal{P}$, $\forall P, P' \in \mathcal{P}.~P \neq P' \implies isolated(H, P, P')$, and
    $\preduce{H}{\mathcal{P}}{H'}{\mathcal{P'}}$ then $\vdash H' \typ \star$, $H' \vdash
    \mathcal{P'}$, and \\ $\forall Q, Q' \in \mathcal{P'}.~Q \neq Q' \implies isolated(H', Q, Q')$.

\end{theorem}

Theorem~\ref{thm:isolation} states that $\leadsto$ preserves isolation
of well-typed processes. Isolation is preserved even when boxes are
transferred by reference between concurrent processes. Informally, the
validity of this statement rests on the preservation of
well-formedness of frames, frame stacks, and processes;
well-formedness guarantees the separation of boxes with available
permissions (def.~\ref{def:box-sep-frame} and
def.~\ref{def:box-separation}), and the separation of boxes with
available permissions from objects ``outside'' of boxes
(def.~\ref{def:box-obj-sep} and def.~\ref{def:global-ocap-sep} in
\ref{app:additional-rules}).

\comment{
\lacasa~supports both Scala 2.11, the current stable Scala release, as
well as Scala 2.10. The latter has been important for our empirical
evaluation (Section~\ref{sec:empirical-eval}), because a project in
our corpus, GeoTrellis, is still built for Scala 2.10.}

\section{Implementation}\label{sec:impl}

\lacasa~is implemented as a combination of a compiler plugin for the
current Scala 2 reference compiler and a runtime library.  The plugin
extends the compilation pipeline with an additional phase right after
regular type checking.  Its main tasks are (a) object-capability
checking, (b) checking the stack locality of boxes and permissions,
and (c) checking the constraints of \lacasa~expressions.  In turn, (c)
requires object-capability checking: type arguments of \verb|mkBox|
invocations must be object-capability safe, and \verb|open| bodies may
only instantiate object-capability safe classes. Certain important
constraints are implemented using spores~\cite{Miller14}.

\comment{
\item Our practical implementation not only checks methods for absence
  of accesses to global objects and instance creations of insecure
  classes. It also takes care of type parameters: instantiating a
  generic class with an insecure type argument is insecure in our
  system.}

\paragraph{Object-Capability Checking in Scala.}

Our empirical study revealed the importance of certain Scala-specific
``tweaks'' to conventional object-capability checking. We describe the
most important one. The Scala compiler generates so-called
``companion'' singleton objects for case classes and custom value
classes if the corresponding companions do not already exist. For a
case class such a synthetic companion object provides, \eg factory and
extractor~\cite{Emir07} methods. Synthetic companion objects are
object-capability safe.

\comment{
\begin{itemize}
\item The Scala compiler generates so-called ``companion'' (singleton)
  objects for certain kinds of classes, such as case classes and
  custom (or derived) value classes, in case the corresponding
  companions do not already exist in the program source. For a case
  class such a synthetic companion object provides, for example, a
  factory method, and an extractor~\cite{Emir07} for pattern matching.
  Synthetic companion objects are object-capability safe.

\item In Scala, type classes~\cite{Wadler89} are realized using a
  pattern~\cite{Oliveira10} which also uses singleton
  objects. Concretely, Scala's collections library uses a type class
  called \verb|CanBuildFrom| to enable uniform return types for
  combinators like map and filter~\cite{Odersky09}. The singleton
  objects used to implement instances of this type class are
  object-capability safe.
\end{itemize}}

\paragraph{Leveraging Spores.}

We leverage constraints supported by spores in several places in
\lacasa. We provide two examples where spores are used in our
implementation.

\comment{
:
\begin{lstlisting}[numbers=left, numberstyle=\scriptsize\color{gray}\ttfamily]
sealed trait OnlyNothing[T]
object OnlyNothing {
  implicit val onlyNothing: OnlyNothing[Nothing] =
    new OnlyNothing[Nothing] {}
}
\end{lstlisting}
\noindent
Since trait \verb|OnlyNothing[T]| is marked as \verb|sealed| it is
impossible to create subclasses outside \lacasa's runtime library. As
a result, there is only the single anonymous subclass in line
4.}

The first example is the body of an \verb|open| expression. According
to the rules of \lacasa, it is not allowed to have free variables (see
Section~\ref{sec:overview}). Using spores, this constraint can be
expressed in the type of the \verb|open| method as follows:
\begin{lstlisting}
def open(fun: Spore[T, Unit])(implicit
         acc: CanAccess { type C = self.C },
         noCapture: OnlyNothing[fun.Captured]): Unit
\end{lstlisting}
\noindent
Besides the implicit access permission, the method also takes an
implicit parameter of type \\ \verb|OnlyNothing[fun.Captured]|.  The
generic \verb|OnlyNothing| type is a trivial type class with only a
single instance, namely for type \verb|Nothing|, Scala's bottom type.
Consequently, for an invocation of \verb|open|, the compiler is
only able to resolve the implicit parameter \verb|noCapture| in the
case where type \verb|fun.Captured| is equal to \verb|Nothing|. In
turn, spores ensure this is only the case when the \verb|fun| spore
does not capture anything.

The second example where \lacasa~leverages spores is permission
consumption in \verb|swap|:
\begin{lstlisting}
def swap[S](select: T => Box[S])
           (assign: (T, Box[S]) => Unit, b: Box[S])
     (fun: Spore[Packed[S], Unit] {
             type Excluded = b.C
           })
     (implicit acc: CanAccess { type C = b.C }): Unit
\end{lstlisting}
\noindent
Besides the field accessor functions \verb|select| and \verb|assign|
(see Section~\ref{sec:overview}), \verb|swap| receives a box \verb|b|
to be put into the unique field. The implicit \verb|acc| parameter
ensures the availability of \verb|b|'s permission. Crucially,
\verb|b|'s permission is consumed by the assignment to the unique
field. Therefore, \verb|b| must not be accessed in the continuation
spore, which is expressed using the spore's \verb|Excluded| type
member. As a result, \verb|fun|'s body can no longer capture the permission.

\comment{
it is impossible to capture the permission within
the body of spore \verb|fun|.}

\subsection{Discarding the stack using exceptions}\label{sec:discarding}

Certain \lacasa~operations require discarding the stack of callers in
order to ensure consumed access permissions become unavailable. For
example, recall the message send shown in Figure~\ref{fig:actors-lacasa}
(see Section~\ref{sec:overview}):
\begin{lstlisting}
s.next.send(packed.box)({
  // continuation closure
})(access)
// unreachable
\end{lstlisting}
\noindent
Here, the \texttt{send} invocation consumes the \texttt{access}
permission: the permission is no longer available in the continuation.
This semantics is enforced by ensuring (a) the \texttt{access}
permission is unavailable within the explicit continuation closure
(line 2), and (b) code following the \texttt{send} invocation (line 4)
is unreachable. The former is enforced analogously to \texttt{swap}
discussed above. The latter is enforced by discarding the stack of
callers.

Discarding the call stack is a well-known technique in Scala, and has
been widely used in the context of event-based actors~\cite{Haller09}
where the stack of callers is discarded when an actor suspends with
just a continuation closure.\footnote{See
  \url{https://github.com/twitter-archive/kestrel/blob/3e64b28ad4e71256213e2bd6e8bd68a9978a2486/src/main/scala/net/lag/kestrel/KestrelHandler.scala}
  for a usage example in a large-scale production system.}  The
implementation consists of throwing an exception which unwinds the
call stack up to the actor's event-loop, or up to the boundary of a
concurrent task.

\begin{figure}
  \centering
\begin{lstlisting}
class MessageHandlerTask(
        receiver: Actor[T],
        packed: Packed[T]) extends Runnable {
  def run(): Unit = {
    // process message in `packed` object
    try {
      // invoke `receive` method of `receiver` actor
      receiver.receive(packed.box)(packed.access)
    } catch {
      case nrc: NoReturnControl => /* do nothing */
    }
    // check for next message
    ...
  }
  ...
}
\end{lstlisting}
  \caption{Handling \texttt{NoReturnControl} within actors.}
  \label{fig:noreturncontrol}
\end{figure}

Prior to throwing the stack-unwinding exception, operations like
\texttt{send} invoke their continuation closure which is provided
explicitly by the programmer:
\begin{lstlisting}
def send(msg: Box[T])(cont: NullarySpore[Unit] {...})
      (implicit acc: CanAccess {...}): Nothing = {
  ...     // enqueue message
  cont()  // invoke continuation closure
  throw new NoReturnControl  // discard stack
}
\end{lstlisting}
\noindent
The thrown \texttt{NoReturnControl} exception is caught either within
the main thread where the main method is wrapped in a
\texttt{try-catch} (see below), or within a worker thread of the actor
system's thread pool. In the latter case, the task that executes actor
code catches the \texttt{NoReturnControl} exception, as shown in
Figure~\ref{fig:noreturncontrol}. Note that the exception handler is at
the actor's ``top level:'' after processing the received message (in
\texttt{packed}) the \texttt{receiver} actor is ready to process the
next message (if any).

Scala's standard library provides a special
\\ \texttt{ControlThrowable} type for such cases where exceptions are
used to manage control flow.  The above \\ \texttt{NoReturnControl} type
extends \texttt{ControlThrowable}. The latter is defined as follows:
\begin{lstlisting}
trait ControlThrowable extends Throwable
                       with NoStackTrace
\end{lstlisting}
Mixing in the \texttt{NoStackTrace} trait disables the generation of
JVM stack traces, which is expensive but not required. The
\texttt{ControlThrowable} type enables exception handling without
disturbing exception-based control-flow transfers:
\begin{lstlisting}
try { ... } catch {
  case c: ControlThrowable => throw c // propagate
  case e: Exception => ...
}
\end{lstlisting}
Crucially, exceptions of a subtype of \texttt{ControlThrowable} are
propagated in order not to influence the in-progress control flow
transfer. Patterns such as the above are unchecked in Scala. However,
in the case of \lacasa, failure to propagate
\texttt{ControlThrowable}s could result in unsoundness. For example,
consider the following addition of a \texttt{try-catch} to the
previous example (shown at the beginning of
Section~\ref{sec:discarding}):
\begin{lstlisting}
try {
  s.next.send(packed.box)({
    // continuation closure
  })(access)
} catch {
  case c: ControlThrowable => // do nothing
}
other.send(packed.box)({
  ...
})(access)
\end{lstlisting}
By catching and not propagating the \texttt{ControlThrowable}, the
\texttt{access} permission remains accessible from line 8, enabling
sending the same object (\texttt{packed.box}) twice.

In order to prevent such soundness issues, the \lacasa~compiler plugin
checks \texttt{try-catch} expressions: a valid catch clause either (a)
does not match any \\ \texttt{ControlThrowable}, or (b) is preceeded by
a catch clause that matches and propagates \texttt{ControlThrowable}
exceptions. Furthermore, to support trusted \lacasa~code, a marker
method permits unsafe catches. For example, the main method of a
\lacasa~program is wrapped in the following trusted
\texttt{try-catch}:
\begin{lstlisting}
try { /* main method body */ } catch {
  case nrc: NoReturnControl => uncheckedCatchControl
}
\end{lstlisting}

\section{Empirical Evaluation}\label{sec:empirical-eval}

\begin{figure*}[t!]
  \begin{tabular}{ l | l | r | l | r | r | r }
    \hline
        {\bfseries ~Project~} & {\bfseries ~Version~} & {\bfseries ~SLOC~} & {\bfseries ~GitHub stats~} & {\bfseries ~\#classes/traits~} & {\bfseries ~\#ocap (\%)~} & {\bfseries ~\#dir. insec. (\%)~} \\
        \hline



Scala Standard Library & 2.11.7     & 33,107 & \stars 5,795 ~ \contribs 257 & 1,505 & 644 (43\%) & 212/861 (25\%) \\
Signal/Collect         & 8.0.6      & 10,159 & \stars 123   ~ \contribs 11  & 236   & 159 (67\%) & 60/77 (78\%) \\
GeoTrellis             & 0.10.0-RC2 & ~      & \stars 400   ~ \contribs 38  & ~     & ~  & ~ \\
-engine                & ~          & 3,868  & ~                            & 190   & 40 (21\%) & 124/150 (83\%) \\
-raster                & ~          & 22,291 & ~                            & 670   & 233 (35\%) & 325/437 (74\%) \\
-spark                 & ~          & 9,192  & ~                            & 326   & 101 (31\%) & 167/225 (74\%) \\

        \hline
\textbf{Total} & ~ & 78,617 & ~ & 2,927 & 1,177 (40\%) & 888/1,750 (51\%)~ \\

\hline
  \end{tabular}

  \caption{Evaluating the object-capability discipline in real Scala
    projects. Each project is an active open-source project hosted on
    GitHub. \protect\stars~represents the number of ``stars'' (or
    interest) a repository has on GitHub, and
    \protect\contribs~represents the number of contributors of the
    project.}

  \label{fig:ocap-stats}
\end{figure*}

The presented approach to object isolation and uniqueness is based on
object capabilities. Isolation is enforced only for instances
of ocap classes, \ie classes adhering to the object-capability
discipline. Likewise, ownership transfer is supported only for
instances of ocap classes. Therefore, it is important to know whether
the object-capability discipline imposes an undue burden
on developers; or whether, on the contrary, developers tend to design
classes and traits in a way that naturally follows the
object-capability discipline. Specifically, our empirical evaluation
aims to answer the following question: How many classes/traits in
medium to large open-source Scala projects already satisfy the
object-capability constraints required by \lacasa?

\comment{
\begin{itemize}
\item RQ1: How many classes/traits in medium to large open-source
  Scala projects already satisfy the object-capability constraints
  required by \lacasa?
\item RQ2: How important is it in practice to combine immutability and
  object-capability constraints?
\end{itemize}}

\comment{
  \footnote{By source lines of code we mean
  lines of code excluding blank lines and comments. The reported
  numbers were obtained using \texttt{cloc}~\cite{Cloc} v1.60.}
}

\paragraph{Methodology}

For our empirical analysis we selected Scala's standard library, a
large and widely-used class library, as well as two medium to large
open-source Scala applications. In total, our corpus comprises 78,617
source lines of code (obtained using~\cite{Cloc}). Determining the
prevalence of ocap classes and traits is especially important in the
case of Scala's standard library, since it tells us for which
classes/traits \lacasa~ supports isolation and ownership transfer
``out of the box,'' \ie without code changes. (We will
refer to both classes and traits as ``classes'' in the
following.)

The two open-source applications are Signal/Collect (S/C) and
GeoTrellis.  S/C~\cite{Stutz10} is a distributed graph processing
framework with applications in machine learning and the semantic web,
among others. Concurrency and distribution are implemented using the
Akka actor framework~\cite{Akka}.  Consequently, S/C could also
benefit from \lacasa's additional safety.  GeoTrellis is a high
performance data processing engine for geographic data, used by the
City of Asheville~\cite{PriorityPlaces} (NC, USA) and the U.S. Army,
among others. Like S/C, GeoTrellis utilizes actor concurrency through
Akka.

\comment{; it also integrates Spark~\cite{Zaharia12}.}

In each case we performed a clean build with the \lacasa~ compiler
plugin enabled. We configured the plugin to check ocap constraints for
{\em all compiled classes}. In addition, we collected statistics on
classes that {\em directly} violate ocap constraints through accesses
to global singleton objects.

\comment{
To determine the
importance of combining ocap constraints with immutability, we
re-analyzed the source of S/C while utilizing information
about singleton objects in Scala's standard library that are immutable
and never access shared global state; accesses to these singletons
were then allowed from within ocap classes.}

\paragraph{Results}

Figure~\ref{fig:ocap-stats} shows the collected statistics.

For Scala's standard library we found that 43\% of all classes follow
the object-capability discipline. While this number might seem low, it
is important to note that a {\em strict} form of ocap checking was
used: accesses to top-level singleton objects were disallowed, even if
these singletons were themselves immutable and object-capability
safe. Thus, classes directly using helper singletons were marked as
insecure. Interestingly, only 25\% of the insecure classes directly
access top-level singleton objects. This means, the majority of
insecure classes is insecure due to dependencies on other insecure
classes. These results can be explained as follows. First, helper
singletons (in particular, ``companion objects'') play an
important role in the architecture of Scala's collections
package~\cite{Odersky09}. In turn, with 22,958 SLOC the collections
package is by far the library's largest package,
accounting for 69\% of its total size.  Second, due to the high degree
of reuse enabled by techniques such as the type class
pattern~\cite{Oliveira10}, even a relatively small number of classes
that directly depend on singletons leads to an overall 57\% of
insecure classes.

In S/C 67\% of all classes satisfy strict ocap constraints,
a significantly higher percentage than for the Scala library. At
the same time, the percentage of classes that are not ocap due to
direct accesses to top-level singletons is also much higher (78\%
compared to 25\%). This means there is less reuse of insecure classes
in S/C. All analyzed components of GeoTrellis have a
similarly high percentage of ``directly insecure''
classes. Interestingly, even with its reliance on ``companion
objects'' and its high degree of reuse, the proportion of ocap classes
in the standard library is significantly higher compared to
GeoTrellis where it ranges between
21\% and 35\%.

\paragraph{Immutability and object capabilities}

Many singleton objects in Scala's standard library (a) are deeply
immutable, (b) only create instances of ocap classes, and (c) never
access global state. Such singletons are safe to access from within
ocap classes.\footnote{Analogous rules have been used for static
  fields in Joe-E~\cite{Mettler10}, an object-capability secure subset
  of Java.} To measure the impact of such singletons on the proportion
of ocap classes, we reanalyzed S/C with knowledge of safe
singletons in the standard library. As a result, the percentage of
ocap classes increased from 67\% to 79\%, while the proportion of
directly insecure classes remained identical.  Thus, knowledge of
``safe singleton objects'' is indeed important for object capabilities
in Scala.

\comment{
Signal/Collect~\cite{Stutz10} is ``a framework for scalable graph
computing.''\footnote{See
  \href{http://uzh.github.io/signal-collect/}{\texttt{http://uzh.github.io/signal-collect/}}}
GeoTrellis is ``a geographic data processing engine for high
performance applications.''\footnote{See
  \href{http://geotrellis.io/}{\texttt{http://geotrellis.io/}}}
}

\section{Other Related Work}\label{sec:other-related}

A number of previous approaches leverages permissions or capabilities
for uniqueness or related notions. Approaches limited to tree-shaped
object structures for unique references
include~\cite{wadler90,Odersky92,Westbrook12,Boyland03,Caires13,Pottier13}. In
contrast, \lacasa~provides external uniqueness~\cite{Clarke03}, which
allows internally-aliased object graphs. Permissions in
\lacasa~indicate which objects (``boxes'') are accessible in the
current scope. In contrast, the deny capabilities of the Pony
language~\cite{Clebsch15} indicate which operations are \emph{denied}
on aliases to the same object. By distinguishing read/write as well as
(actor-)local and global aliases, Pony derives a fine-grained matrix
of reference capabilities, which are more expressive than the
presented system. While Pony is a new language design,
\lacasa~integrates affine references into an existing language, while
minimizing the effort for reusing existing classes.

The notion of uniqueness provided by our system is similar to
UTT~\cite{Muller07}, an extension of Universe types~\cite{DietlM05}
with ownership transfer. Overall, UTT is more flexible, whereas
\lacasa~requires fewer annotations for reusing existing code; it also
integrates with Scala's local type inference. Active
ownership~\cite{Clarke08} shares our goal of providing a minimal type
system extension, however it requires owner-polymorphic methods and
existential owners whose integration with local type inference is not
clear. A more general overview of ownership-based aliasing control is
provided in \cite{Clarke13}. There is a long line of work on unique
object references~\cite{Hogg91,Baker95,Minsky96,Almeida97,Boyland01}
which are more restrictive than external uniqueness; a recurring theme
is the interaction between unique, immutable, and read-only
references, which is also exploited in a variant of C\# for systems
programming~\cite{Gordon12}. Several systems combine ownership with
concurrency control to prevent data races. RaceFree
Java~\cite{Abadi06} associates fields with locks, and an effect system
ensures correct lock acquisition. Boyapati et al.~\cite{Boyapati02}
and Zhao~\cite{Zhao} extend type system guarantees to deadlock
prevention.

Our system takes important inspiration from
Loci~\cite{WrigstadPMZV09}, a type system for enforcing thread
locality which requires very few source annotations. However,
\lacasa~supports ownership transfer, which is outside the domain of
Loci. Kilim~\cite{Srinivasan08} combines type qualifiers with an
intra-procedural shape analysis to ensure isolation of Java-based
actors. To simplify the alias analysis and annotation system, messages
must be tree-shaped. Messages in \lacasa~are not restricted to trees;
moreover, \lacasa~uses a type-based approach rather than static
analysis. StreamFlex~\cite{SpringPGV07} and
FlexoTasks~\cite{AuerbachBGSV08} are implicit ownership systems for
stream-based programming; like \lacasa, they allow reusing classes
which pass certain sanity checks, but the systems are more restrictive
than external uniqueness.

\section{Conclusion}\label{sec:conclusion}

This paper presents a new approach to integrating isolation and
uniqueness into an existing full-featured language. A key novelty of
the system is its minimization of annotations necessary for reusing
existing code. Only a single bit of information per class is required
to determine its reusability. Interestingly, this information is
provided by the object capability model, a proven methodology for
applications in security, such as secure sandboxing.  We present a
complete formal account of our system, including proofs of key
soundness theorems. We implement the system for the full Scala
language, and evaluate the object capability model on a corpus of over
75,000 LOC of popular open-source projects. Our results show that
between 21\% and 79\% of the classes of a project adhere to a strict
object capability discipline. In summary, we believe our approach has the potential to make a
flexible form of uniqueness practical on a large scale and in existing
languages with rich type systems.

\bibliographystyle{abbrvnat}





\bibliography{bib}

\appendix

\include{lacasa_core2_proofs_corrected_no}

\include{lacasa_core2_progress_proof_revised}

\include{lacasa_mechanized}

\end{document}

%% file: lacasa_core2_proofs_corrected_no.tex
\section{Full Proofs}

\begin{figure}
  \centering

  \infrule[\textsc{Ocap-Class}] {
    p \vdash \texttt{class}~C~\texttt{extends}~D~\{\seq{fd}~\seq{md}\} \\
    C \vdash_{ocap} \seq{md} \andalso ocap(D) \\
    \forall~\texttt{var}~f \typ \sigma \in \seq{fd}.~ocap(\sigma) \lor \sigma = \texttt{Box}[E] \land ocap(E)
  } {
    ocap(C)
  }

  \vspace{0.3cm}

  \infrule[\textsc{Ocap-Method1}] {
    \texttt{this} \typ C, x \typ D ~;~ \texttt{ocap} \vdash t \typ E' \andalso E' \sub E
  } {
    C \vdash_{ocap} \texttt{def}~m(x \typ D) \typ E = t
  }

  \vspace{0.3cm}

  \infrule[\textsc{Ocap-Method2}] {
    \Gamma = \texttt{this} \typ C, x \typ \GuaT Q D, \Perm Q \andalso Q~\text{fresh} \\
    \Gamma ~;~ \texttt{ocap} \vdash t \typ E' \andalso E' \sub E
  } {
    C \vdash_{ocap} \texttt{def}~m(x \typ \BoxT D) \typ E = t
  }

  \caption{Well-formed ocap classes.}
  \label{fig:clc2-wf-ocap-classes}
\end{figure}

\begin{figure}
  \centering

  \infrule[\textsc{WF-Perm}] {
    \gamma : permTypes(\Gamma) \longrightarrow P \text{ injective } \\
    \forall x \in dom(\Gamma), \\
    \Gamma(x) = \GuaT Q C \wedge L(x) = b(o, p) \wedge \Perm Q \in \Gamma 
    \Longrightarrow \\
    \gamma(Q) = p
  } {
    \vdash \Gamma ; L ; P
  }

  \vspace{0.3cm}

  \infrule[\textsc{WF-Var}] {
    L(x) = \texttt{null}~ \lor \\
    L(x) = o \land typeof(H, o) \sub \Gamma(x)~ \lor \\
    L(x) = \highlight{b(o, p)} \land \Gamma(x) = \highlight{\GuaT Q C} \land \\ typeof(H, o) \sub C
  } {
    H \vdash \Gamma ; L ; x
  }

  \vspace{0.3cm}

  \infrule[\textsc{T-Frame1}] {
    \Gamma ~;~ a \vdash t \typ \sigma \andalso \highlight{l \neq \epsilon \implies \sigma \sub C} \\
    H \vdash \Gamma ; L \andalso H \vdash \Gamma ; L ; P
  } {
    H \vdash \highlight{\pframe L t P l} \typ \sigma
  }

  \vspace{0.3cm}

  \infrule[\textsc{T-Frame2}] {
    \Gamma, x \typ \tau ~;~ a \vdash t \typ \sigma \andalso \highlight{l \neq \epsilon \implies \sigma \sub C} \\
    H \vdash \Gamma ; L \andalso H \vdash \Gamma ; L ; P
  } {
    H \vdash^{\tau}_x \highlight{\pframe L t P l} \typ \sigma
  }

  \vspace{0.3cm}

  \infrule[\textsc{F-ok}] {
    boxSep(H, F) \andalso boxObjSep(H, F) \\
    boxOcap(H, F) \\
    a = \texttt{ocap} \implies globalOcapSep(H, F) \\
    fieldUniqueness(H, F)
  } {
    H ~;~ a \vdash F \ok
  }

  \caption{Frame and frame stack typing.}
  \label{fig:core2-frame-stack-typing}
\end{figure}

\begin{figure}
  \centering
  \infrule {
    x \mapsto b(o, p) \in L \andalso p \in P
  } {
    boxRoot(o, \pframe L t P l)
  }
  \vspace{0.3cm}

  \infrule {
    boxRoot(o, F)
  } {
    boxRoot(o, F \circ \epsilon)
  }
  \vspace{0.3cm}

  \infrule {
    boxRoot(o, F) \lor boxRoot(o, FS)
  } {
    boxRoot(o, F \circ FS)
  }
  \vspace{0.3cm}

  \infrule {
    x \mapsto b(o, p) \in L \andalso p \in P
  } {
    boxRoot(o, \pframe L t P l, p)
  }
  \vspace{0.3cm}

  \infrule {
    boxRoot(o, F, p)
  } {
    boxRoot(o, F \circ \epsilon, p)
  }
  \vspace{0.3cm}

  \infrule {
    boxRoot(o, F, p) \lor boxRoot(o, FS, p)
  } {
    boxRoot(o, F \circ FS, p)
  }
  \vspace{0.3cm}

  \infrule {
    boxRoot(o, FS) \andalso x \mapsto o' \in env(F) \andalso reach(H, o, o')
  } {
    openbox(H, o, F, FS)
  }
  \caption{Auxiliary predicates.}
  \label{fig:core2-aux-predicates}
\end{figure}

\subsection{Additional Rules and Definitions}\label{app:additional-rules}

Figure~\ref{fig:clc2-wf-ocap-classes} shows \CLCTWO's \textsc{Ocap-*}
rules. Figure~\ref{fig:core2-frame-stack-typing} shows the updated
rules for frame and frame stack typing in
\CLCTWO. Figure~\ref{fig:core2-aux-predicates} shows the auxiliary
predicates $boxRoot$ and $openbox$.

\begin{definition}[Box-Object Separation]\label{def:box-obj-sep}
  Frame $F$ satisfies the {\em box-object separation} invariant in
  $H$, written \\ $boxObjSep(H, F)$, iff

  $F = \pframe L t P l \land \forall x \mapsto b(o, p), y \mapsto o'
  \in L.~sep(H, o, o')$
\end{definition}

\begin{definition}[Box Ocap Invariant]\label{def:box-ocap}
  Frame $F$ satisfies the {\em box ocap} invariant in $H$, written
  $boxOcap(H, F)$, iff

  $F = \pframe L t P l \land \forall x \mapsto b(o, p) \in L, o' \in
  dom(H).~p \in P \land reach(H, o, o') \implies ocap(typeof(H, o'))$
\end{definition}

\begin{definition}[Global Ocap Separation]\label{def:global-ocap-sep}
  Frame $F$ satisfies the {\em global ocap separation} invariant in
  $H$, written $globalOcapSep(H, F)$, iff

  $F = \pframe L t P l \land \forall x \mapsto o \in L, y \mapsto o'
  \in L_0.~ \\ ocap(typeof(H, o)) \land sep(H, o, o')$
\end{definition}

\begin{definition}[Dominating Edge]\label{def:dom-edge}
  Field $f$ of $\hat{o}$ is a dominating edge for paths from $o$ to
  $o'$ in $H$, written \\ $domedge(H, \hat{o}, f, o, o')$, iff

  $\forall P \in path(H, o, o').~P = o \ldots \hat{o}, FM(f) \ldots
  o'$ \\ where $H(\hat{o}) = \obj C {FM}$ and $f \in dom(FM)$.
\end{definition}

\begin{definition}
  \(permTypes(\Gamma)\) is the set of \emph{permissions} in a typing context \(\Gamma\),
  \[permTypes(\Gamma) = \{Q~|~\Perm Q \in \Gamma\}\]
\end{definition}

\subsection{Proof of Theorem~\ref{thm:core2-preservation}}\label{app:core2-preservation-proof}

\begin{figure}
  \centering

  \infax[\textsc{E-Null}] {
    \freducebreak H {\pframe L {\texttt{let}~x = \texttt{null}~\texttt{in}~t} P l} H {\pframe {L[x \mapsto \texttt{null}]} t P l}
  }

  \vspace{0.3cm}

  \infax[\textsc{E-Var}] {
    \freducebreak H {\pframe L {\texttt{let}~x = y~\texttt{in}~t} P l} H {\pframe {L[x \mapsto L(y)]} t P l}
  }

  \vspace{0.3cm}

  \infrule[\textsc{E-Select}] {
    H(L(y)) = \obj C {FM} \andalso f \in dom(FM)
  } {
    \freducebreak H {\pframe L {\texttt{let}~x = y.f~\texttt{in}~t} P l} H {\pframe {L[x \mapsto FM(f)]} t P l}
  }

  \vspace{0.3cm}

  \infrule[\textsc{E-Assign}] {
    L(y) = o \andalso H(o) = \obj C {FM} \\
    H' = H[o \mapsto \obj C {FM[f \mapsto L(z)]}] \\
  } {
    \freducebreak H {\pframe L {\texttt{let}~x = y.f = z~\texttt{in}~t} P l} {H'} {\pframe L {\texttt{let}~x = z~\texttt{in}~t} P l}
  }

  \vspace{0.3cm}

  \infrule[\textsc{E-New}] {
    o \notin dom(H) \andalso fields(C) = \seq{f} \\
    H' = H[o \mapsto \obj C {\seq{f \mapsto \texttt{null}}}] \\
  } {
    \freducebreak H {\pframe L {\texttt{let}~x = \texttt{new}~C~\texttt{in}~t} P l} {H'} {\pframe {L[x \mapsto o]} t P l}
  }

  \caption{\CLCTWO~frame transition rules.}
  \label{fig:core2-frame-rules}
\end{figure}

\begin{lemma}\label{lem:lemma1}
If $\vdash H \typ \star$ then:

If $H \vdash F \typ \sigma$, $H ~;~ a \vdash F \ok$, and $\freduce H F
{H'} {F'}$ then $\vdash H' \typ \star$, $H' \vdash F' \typ \sigma$,
and $H' ~;~ a \vdash F' \ok$.
\end{lemma}

\begin{proof}
By induction on the derivation of $\freduce H F {H'} {F'}$.

\begin{itemize}

\item[-] Case E-Null
  \begin{enumerate}
  \item By the assumptions
    \begin{enumerate}[label=(\alph*)]
    \item $\vdash H \typ \star$
    \item $H \vdash F \typ \sigma$
    \item $H ~;~ a \vdash F \ok$
    \item $\freduce H F H {F'}$
    \item $F = \pframe L {\texttt{let}~x = \texttt{null}~\texttt{in}~t} P l$
    \item $F' = \pframe {L[x \mapsto \texttt{null}]} t P l$
    \end{enumerate}
  \item By 1.b), 1.e), and T-Frame1
    \begin{enumerate}[label=(\alph*)]
    \item $\Gamma ~;~ b \vdash \texttt{let}~x = \texttt{null}~\texttt{in}~t \typ \sigma$
    \item $l \neq \epsilon \implies \sigma \sub C$
    \item $H \vdash \Gamma ; L$
    \item $\vdash \Gamma; L; P$
    \end{enumerate}
  \item By 2.a) and T-Let
    \begin{enumerate}[label=(\alph*)]
    \item $\Gamma ~;~ b \vdash \texttt{null} \typ \tau$
    \item $\Gamma , x \typ \tau ~;~ b \vdash t \typ \sigma$
    \end{enumerate}
  \item Define
    \begin{enumerate}[label=(\alph*)]
    \item $\Gamma' := \Gamma , x \typ \tau$
    \item $L' := L[x \mapsto \texttt{null}]$
    \end{enumerate}
  \item By 4.a-b) and WF-Var, $H \vdash \Gamma' ; L' ; x$
  \item By 2.c), 4.a-b), 5., and WF-Env, $H \vdash \Gamma' ; L'$

  \item By 4.a-b), 
    \begin{enumerate}[label=(\alph*)]
    \item $permTypes(\Gamma') = permTypes(\Gamma)$
    \item $\forall (x,o,p), L(x) = b(o, p) \Longleftrightarrow L'(x) = b(o, p)$
    \end{enumerate}

  \item By 2.d), 7.a-b), WF-Perm, \( \vdash \Gamma' ; L' ; P\)
    
  \item By 2.b), 3.b), 4.a-b), 6., 8., and T-Frame1, $H \vdash F' \typ \sigma$
  \item By 1.c), 1.e-f), and F-ok, $H ~;~ a \vdash F' \ok$
  \item 1.a), 7., and 9. conclude this case.
  \end{enumerate}

\item[-] Case E-Select
  \begin{enumerate}
  \item By the assumptions
    \begin{enumerate}[label=(\alph*)]
    \item $\vdash H \typ \star$
    \item $H \vdash F \typ \sigma$
    \item $H ~;~ a \vdash F \ok$
    \item $\freduce H F H {F'}$
    \end{enumerate}
  \item By 1.d) and E-Select
    \begin{enumerate}[label=(\alph*)]
    \item $F = \pframe L {\texttt{let}~x = y.f~\texttt{in}~t} P l$
    \item $H(L(y)) = \obj C {FM}$
    \item $f \in dom(FM)$
    \item $F' = \pframe {L[x \mapsto FM(f)]} t P l$
    \end{enumerate}
  \item By 1.b), 2.a), and T-Frame1
    \begin{enumerate}[label=(\alph*)]
    \item $\Gamma ~;~ b \vdash \texttt{let}~x = y.f~\texttt{in}~t \typ \sigma$
    \item $l \neq \epsilon \implies \sigma \sub E'$
    \item $H \vdash \Gamma ~;~ L$
    \item $\vdash \Gamma ; L ; P$
    \end{enumerate}
  \item By 3.a) and T-Let
    \begin{enumerate}[label=(\alph*)]
    \item $\Gamma ~;~ b \vdash y.f \typ \tau$
    \item $\Gamma , x \typ \tau ~;~ b \vdash t \typ \sigma$
    \end{enumerate}
  \item By 4.a) and T-Select
    \begin{enumerate}[label=(\alph*)]
    \item $\Gamma ~;~ b \vdash y \typ D$
    \item $ftype(D, f) = E$
    \item $\tau = E$
    \end{enumerate}
  \item By 5.a) and T-Var, $\Gamma(y) = D$
  \item By 3.c), 6., and WF-Env, $H \vdash \Gamma ; L ; y$
  \item By 2.b), 6., 7., and WF-Var, $C = typeof(H, L(y)) \sub D$
  \item By 8., $\sub$-Ext, and WF-Class, $ftype(C, f) = ftype(D, f)$
  \item By 1.a), 2.b-c), and def.~\ref{def:well-typed-heap}, $FM(f) = \texttt{null} \lor typeof(H, FM(f)) \sub ftype(C, f)$
  \item Define
    \begin{enumerate}[label=(\alph*)]
    \item $\Gamma' := \Gamma , x \typ E$
    \item $L' := L[x \mapsto FM(f)]$
    \end{enumerate}
  \item By 5.b), 9., 10., and 11.b), $L'(x) = \texttt{null} ~\lor~ \\ typeof(H, L'(x)) \sub E$
  \item By 11.a-b), 12., and WF-Var, $H \vdash \Gamma' ; L' ; x$
  \item By 3.c), 11.a-b), 13., and WF-Env, $H \vdash \Gamma' ; L'$
  \item By 4.b), 5.c), and 11.a), $\Gamma' ; b \vdash t \typ \sigma$
    
  \item By 10., 11.a-b),
    \begin{enumerate}[label=(\alph*)]
    \item $permTypes(\Gamma') = permTypes(\Gamma)$
    \item $\forall (x,o,p), L(x) = b(o, p) \Longleftrightarrow L'(x) = b(o, p)$
    \end{enumerate}

  \item By 3.d), 16.a-b), WF-Perm, \( \vdash \Gamma' ; L' ; P\)

  \item By 2.d), 3.b), 14., 15., 17 and T-Frame1, $H \vdash F' \typ \sigma$

  \item By 1.c), 2.a-d), and F-ok, $H ~;~ a \vdash F' \ok$

  \item 1.a), 18., and 19. conclude this case.
  \end{enumerate}

\item[-] Case E-Assign
  \begin{enumerate}
  \item By the assumptions
    \begin{enumerate}[label=(\alph*)]
    \item $\vdash H \typ \star$
    \item $H \vdash F \typ \sigma$
    \item $H ~;~ a \vdash F \ok$
    \item $\freduce H F {H'} {F'}$
    \end{enumerate}
  \item By 1.d) and E-Assign
    \begin{enumerate}[label=(\alph*)]
    \item $F = \pframe L {\texttt{let}~x = y.f = z~\texttt{in}~t} P l$
    \item $L(y) = o$
    \item $H(o) = \obj C {FM}$
    \item $H' = H[o \mapsto \obj C {FM[f \mapsto L(z)]}]$
    \item $F' = \pframe L {\texttt{let}~x = z~\texttt{in}~t} P l$
    \end{enumerate}
  \item By 1.b), 2.a), and T-Frame1
    \begin{enumerate}[label=(\alph*)]
    \item $\Gamma ~;~ b \vdash \texttt{let}~x = y.f = z~\texttt{in}~t \typ \sigma$
    \item $l \neq \epsilon \implies \sigma \sub E'$
    \item $H \vdash \Gamma ; L$
    \item $\vdash \Gamma ; L ; P$ 
    \end{enumerate}
  \item By 3.a) and T-Let
    \begin{enumerate}[label=(\alph*)]
    \item $\Gamma ~;~ b \vdash y.f = z \typ \tau$
    \item $\Gamma , x \typ \tau ~;~ b \vdash t \typ \sigma$
    \end{enumerate}
  \item By 4.a) and T-Assign
    \begin{enumerate}[label=(\alph*)]
    \item $\Gamma ~;~ b \vdash y \typ D$
    \item $\Gamma ~;~ b \vdash z \typ E$
    \item $E \sub ftype(D, f)$
    \item $\tau = E$
    \end{enumerate}
  \item By 4.b), 5.b), 5.d), and T-Let, $\Gamma ~;~ b \vdash \texttt{let}~x = z~\texttt{in}~t \typ \sigma$
  \item By 2.c-d), 3.c), and WF-Env, $H' \vdash \Gamma ; L$
  \item By 3.b), 3.d), 6., 7., and T-Frame1, $H' \vdash F' \typ \sigma$
  \item By 5.b) and T-Var, $\Gamma(z) = E$
  \item By 3.c), 9., and WF-Env, $H \vdash \Gamma ; L ; z$
  \item By 9., 10., and WF-Var, $L(z) = \texttt{null} \lor L(z) = o' \land typeof(H, o') \sub E$
  \item By 1.c) and F-ok, $\forall x \mapsto b(o, p), y \mapsto o' \in L.~sep(H, o, o')$
  \item By 11. and 12., $L(z) = \texttt{null} \lor \forall w \mapsto b(o'', p'') \in L.~sep(H, L(z), o'')$
  \item By 2.b) and 12., $\forall w \mapsto b(o'', p'') \in L.~sep(H, L(y), o'')$
  \item By 2.c-d), 13., and 14., $\forall w \mapsto b(o'', p'') \in L.~ \\ sep(H', L(y), o'')$
  \item By 2.c-d) and 13., $L(z) = \texttt{null} \lor \forall w \mapsto b(o'', p'') \in L.~sep(H', L(z), o'')$
  \item By 12., 15., and 16., $\forall w \mapsto b(o'', p''), \hat{w} \mapsto \hat{o} \in L.~sep(H', o'', \hat{o})$
  \item By 1.c) and F-ok, $a = \texttt{ocap} \implies \forall x \mapsto o \in L.~ocap(typeof(H, o))$
  \item By 2.c-d) and 18., $a = \texttt{ocap} \implies \forall x \mapsto o \in L.~ocap(typeof(H', o))$
  \item By 1.c) and F-ok, $a = \texttt{ocap} \implies \forall x \mapsto o \in L, y \mapsto o' \in L_0.~sep(H, o, o')$
  \item By 2.b) and 20., $a = \texttt{ocap} \implies \forall \hat{w} \mapsto \hat{o} \in L_0.~sep(H, o, \hat{o})$
  \item By 11. and 20., $a = \texttt{ocap} \implies L(z) = \texttt{null} \lor \forall \hat{w} \mapsto \hat{o} \in L_0.~sep(H, o', \hat{o})$
  \item By 2.b-d), 21., and 22.
    \begin{enumerate}[label=(\alph*)]
    \item $a = \texttt{ocap} \implies \forall \hat{w} \mapsto \hat{o} \in L_0.~sep(H', o, \hat{o})$
    \item $a = \texttt{ocap} \implies L(z) = \texttt{null} \lor \forall \hat{w} \mapsto \hat{o} \in L_0.~sep(H', o', \hat{o})$
    \end{enumerate}
  \item By 2.b-d), 20., and 23., $a = \texttt{ocap} \implies \forall x \mapsto o \in L, y \mapsto o' \in L_0.~sep(H', o, o')$
  \item By 1.c), 2.a-e), 17., 19., 24., and F-ok, $H' ~;~ a \vdash F' \ok$
  \item By 5.c), 11., and $\sub$-Trans, $L(z) = \texttt{null} ~\lor~ \\ typeof(H, L(z)) \sub ftype(D, f)$
  \item By 5.a) and T-Var, $\Gamma(y) = D$
  \item By 3.c), 27., and WF-Env, $H \vdash \Gamma ; L ; y$
  \item By 2.b), 27., 28., and WF-Var, $typeof(H, L(y)) \sub D$
  \item By 2.b-c), 29., and def.~\ref{def:well-typed-heap}, $C \sub D$
  \item By 30., $\sub$-Ext, and WF-Class, $ftype(C, f) = ftype(D, f)$
  \item By 26. and 31., $L(z) = \texttt{null} \lor typeof(H, L(z)) \sub ftype(C, f)$
  \item By 2.c-d) and 32., $L(z) = \texttt{null} \lor typeof(H', L(z)) \sub ftype(C, f)$
  \item By 1.a), 2.c-d), 33., and def.~\ref{def:well-typed-heap}, $\vdash H' \typ \star$







  \item 8., 25., and 34. conclude this case.
  \end{enumerate} 

\end{itemize}

\end{proof}

\begin{lemma}\label{lem:lemma2}
If $\vdash H \typ \star$ then:

If $H \vdash FS$, $H ~;~ a \vdash FS \ok$, and $\fsreduce H {FS} {H'}
{FS'}$ then $\vdash H' \typ \star$, $H' \vdash FS'$, and $H' ~;~ b
\vdash FS' \ok$.
\end{lemma}

\begin{proof}
By induction on the derivation of $\fsreduce H {FS} {H'} {FS'}$

\begin{itemize}

\item[-] Case E-Invoke.
  \begin{enumerate}
  \item By the assumptions
    \begin{enumerate}[label=(\alph*)]
    \item $\vdash H \typ \star$
    \item $H \vdash FS$
    \item $H ~;~ a \vdash FS \ok$
    \item $\fsreduce H {FS} H {FS'}$
    \item $FS = F \circ GS$
    \item $F = \pframe L {\texttt{let}~x = y.m(z)~\texttt{in}~t} P l$
    \end{enumerate}
  \item By 1.d-f) and E-Invoke
    \begin{enumerate}[label=(\alph*)]
    \item $FS' = G' \circ G \circ GS$
    \item $G' = \pframe {L'} {t'} {P'} x$
    \item $G = \pframe L t P l$
    \item $L' = L_0[\texttt{this} \mapsto L(y), x \mapsto L(z)]$
    \item $H(L(y)) = \obj C {FM}$
    \item $mbody(C, m) = x \rightarrow t'$
    \item $P' = \begin{cases} \{p\} & \text{if } L(z) = b(o, p) \\
      \emptyset & \text{otherwise} \\
    \end{cases}$
    \end{enumerate}
  \item By 1.b), 1.e), T-FS-A, and T-FS-NA
    \begin{enumerate}[label=(\alph*)]
    \item $H \vdash F \typ \sigma$
    \item $l = \epsilon \implies H \vdash GS \land l = y \implies H \vdash^{\sigma}_y GS$
    \end{enumerate}
  \item By 1.f), 3.a), and T-Frame1
    \begin{enumerate}[label=(\alph*)]
    \item $\Gamma ~;~ b \vdash \texttt{let}~x = y.m(z)~\texttt{in}~t \typ \sigma$
    \item $H \vdash \Gamma ; L$
    \item $l \neq \epsilon \implies \sigma \sub C'$
    \item $\vdash \Gamma ; L ; P$
    \end{enumerate}
  \item By 4.a) and T-Let
    \begin{enumerate}[label=(\alph*)]
    \item $\Gamma ~;~ b \vdash y.m(z) \typ \tau$
    \item $\Gamma , x \typ \tau ~;~ b \vdash t \typ \sigma$
    \end{enumerate}
  \item By 4.b-c), 5.b), and T-Frame2, $H \vdash^{\tau}_x \pframe L t P l \typ \sigma$
  \item By 3.b), 6., T-FS-A2, and T-FS-NA2, $H \vdash^{\tau}_x \pframe L t P l \circ GS$
  \item By 5.a) and T-Invoke
    \begin{enumerate}[label=(\alph*)]
    \item $\Gamma ~;~ b \vdash y \typ D$
    \item $mtype(D, m) = \tau' \rightarrow \tau$
    \item $\Gamma ~;~ b \vdash z \typ \sigma'$
    \item $\tau' = E \implies \sigma' = E'$ for some $E' \sub E$
    \item $\tau' = \BoxT E \implies \sigma' = \GuaT Q {E'} \land \Perm Q \in \Gamma$ for some $Q, E' \sub E$
    \end{enumerate}
  \item By 8.a) and T-Var, $\Gamma(y) = D$
  \item By 2.e), 4.b), 9., WF-Env, and WF-Var, $C \sub D$
  \item By 8.b), 10., WF-Class, and WF-Override, $mtype(C, m) = mtype(D, m) = \tau' \rightarrow \tau$
  \item Define $\Gamma' := \begin{cases}
    \texttt{this} \typ C , x \typ E \text{~if } \tau' = E \\
    \texttt{this} \typ C , x \typ \GuaT Q E , \Perm Q \text{~otherwise }, \\ Q~\text{fresh} \\
  \end{cases}$
  \item By 2.f), 11., 12., WF-Method1, and WF-Method2, $\Gamma' ~;~ a' \vdash t' \typ \tau$
  \item By 8.c) and T-Var, $\Gamma(z) = \sigma'$
  \item By 4.b), 14., and WF-Env, $H \vdash \Gamma ; L ; z$
  \item By 8.d-e), 14., 15., and WF-Var, $L(z) = \texttt{null} \lor \\
    (\tau' = E \land \Gamma(z) = E' \land L(z) = o \land typeof(H, o) \sub E'$ for some $E' \sub E) \lor \\
    (\tau' = \BoxT E \land \Gamma(z) = \GuaT Q {E'} \land \Perm Q \in \Gamma \land L(z) = b(o, p) \land typeof(H, o) \sub E'$ for some $E' \sub E)$
  \item By 2.d), 12., 16., and $\sub$-Trans, $L'(x) = \texttt{null} \lor \\
    (\tau' = E \land \Gamma'(x) = E \land L'(x) = o \land typeof(H, o) \sub E) \lor \\
    (\tau' = \BoxT E \land \Gamma'(x) = \GuaT Q E \land \Perm Q \in \Gamma' \land \\
    L'(x) = b(o, p) \land typeof(H, o) \sub E)$
  \item By 17. and WF-Var, $H \vdash \Gamma' ; L' ; x$
  \item By 4.b), 9., and WF-Env, $H \vdash \Gamma ; L ; y$
  \item By 9., 19., and WF-Var, $L(y) = \texttt{null} \lor \Gamma(y) = D \land L(y) = o' \land typeof(H, o') \sub D$
  \item By 2.e), 20., and $\sub$-Refl, $L(y) = \texttt{null} \lor \Gamma(y) = D \land L(y) = o' \land typeof(H, o') \sub C$
  \item By 2.d), 12., and 21., $L'(\this) = \texttt{null} \lor \Gamma'(\this) = C \land L'(\this) = o' \land typeof(H, o') \sub C$
  \item By 22. and WF-Var, $H \vdash \Gamma' ; L' ; \this$
  \item By 2.d) and 12., $dom(\Gamma') \subseteq dom(L')$
  \item By 18., 23., 24., and WF-Env, $H \vdash \Gamma' ; L'$
  \item By 4.d), 8.b), 13., 25., WF-Method1, WF-Method2, and T-Frame1, $H \vdash \pframe {L'} {t'} {P'} x \typ \tau$
  \item By 7., 26., and T-FS-A, $H \vdash FS'$
  \item By 1.c) and FS-ok
    \begin{enumerate}[label=(\alph*)]
    \item $H ~;~ a \vdash F \ok$
    \item $H ~;~ c \vdash GS \ok$
    \item $a = \begin{cases}\ocap & \text{if } c = \ocap \lor label(F) = \epsilon \\
      \epsilon & \text{otherwise} \\
    \end{cases}$
    \item $boxSeparation(H, F, GS)$
    \item $uniqueOpenBox(H, F, GS)$
    \item $openBoxPropagation(H, F, GS)$
    \end{enumerate}
  \item By 28.d) and def. $boxSeparation$, \\ $boxSeparation(H, G, GS)$
  \item By 2.a-e), 2.g), and 29., $boxSeparation(H, G', G \circ GS)$
  \item By 28.e) and 2.d), $uniqueOpenBox(H, G', G \circ GS)$
  \item By 2.d) and 28.f), $openBoxPropagation(H, G', G \circ GS)$
  \item By 1.c) and FS-ok, $H ~;~ a \vdash G \circ GS \ok$
  \item By 2.d), 2.g), and 28.a), $H ~;~ a \vdash G' \ok$
  \item By 30., 31., 32., 33., 34., and FS-ok, $H ~;~ a \vdash FS' \ok$
  \item 1.a), 27., and 35. conclude this case.
    
  \end{enumerate}

\item[-] Case E-Return1.
  \begin{enumerate}
  \item By the assumptions
    \begin{enumerate}[label=(\alph*)]
    \item $\vdash H \typ \star$
    \item $H \vdash FS$
    \item $H ~;~ a \vdash FS \ok$
    \item $\fsreduce H {FS} H {FS'}$
    \item $FS = F \circ F' \circ GS$
    \item $F = \pframe L x P y$
    \item $F' = \pframe {L'} {t'} {P'} l$
    \end{enumerate}
  \item By 1.d-g) and E-Return1
    \begin{enumerate}[label=(\alph*)]
    \item $FS' = G' \circ GS$
    \item $G' = \pframe {L'[y \mapsto L(x)]} {t'} {P'} l$
    \end{enumerate}
  \item By 1.b), 1.e-f), and T-FS-A
    \begin{enumerate}[label=(\alph*)]
    \item $H \vdash F \typ \sigma$
    \item $H \vdash^{\sigma}_y F' \circ GS$
    \end{enumerate}
  \item By 1.f), 3.a), and T-Frame1
    \begin{enumerate}[label=(\alph*)]
    \item $\Gamma ~;~ b \vdash x \typ \sigma$
    \item $H \vdash \Gamma ; L$
    \item $\sigma \sub C$
    \item $\vdash \Gamma ; L; P$
    \end{enumerate}
  \item By 4.a), 4.c), and T-Var, $\Gamma(x) = \sigma \sub C$
  \item By 4.b), 5., and WF-Env, $H \vdash \Gamma ; L ; x$
  \item By 1.g), 3.b), T-FS-A2, and T-FS-NA2
    \begin{enumerate}[label=(\alph*)]
    \item $H \vdash^{\sigma}_y F' \typ \tau$
    \item $l = \epsilon \implies H \vdash GS \land l = z \implies H \vdash^{\tau}_z GS$
    \end{enumerate}
  \item By 1.g), 7.a), and T-Frame2
    \begin{enumerate}[label=(\alph*)]
    \item $\Gamma' , y \typ \sigma ; b' \vdash t' \typ \tau$
    \item $H \vdash \Gamma' ; L'$
    \item $l \neq \epsilon \implies \tau \sub D$
    \item $\vdash \Gamma' ; L' ; P'$
    \end{enumerate}
  \item By 6., 8.b), and WF-Env, $H \vdash (\Gamma' , y \typ \sigma) ; L'[y \mapsto L(x)]$
  \item By 2.b), 8.a), 8.c), 8.d), 9., and T-Frame1, $H \vdash G' \typ \tau$
  \item By 7.b), 10., T-FS-A, and T-FS-NA, $H \vdash FS'$
  \item By 1.c), 1.e-f), and FS-ok
    \begin{enumerate}[label=(\alph*)]
    \item $H ~;~ a \vdash F \ok$
    \item $H ~;~ c \vdash F' \circ GS \ok$
    \item $a = \begin{cases} \ocap & \text{if } c = \ocap \\
      \epsilon & \text{otherwise} \\
    \end{cases}$
    \item $boxSeparation(H, F, F' \circ GS)$
    \item $uniqueOpenBox(H, F, F' \circ GS)$
    \item $openBoxPropagation(H, F, F' \circ GS)$
    \end{enumerate}
  \item By 12.b) and FS-ok
    \begin{enumerate}[label=(\alph*)]
    \item $H ~;~ c \vdash F' \ok$
    \item $H ~;~ d \vdash GS \ok$
    \item $c = \begin{cases} \ocap & \text{if } d = \ocap \lor l = \epsilon \\
      \epsilon & \text{otherwise} \\
    \end{cases}$
    \item $boxSeparation(H, F', GS)$
    \item $uniqueOpenBox(H, F', GS)$
    \item $openBoxPropagation(H, F', GS)$
    \end{enumerate}
  \item By 13.a) and F-ok, $L(x) = \texttt{null} \implies H ~;~ c \vdash G' \ok$
  \item By 5., 6., and WF-Var, $L(x) = \texttt{null} \lor L(x) = o \land typeof(H, o) \sub C$
  \item By 13.a) and F-ok, $\forall z \mapsto b(o', p') \in L'.~ \\ \lnot openbox(H, o', F', GS)$
  \item By 12.f) and 16., $\forall z \mapsto b(o', p') \in L'.~ \\ \lnot openbox(H, o', F, F' \circ GS)$
  \item By 17. and def. openbox, $L(x) = o \implies \forall z \mapsto b(o', p') \in L'.~sep(H, L(x), o')$
  \item By 12.a), 12.c), and F-ok, $c = \ocap \implies L(x) = o \implies ocap(typeof(H, L(x)))$
  \item By 1.g), 2.b), 13.a), and F-ok, $fieldUniqueness(H, G')$
  \item By 2.b), 13.a), 14., 15., 18., 19., 20., and F-ok, $H ~;~ c \vdash G' \ok$
  \item By 13.d), 15., and 18., $boxSeparation(H, G', GS)$
  \item By 13.e), 15., and 18., $uniqueOpenBox(H, G', GS)$
  \item By 13.f), 15., and 18., $openBoxPropagation(H, G', GS)$
  \item By 13.b), 13.c), 21., 22., 23., 24., and FS-ok, $H ~;~ c \vdash FS' \ok$
  \item 1.a), 11., and 25. conclude this case.
  \end{enumerate}

\item[-] Case E-Open.
  \begin{enumerate}
  \item By the assumptions
    \begin{enumerate}[label=(\alph*)]
    \item $\vdash H \typ \star$
    \item $H \vdash FS$
    \item $H ~;~ a \vdash FS \ok$
    \item $\fsreduce H {FS} H {FS'}$
    \item $FS = F \circ GS$
    \item $F = \pframe L {\texttt{let}~x = y.\texttt{open}~\{ z \Rightarrow t' \}~\texttt{in}~t} P l$
    \end{enumerate}
  \item By 1.d-f) and E-Open
    \begin{enumerate}[label=(\alph*)]
    \item $FS' = G' \circ G \circ GS$
    \item $G' = \pframe {L'} {t'} {\emptyset} {\epsilon}$
    \item $G = \pframe {L[x \mapsto L(y)]} t P l$
    \item $L(y) = b(o, p)$
    \item $p \in P$
    \item $L' = [z \mapsto o]$
    \end{enumerate}
  \item By 1.b), 1.e-f), T-FS-A, and T-FS-NA
    \begin{enumerate}[label=(\alph*)]
    \item $H \vdash F \typ \sigma$
    \item $l = \epsilon \implies H \vdash GS \land l = w \implies H \vdash^{\sigma}_w GS$
    \end{enumerate}
  \item By 1.f), 3.a), and T-Frame1
    \begin{enumerate}[label=(\alph*)]
    \item $\Gamma ~;~ b \vdash \texttt{let}~x = y.\texttt{open}~\{ z \Rightarrow t' \}~\texttt{in}~t \typ \sigma$
    \item $l \neq \epsilon \implies \sigma \sub C$
    \item $H \vdash \Gamma ; L$
    \item $\vdash \Gamma ; L ; P$
    \end{enumerate}
  \item By 4.a) and T-Let
    \begin{enumerate}[label=(\alph*)]
    \item $\Gamma ~;~ b \vdash y.\texttt{open}~\{ z \Rightarrow t' \} \typ \tau$
    \item $\Gamma , x \typ \tau ~;~ b \vdash t \typ \sigma$
    \end{enumerate}
  \item By 5.a) and T-Open
    \begin{enumerate}[label=(\alph*)]
    \item $\Gamma ~;~ b \vdash y \typ \GuaT Q D$
    \item $\Perm Q \in \Gamma$
    \item $z \typ D ~;~ \ocap \vdash t' \typ \tau'$
    \item $\tau = \GuaT Q D$
    \end{enumerate}
  \item By 6.a) and T-Var, $\Gamma(y) = \GuaT Q D$
  \item By 4.c), 7., and WF-Env, $H \vdash \Gamma ; L ; y$
  \item By 2.d), 7., 8., and WF-Var, $typeof(H, o) \sub D$
  \item By 9. and WF-Var, $H \vdash (z \typ D) ; [z \mapsto o] ; z$
  \item By 10. and WF-Env, $H \vdash (z \typ D) ; [z \mapsto o]$
  \item By 6.c), 11., and T-Frame1, $H \vdash G' \typ \tau'$
  \item By 2.d), 6.d), 9., and WF-Var, $H \vdash \Gamma , x \typ \tau ; L[x \mapsto L(y)] ; x$
  \item By 4.c), 13., and WF-Env, $H \vdash \Gamma , x \typ \tau ; L[x \mapsto L(y)]$
  \item By 4.b), 4.d), 5.b), 14., and T-Frame1, $H \vdash G \typ \sigma$
  \item By 3.b), 15., T-FS-A, and T-FS-NA, $H \vdash G \circ GS$
  \item By 12., 16., and T-FS-NA, $H \vdash FS'$
  \item By 1.c) and FS-ok
    \begin{enumerate}[label=(\alph*)]
    \item $H ~;~ a \vdash F \ok$
    \item $H ~;~ c \vdash GS \ok$
    \item $a = \begin{cases}\ocap & \text{if } c = \ocap \lor l = \epsilon \\
      \epsilon & \text{otherwise} \\
    \end{cases}$
    \item $boxSeparation(H, F, GS)$
    \item $uniqueOpenBox(H, F, GS)$
    \item $openBoxPropagation(H, F, GS)$
    \end{enumerate}
  \item By 18.a) and F-ok
    \begin{enumerate}[label=(\alph*)]
    \item $boxSep(H, F)$
    \item $boxObjSep(H, F)$
    \item $boxOcap(H, F)$
    \item $a = \ocap \implies globalOcapSep(H, F)$
    \item $fieldUniqueness(H, F)$
    \end{enumerate}
  \item By 2.d-e) and 19.c), $ocap(typeof(H, o))$
  \item By 2.b) and 2.f)
    \begin{enumerate}[label=(\alph*)]
    \item $boxSep(H, G')$
    \item $boxObjSep(H, G')$
    \item $boxOcap(H, G')$
    \item $fieldUniqueness(H, G')$
    \end{enumerate}
  \item By 1.f), 2.b,d,f), 19.b), and 20., $c = \ocap \implies globalOcapSep(H, G')$
  \item By 21.a-d), 22., and F-ok, $H ~;~ c \vdash G' \ok$
  \item By 1.f), 2.c), 18.a), and F-ok, $H ~;~ a \vdash G \ok$
  \item By 1.f), 2.c), 18., 24., and FS-ok, $H ~;~ a \vdash G \circ GS \ok$
  \item By 2.a-b), 23., 25., and FS-ok, $H ~;~ c \vdash FS' \ok$
  \item 1.a), 17., and 26. conclude this case.
  \end{enumerate}

\item[-] Case E-Box.
  \begin{enumerate}
  \item By the assumptions
    \begin{enumerate}[label=(\alph*)]
    \item $\vdash H \typ \star$
    \item $H \vdash FS$
    \item $H ~;~ a \vdash FS \ok$
    \item $\fsreduce H {FS} {H'} {FS'}$
    \item $FS = F \circ GS$
    \item $F = \pframe L {\texttt{box}[C]~\{ x \Rightarrow t \}} P l$
    \end{enumerate}
  \item By 1.d-f) and E-Box
    \begin{enumerate}[label=(\alph*)]
    \item $FS' = G \circ \epsilon$
    \item $G = \pframe {L[x \mapsto b(o, p)]} t {P \cup \{p\}} \epsilon$
    \item $o \notin dom(H)$
    \item $fields(C) = \seq{f}$
    \item $p~\text{fresh}$
    \item $H' = H[o \mapsto \obj C {\seq{f \mapsto \texttt{null}}}]$
    \end{enumerate}
  \item By 1.b), 1.e-f), T-FS-A, and T-FS-NA
    \begin{enumerate}[label=(\alph*)]
    \item $H \vdash F \typ \sigma$
    \item $l = \epsilon \implies H \vdash GS \land l = w \implies H \vdash^{\sigma}_w GS$
    \end{enumerate}
  \item By 1.f), 3.a), and T-Frame1
    \begin{enumerate}[label=(\alph*)]
    \item $\Gamma ~;~ b \vdash \texttt{box}[C]~\{ x \Rightarrow t \} \typ \sigma$
    \item $l \neq \epsilon \implies \sigma \sub C$
    \item $H \vdash \Gamma ; L$
    \item $\vdash \Gamma ; L; P$
    \end{enumerate}
  \item By 4.a) and T-Box
    \begin{enumerate}[label=(\alph*)]
    \item $ocap(C)$
    \item $\Gamma , x \typ \GuaT Q C , \Perm Q ~;~ b \vdash t \typ \tau$, $Q$ fresh
    \item $\sigma = \bot$
    \end{enumerate}
  \item By 2.c), 2.f), 4.c), and WF-Env, $H' \vdash \Gamma ; L$
  \item By 2.f) and $\sub$-Refl, $typeof(H', o) \sub C$
  \item Define
    \begin{enumerate}[label=(\alph*)]
    \item $\Gamma' := \Gamma , x \typ \GuaT Q C , \Perm Q$
    \item $L' := L[x \mapsto b(o, p)]$
    \end{enumerate}
  \item By 7., 8.a-b), and WF-Var, $H' \vdash \Gamma' ; L' ; x$
  \item By 6., 8.a-b), 9., and WF-Env, $H' \vdash \Gamma' ; L'$
  \item By 2.b), 4.d), 5.b), 8.a-b), 10., with $\gamma' = \gamma[Q\mapsto p]$, and T-Frame1, $H' \vdash G \typ \tau$
  \item By 11., T-EmpFS, and T-FS-NA, $H' \vdash FS'$
  \item By 1.c) and FS-ok
    \begin{enumerate}[label=(\alph*)]
    \item $H ~;~ a \vdash F \ok$
    \item $H ~;~ c \vdash GS \ok$
    \item $a = \begin{cases}\ocap & \text{if } c = \ocap \lor l = \epsilon \\
      \epsilon & \text{otherwise} \\
    \end{cases}$
    \item $boxSeparation(H, F, GS)$
    \item $uniqueOpenBox(H, F, GS)$
    \item $openBoxPropagation(H, F, GS)$
    \end{enumerate}
  \item By 2.f) and 5.a), $ocap(typeof(H', o))$
  \item By 1.f), 2.b,c,f), 13.a), 14., and F-ok, $H' ~;~ a \vdash G \ok$
  \item By 15. and SingFS-ok, $H' ~;~ a \vdash FS' \ok$
  \item By 1.a), 2.d,f), and def. well-typed heap, $\vdash H' \typ \star$
  \item 12., 16., and 17. conclude this case.
  \end{enumerate}

\item[-] Case E-Capture.
  \begin{enumerate}
  \item By the assumptions
    \begin{enumerate}[label=(\alph*)]
    \item $\vdash H \typ \star$
    \item $H \vdash FS$
    \item $H ~;~ a \vdash FS \ok$
    \item $\fsreduce H {FS} {H'} {FS'}$
    \item $FS = F \circ GS$
    \item $F = \pframe L {\texttt{capture}(x.f, y)~\{ z \Rightarrow t \}} P l$
    \end{enumerate}
  \item By 1.d-f) and E-Capture
    \begin{enumerate}[label=(\alph*)]
    \item $FS' = F' \circ \epsilon$
    \item $F' = \pframe {L[z \mapsto L(x)]} t {P \setminus \{p'\}} \epsilon$
    \item $L(x) = b(o, p)$
    \item $L(y) = b(o', p')$
    \item $\{ p, p' \} \subseteq P$
    \item $H(o) = \obj C {FM}$
    \item $H' = H[o \mapsto \obj C {FM[f \mapsto o']}]$
    \end{enumerate}
  \item By 1.b), 1.e-f), T-FS-A, and T-FS-NA
    \begin{enumerate}[label=(\alph*)]
    \item $H \vdash F \typ \sigma$
    \item $l = \epsilon \implies H \vdash GS \land l = w \implies H \vdash^{\sigma}_w GS$
    \end{enumerate}
  \item By 1.f), 3.a), and T-Frame1
    \begin{enumerate}[label=(\alph*)]
    \item $\Gamma ~;~ b \vdash \texttt{capture}(x.f, y)~\{ z \Rightarrow t \} \typ \sigma$
    \item $l \neq \epsilon \implies \sigma \sub C$
    \item $H \vdash \Gamma ; L$
    \item $\vdash \Gamma ; L ; P$
    \end{enumerate}
  \item By 4.a) and T-Capture
    \begin{enumerate}[label=(\alph*)]
    \item $\Gamma ~;~ b \vdash x \typ \GuaT Q C$
    \item $\Gamma ~;~ b \vdash y \typ \GuaT {Q'} {D'}$
    \item $\{ \Perm Q , \Perm {Q'} \} \subseteq \Gamma$
    \item $D' \sub ftype(C, f)$
    \item $\Gamma \setminus \{ \Perm {Q'} \} , z \typ \GuaT Q C ~;~ b \vdash t \typ \tau$
    \item $\sigma = \bot$
    \end{enumerate}
  \item By 5.a-b) and T-Var
    \begin{enumerate}[label=(\alph*)]
    \item $\Gamma(x) = \GuaT Q C$
    \item $\Gamma(y) = \GuaT {Q'} {D'}$
    \end{enumerate}
  \item By 2.c), 2.f-g), 6.a), WF-Var, and WF-Env, $H' \vdash \Gamma ; L$
  \item Define
    \begin{enumerate}[label=(\alph*)]
    \item $\Gamma' := \Gamma \setminus \{ \Perm {Q'} \} , z \typ \GuaT Q C$
    \item $L' := L[z \mapsto L(x)]$
    \end{enumerate}
  \item By 2.c), 2.g), 8.a-b), $\sub$-Refl, and WF-Var, $H' \vdash \Gamma' ; L' ; z$
  \item By 7., 9., and WF-Env, $H' \vdash \Gamma' ; L'$
  \item By 2.b), 4.d), 5.e), 8.a), 10., and T-Frame1, $H' \vdash F' \typ \tau$
  \item By 2.a), 11., T-EmpFS, and T-FS-NA, $H' \vdash FS'$
  \item By 1.c) and FS-ok
    \begin{enumerate}[label=(\alph*)]
    \item $H ~;~ a \vdash F \ok$
    \item $H ~;~ c \vdash GS \ok$
    \item $a = \begin{cases}\ocap & \text{if } c = \ocap \lor l = \epsilon \\
      \epsilon & \text{otherwise} \\
    \end{cases}$
    \item $boxSeparation(H, F, GS)$
    \item $uniqueOpenBox(H, F, GS)$
    \item $openBoxPropagation(H, F, GS)$
    \end{enumerate}
  \item By 13.a) and F-ok
    \begin{enumerate}[label=(\alph*)]
    \item $boxSep(H, F)$
    \item $boxObjSep(H, F)$
    \item $boxOcap(H, F)$
    \item $a = \ocap \implies globalOcapSep(H, F)$
    \item $fieldUniqueness(H, F)$
    \end{enumerate}
  \item By 2.c-g), 8.b), and 14.a), $boxSep(H, F')$
  \item By 2.c-g), 8.b), and 14.b), $boxObjSep(H', F')$
  \item By 2.c-d) and 14.c)
    \begin{enumerate}[label=(\alph*)]
    \item $\forall o_1 \in dom(H).~ \\ reach(H, o, o_1) \implies ocap(typeof(H, o_1))$
    \item $\forall o_2 \in dom(H).~ \\ reach(H, o', o_2) \implies ocap(typeof(H, o_2))$
    \end{enumerate}
  \item By 2.g) and 17.a-b), $\forall o_1 \in dom(H').~ \\ reach(H', o, o_1) \implies ocap(typeof(H', o_1))$
  \item By 2.c,e,g), 8.b), 14.c), and 18., $boxOcap(H', F')$
  \item By 2.c-g), 8.b), 14.b,d), $a = \ocap \implies \forall x_1 \mapsto o_1 \in L', x_2 \mapsto o_2 \in L_0.~ocap(typeof(H', o_1)) \land sep(H', o_1, o_2)$
  \item By 1.f), 2.b-g), and 14.e), $fieldUniqueness(H', F')$
  \item By 2.b), 8.b), 15., 16., 19., 20., 21., and F-ok, $H' ~;~ a \vdash F' \ok$
  \item By 2.a), 22., and SingFS-ok, $H' ~;~ a \vdash FS' \ok$
  \item By 2.d), 4.c), 6.b), and WF-Env, $H \vdash \Gamma ; L ; y$
  \item By 2.d), 6.b), 24., and WF-Var, $typeof(H, o') \sub D'$
  \item By 5.d), 25., and $\sub$-Trans, $typeof(H, o') \sub ftype(C, f)$
  \item By 2.g) and 26., $typeof(H', o') \sub ftype(C, f)$
  \item By 1.a), 2.f-g), 27., and def. well-typed heap, $\vdash H' \typ \star$
  \item 12., 23., 28., and 29. conclude this case.
  \end{enumerate}

\item[-] Case E-Swap.
  \begin{enumerate}
  \item By the assumptions
    \begin{enumerate}[label=(\alph*)]
    \item $\vdash H \typ \star$
    \item $H \vdash FS$
    \item $H ~;~ a \vdash FS \ok$
    \item $\fsreduce H {FS} {H'} {FS'}$
    \item $FS = F \circ GS$
    \item $F = \pframe L {\texttt{swap}(x.f, y)~\{ z \Rightarrow t \}} P l$
    \end{enumerate}
  \item By 1.d-f) and E-Swap
    \begin{enumerate}[label=(\alph*)]
    \item $FS' = F' \circ \epsilon$
    \item $F' = \pframe {L[z \mapsto b(o'', p'')]} t {P \cup \{p''\} \setminus \{p'\}} \epsilon$
    \item $L(x) = b(o, p)$
    \item $L(y) = b(o', p')$
    \item $\{ p, p' \} \subseteq P$
    \item $H(o) = \obj C {FM}$
    \item $H' = H[o \mapsto \obj C {FM[f \mapsto o']}]$
    \item $p''$ fresh
    \end{enumerate}
  \item By 1.b), 1.e-f), T-FS-A, and T-FS-NA
    \begin{enumerate}[label=(\alph*)]
    \item $H \vdash F \typ \sigma$
    \item $l = \epsilon \implies H \vdash GS \land l = w \implies H \vdash^{\sigma}_w GS$
    \end{enumerate}
  \item By 1.f), 3.a), and T-Frame1
    \begin{enumerate}[label=(\alph*)]
    \item $\Gamma ~;~ b \vdash \texttt{swap}(x.f, y)~\{ z \Rightarrow t \} \typ \sigma$
    \item $l \neq \epsilon \implies \sigma \sub C$
    \item $H \vdash \Gamma ; L$
    \item $\vdash \Gamma ; L ; P$
    \end{enumerate}
  \item By 4.a) and T-Swap
    \begin{enumerate}[label=(\alph*)]
    \item $\Gamma ~;~ b \vdash x \typ \GuaT Q C$
    \item $\Gamma ~;~ b \vdash y \typ \GuaT {Q'} {D'}$
    \item $\{ \Perm Q , \Perm {Q'} \} \subseteq \Gamma$
    \item $\texttt{Box}[D] = ftype(C, f)$
    \item $\Gamma \setminus \{ \Perm {Q'} \} , z \typ \GuaT R D, \Perm R ~;~ b \vdash t \typ \tau$
    \item $\sigma = \bot$
    \item $D'  \sub D$
    \item $R$ fresh
    \end{enumerate}
  \item By 5.a-b) and T-Var
    \begin{enumerate}[label=(\alph*)]
    \item $\Gamma(x) = \GuaT Q C$
    \item $\Gamma(y) = \GuaT {Q'} {D'}$
    \end{enumerate}
  \item By 2.c), 2.f-g), 6.a), WF-Var, and WF-Env, $H' \vdash \Gamma ; L$
  \item Define
    \begin{enumerate}[label=(\alph*)]
    \item $\Gamma' := \Gamma \setminus \{ \Perm {Q'} \} , z \typ \GuaT R D , \Perm R$
    \item $L' := L[z \mapsto b(o'', p'')]$
    \end{enumerate}
  \item By 2.c), 2.g), 8.a-b), $\sub$-Refl, and WF-Var, $H' \vdash \Gamma' ; L' ; z$
  \item By 7., 9., and WF-Env, $H' \vdash \Gamma' ; L'$
  \item By 2.b), 4.d), 5.e), 8.a), 10., and T-Frame1, $H' \vdash F' \typ \tau$
  \item By 2.a), 11., T-EmpFS, and T-FS-NA, $H' \vdash FS'$
  \item By 1.c) and FS-ok
    \begin{enumerate}[label=(\alph*)]
    \item $H ~;~ a \vdash F \ok$
    \item $H ~;~ c \vdash GS \ok$
    \item $a = \begin{cases}\ocap & \text{if } c = \ocap \lor l = \epsilon \\
      \epsilon & \text{otherwise} \\
    \end{cases}$
    \item $boxSeparation(H, F, GS)$
    \item $uniqueOpenBox(H, F, GS)$
    \item $openBoxPropagation(H, F, GS)$
    \end{enumerate}
  \item By 13.a) and F-ok
    \begin{enumerate}[label=(\alph*)]
    \item $boxSep(H, F)$
    \item $boxObjSep(H, F)$
    \item $boxOcap(H, F)$
    \item $a = \ocap \implies globalOcapSep(H, F)$
    \item $fieldUniqueness(H, F)$
    \end{enumerate}
  \item (Removed.)
  \item By 2.c), 2.e-f), 2.i), 5.d), 14.e), and def. $fieldUniqueness$, $domedge(H, o, f, o, o'')$
  \item By 1.f), 2.c-e), 14.a), and def. $boxSep$, $sep(H, o, o')$
    \item By 2.b-h), 16., 17., and def. $boxSep$, $boxSep(H', F')$
  \item By 2.c-g), 8.b), and 14.b), $boxObjSep(H', F')$
  \item By 2.c-d) and 14.c)
    \begin{enumerate}[label=(\alph*)]
    \item $\forall o_1 \in dom(H).~ \\ reach(H, o, o_1) \implies ocap(typeof(H, o_1))$
    \item $\forall o_2 \in dom(H).~ \\ reach(H, o', o_2) \implies ocap(typeof(H, o_2))$
    \end{enumerate}
  \item By 2.g) and 17.a-b), $\forall o_1 \in dom(H').~ \\ reach(H', o, o_1) \implies ocap(typeof(H', o_1))$
  \item By 2.c,e,g), 8.b), 14.c), and 18., $boxOcap(H', F')$
  \item By 2.c-g), 8.b), 14.b,d), $a = \ocap \implies \forall x_1 \mapsto o_1 \in L', x_2 \mapsto o_2 \in L_0.~ocap(typeof(H', o_1)) \land sep(H', o_1, o_2)$
  \item By 2.f-g) and 17., $\forall \hat{o} \in dom(H').~reach(H', o', \hat{o}) \implies domedge(H', o, f, o, \hat{o})$
  \item By 1.f), 2.c-e), 14.e), and def. $fieldUniqueness$
    \begin{enumerate}[label=(\alph*)]
    \item $\forall o_1, o_2 \in dom(H).~ \\ reach(H, o, o_1) \land H(o_1) = \obj {C_1} {FM_1} \land ftype(C_1, f_1) = \texttt{Box}[D_1] \land reach(H, FM_1(f_1), o_2) \implies \\ domedge(H, o_1, f_1, o, o_2)$
    \item $\forall o_1, o_2 \in dom(H).~ \\ reach(H, o', o_1) \land H(o_1) = \obj {C_1} {FM_1} \land ftype(C_1, f_1) = \texttt{Box}[D_1] \land reach(H, FM_1(f_1), o_2) \implies \\ domedge(H, o_1, f_1, o', o_2)$
    \end{enumerate}
    \item By 2.b-i), 24., 25.a-b), and def. $fieldUniqueness$, $fieldUniqueness(H', F')$
  \item By 2.b), 8.b), 18., 19., 22., 23., 26., and F-ok, $H' ~;~ a \vdash F' \ok$
  \item By 2.a), 27., and SingFS-ok, $H' ~;~ a \vdash FS' \ok$
  \item By 2.d), 4.c), 6.b), and WF-Env, $H \vdash \Gamma ; L ; y$
  \item By 2.d), 6.b), 29., and WF-Var, $typeof(H, o') \sub D'$
  \item By 5.g), 30., and $\sub$-Trans, $typeof(H, o') \sub D$
  \item By 2.g) and 31., $typeof(H', o') \sub D$
  \item By 1.a), 2.f-g), 32., and def.~\ref{def:well-typed-heap2}, $\vdash H' \typ \star$
  \item 12., 28., and 33. conclude this case.
  \end{enumerate}

\item[-] Case E-Frame. We only consider the case where $\freduce H F
  {H'} {F'}$ by E-Assign; the other cases follow analogously.
  \begin{enumerate}
  \item By the assumptions
    \begin{enumerate}[label=(\alph*)]
    \item $\vdash H \typ \star$
    \item $H \vdash F \circ FS$
    \item $H ~;~ a \vdash F \circ FS \ok$
    \item $\fsreduce H {F \circ FS} {H'} {F' \circ FS}$
    \item $\freduce H F {H'} {F'}$
    \item $F = \pframe L {\texttt{let}~x = y.f = z~\texttt{in}~t} P l$
    \item $F' = \pframe L {\texttt{let}~x = z~\texttt{in}~t} P l$
    \item $L(y) = o$
    \item $H(o) = \obj C {FM}$
    \item $H' = H[o \mapsto \obj C {FM[f \mapsto L(z)]}]$
    \end{enumerate}
  \item By 1.b), T-FS-A, and T-FS-NA
    \begin{enumerate}[label=(\alph*)]
    \item $H \vdash F \typ \sigma$
    \item $l = \epsilon \implies H \vdash FS \land l = w \implies H \vdash^{\sigma}_w FS$
    \end{enumerate}
  \item By 1.c) and FS-ok, $H ~;~ a \vdash F \ok$
  \item By 1.a), 1.e), 2.a), 3., and Lemma~\ref{lem:lemma1}
    \begin{enumerate}[label=(\alph*)]
    \item $\vdash H' \typ \star$
    \item $H' \vdash F' \typ \sigma$
    \item $H' ~;~ a \vdash F' \ok$
    \item $\vdash \Gamma ; L ; P$
    \end{enumerate}
  \item By 3. and F-ok
    \begin{enumerate}[label=(\alph*)]
    \item $boxSep(H, F)$
    \item $boxObjSep(H, F)$
    \item $boxOcap(H, F)$
    \item $a = \texttt{ocap} \implies globalOcapSep(H, F)$
    \item $fieldUniqueness(H, F)$
    \end{enumerate}
  \item By 1.f), 2.a), and T-Frame1
    \begin{enumerate}[label=(\alph*)]
    \item $\Gamma ~;~ b \vdash \texttt{let}~x = y.f = z~\texttt{in}~t \typ \sigma$
    \item $l \neq \epsilon \implies \sigma \sub \hat{C}$
    \item $H \vdash \Gamma ; L$
    \end{enumerate}
  \item By 6.a) and T-Let
    \begin{enumerate}[label=(\alph*)]
    \item $\Gamma ~;~ b \vdash y.f = z \typ \tau$
    \item $\Gamma , x \typ \tau ~;~ b \vdash t \typ \sigma$
    \end{enumerate}
  \item By 7.a) and T-Assign
    \begin{enumerate}[label=(\alph*)]
    \item $\Gamma ~;~ b \vdash y \typ D$
    \item $ftype(D, f) = E$
    \item $\Gamma ~;~ b \vdash z \typ E'$
    \item $E' \sub E$
    \item $\tau = E$
    \end{enumerate}
  \item By 8.c) and T-Var, $\Gamma(z) = E'$
  \item By 6.c), 9., and WF-Env, $H \vdash \Gamma ; L ; z$
  \item By 9., 10., and WF-Var, $L(z) = \texttt{null} \lor L(z) \in dom(H) \land typeof(H, L(z)) \sub E'$
  \item By 1.h), 5.b), 11., and def.~\ref{def:box-obj-sep}
    \begin{enumerate}[label=(\alph*)]
    \item $\forall \hat{x} \mapsto b(\hat{o}, \hat{p}) \in L.~sep(H, \hat{o}, o)$
    \item $\forall \hat{x} \mapsto b(\hat{o}, \hat{p}) \in L.~L(z) = \texttt{null} \lor sep(H, \hat{o}, L(z))$
    \end{enumerate}
  \item Case: $openbox(H, \hat{o}, F, FS)$ for some $\hat{o} \in
    dom(H)$. Then by Lemma~\ref{lem:lemma3}, $a = \texttt{ocap}$.
  \item By 5.d) and 13., $globalOcapSep(H, F)$
  \item By 4.c), 13., and F-ok
    \begin{enumerate}[label=(\alph*)]
    \item $boxSep(H', F')$
    \item $boxObjSep(H', F')$
    \item $boxOcap(H', F')$
    \item $globalOcapSep(H', F')$
    \item $fieldUniqueness(H', F')$
    \end{enumerate}
  \item By 1.c) and FS-ok
    \begin{enumerate}[label=(\alph*)]
    \item $H ~;~ a \vdash F \ok$
    \item $H ~;~ c \vdash FS \ok$
    \item $a = \begin{cases}
      \texttt{ocap} & \text{if } c = \texttt{ocap} \lor l = \epsilon \\
      \epsilon      & \text{otherwise} \\
    \end{cases}$
    \item $boxSep(H, F, FS)$
    \item $uniqueOpenBox(H, F, FS)$
    \item $openBoxPropagation(H, F, FS)$
    \end{enumerate}
  \item We show $H' ~;~ c \vdash FS \ok$ by induction on the size of $FS$. Let $FS = G \circ \epsilon$. By 16.b) and SingFS-ok, $H ~;~ c \vdash G \ok$
  \item By 17. and F-ok
    \begin{enumerate}[label=(\alph*)]
    \item $boxSep(H, G)$
    \item $boxObjSep(H, G)$
    \item $boxOcap(H, G)$
    \item $c = \texttt{ocap} \implies globalOcapSep(H, G)$
    \item $fieldUniqueness(H, G)$
    \end{enumerate}
  \item By 1.j), 13., and 16.e), $boxSep(H', G)$
  \item Assume $l \neq \epsilon$. Then by 13., 16.f), and def.~\ref{def:open-box-propag1}, $openbox(H, \hat{o}, G, \epsilon)$. Contradiction. Therefore, $l = \epsilon$.
    \item By 1.j), 20., $\twoheadrightarrow$, and E-Open, $G = \pframe
      {L_G} {t_G} {P_G} m \implies \forall \hat{x} \mapsto \hat{o} \in
      L_G.~sep(H', o, \hat{o}) \land (L(z) = \texttt{null} \lor
      sep(H', L(z), \hat{o}))$
    \item By 1.j), 18.b), and 21., $boxObjSep(H', G)$
    \item By 1.j), 13., 18.c), 20., $\twoheadrightarrow$, and E-Open, $boxOcap(H', G)$
    \item By 1.j), 13., 18.d), 20., $\twoheadrightarrow$, and E-Open, $c = \texttt{ocap} \implies globalOcapSep(H', G)$
    \item By 1.j), 15.e), 18.e), and def.~\ref{def:field-uniqueness}, \\ $fieldUniqueness(H', G)$
    \item By 19., 22., 23., 24., 25., and F-ok, $H' ~;~ c \vdash G \ok$
    \item By 26. and SingFS-ok, $H' ~;~ c \vdash FS \ok$
    \item Let $FS = G \circ GS$
    \item If $l = \epsilon$ then $H' ~;~ c \vdash G \ok$ as before.
    \item Let $l = v \neq \epsilon$
    \item By 13., 16.f), 30., and def.~\ref{def:open-box-propag1}, $openbox(H, \hat{o}, G, GS)$
    \item By 1.j), 13., and 16.e), $boxSep(H', G)$
    \item By 1.j), 5.b), 18.b), 31., and def.~\ref{def:box-obj-sep}, $boxObjSep(H', G)$
    \item By 1.j), 5.b), 13., 18.c), and 31., $boxOcap(H', G)$
    \item By 1.j), 14., and 18.d), \\ $c = \texttt{ocap} \implies globalOcapSep(H', G)$
    \item By 1.j), 15.e), and 18.e), $fieldUniqueness(H', G)$
    \item By 32., 33., 34., 35., 36., and F-ok, $H' ~;~ c \vdash G \ok$
    \item By 29. and 37., $H' ~;~ c \vdash G \ok$
    \item By 16.b) and FS-ok
      \begin{enumerate}[label=(\alph*)]
      \item $H ~;~ c \vdash G \ok$
      \item $H ~;~ d \vdash GS \ok$
      \item $c = \begin{cases}
        \texttt{ocap} & \text{if } d = \texttt{ocap} \lor label(G) = \epsilon \\
        \epsilon      & \text{otherwise} \\
      \end{cases}$
      \item $boxSep(H, G, GS)$
      \item $uniqueOpenBox(H, G, GS)$
      \item $openBoxPropagation(H, G, GS)$
      \end{enumerate}
    \item By 1.h-g), 13., 16.e), 31., and 33., $boxSep(H', G, GS)$
    \item By 1.h-g), 13., 16.e), 31., and 39.e), \\ $uniqueOpenBox(H', G, GS)$
    \item By 1.h-g), 5.b), 13., 31., and 39.f), \\ $openBoxPropagation(H', G, GS)$
    \item By 38., 39.c), 40., 41., 42., and the IH, $H' ~;~ c \vdash FS \ok$
    \item The case $\lnot \exists \hat{o} \in dom(H).~openbox(H, \hat{o}, F, FS)$ follows analogously to 13.-43.
    \item By 1.f-j), 5.b), and 16.d), $boxSep(H', F', FS)$
    \item By 1.f-j), 5.b), and 16.e), $uniqueOpenBox(H', F', FS)$
    \item By 1.f-j), 5.b), and 16.e-f), \\ $openBoxPropagation(H', F', FS)$
    \item By 4.c), 16.c), 43., 44., 45., 46., 47., and FS-ok, $H' ~;~ a \vdash F' \circ FS \ok$
    \item By 1.i-j), 2.b), WF-Env, and WF-Var, $l = \epsilon \implies H' \vdash FS \land l = w \implies H' \vdash^{\sigma}_w FS$
    \item By 4.b), 49., T-FS-A, and T-FS-NA, $H' \vdash F' \circ FS$
    \item 4.a), 48., and 50. conclude this case.

  \end{enumerate}

\end{itemize}

\end{proof}

\begin{lemma}\label{lem:lemma3}
If $\vdash H \typ \star$ then:

If $(H \vdash F \circ FS \lor H \vdash^{\sigma}_x F \circ FS)$ and $H
~;~ a \vdash F \circ FS \ok$ then $openbox(H, \hat{o}, F, FS) \implies
a = \texttt{ocap}$
\end{lemma}

\begin{proof}
  By induction on the size of $FS$.
  \begin{enumerate}
  \item Let $FS = G \circ \epsilon$. By the assumptions
    \begin{enumerate}[label=(\alph*)]
      \item $\vdash H \typ \star$
      \item $H \vdash F \circ FS$
      \item $H ~;~ a \vdash F \circ FS \ok$
      \item $openbox(H, \hat{o}, F, FS)$
    \end{enumerate}
  \item By 1.c) and FS-ok
    \begin{enumerate}[label=(\alph*)]
      \item $H ~;~ a \vdash F \ok$
      \item $H ~;~ b \vdash G \circ \epsilon \ok$
      \item $a = \begin{cases}
        \texttt{ocap} & \text{if } b = \texttt{ocap} \lor label(F) = \epsilon \\
        \epsilon      & \text{otherwise} \\
      \end{cases}$
      \item $uniqueOpenBox(H, F, FS)$
      \item $openBoxPropagation(H, F, FS)$
    \end{enumerate}
  \item Assume $label(F) \neq \epsilon$. Then by 1.d) and 2.e), \\ $openbox(H, \hat{o}, G, \epsilon)$. Contradiction. Therefore, $label(F) = \epsilon$ and by 2.c), $a = \texttt{ocap}$.
  \item Let $FS = G \circ GS$. By the assumptions
    \begin{enumerate}[label=(\alph*)]
      \item $\vdash H \typ \star$
      \item $H \vdash F \circ FS$
      \item $H ~;~ a \vdash F \circ FS \ok$
      \item $openbox(H, \hat{o}, F, FS)$
    \end{enumerate}
  \item By 4.c) and FS-ok
    \begin{enumerate}[label=(\alph*)]
      \item $H ~;~ a \vdash F \ok$
      \item $H ~;~ b \vdash G \circ GS \ok$
      \item $a = \begin{cases}
        \texttt{ocap} & \text{if } b = \texttt{ocap} \lor label(F) = \epsilon \\
        \epsilon      & \text{otherwise} \\
      \end{cases}$
      \item $uniqueOpenBox(H, F, FS)$
      \item $openBoxPropagation(H, F, FS)$
    \end{enumerate}
  \item Assume $label(F) = x \neq \epsilon$. Then by 4.d) and 5.e), $openbox(H, \hat{o}, G, GS)$.
  \item By 4.b), 6., and T-FS-A
    \begin{enumerate}[label=(\alph*)]
      \item $H \vdash F^x \typ \sigma$
      \item $H \vdash^{\sigma}_x G \circ GS$
    \end{enumerate}
  \item By 4.a), 7.b), 5.b), 6., and the IH, $b = \texttt{ocap}$
  \item By 5.c) and 8., $a = \texttt{ocap}$

  \end{enumerate}
\end{proof}


%% file: lacasa_core2_progress_proof_revised.tex
\subsection{Proof of Theorem~\ref{thm:core2-progress}}
\label{app:core2-progress-proof}




\begin{lemma}\label{lem:subclass-field}
  If \(C' \sub C \) and \(f \in fields(C)\), then
  \begin{enumerate}
  \item \(f \in fields(C')\)
  \item \(ftype(C, f) = ftype(C', f)\)
  \end{enumerate}
\end{lemma}
\begin{proof}
  Directly by \(\sub\)-Ext and WF-Class.
\end{proof}


\begin{theoremnonum}[Progress]
  If $\vdash H \typ \star$ then:

  If $H \vdash FS$ and $H ~;~ a \vdash FS \ok$ then either
  $\fsreduce H {FS} {H'} {FS'}$ or $FS = \pframe L x P l \circ \epsilon$ or $FS = F \circ GS$
  where
  \begin{itemize}
  \item $F = \pframe L {\texttt{let}~x = t~\texttt{in}~t'} P l$,
        $t \in \{ y.f, y.f = z, y.m(z), y.\texttt{open}~\{ z \Rightarrow t'' \} \}$,
        and $L(y) = \texttt{null}$; or
  \item $F = \pframe L {\texttt{capture}(x.f, y)~\{ z \Rightarrow t \}} P l$ where $L(x) = \texttt{null} \wedge L(y) = \texttt{null}$; or
  \item $F = \pframe L {\texttt{swap}(x.f, y)~\{ z \Rightarrow t \}} P l$ where $L(x) = \texttt{null} \wedge L(y) = \texttt{null}$.
  \end{itemize}

\end{theoremnonum}
\begin{proof}
  
\begin{itemize}
\item [-]
  By \(H ; a \vdash FS \ok\), \(FS = F \circ FS'\).
\item [-] Define \(\pframe L t P l := F\).

\item [-]
By induction on the structure of \(t\),

\item [-]  Case \(t=x\)
  \begin{enumerate}
  \item Assume \(FS' = F' \circ FS''\); otherwise the theorem conclusion applies.

  \item By the assumptions,
    \(H \vdash FS\).

  \item By 2., T-FS-NA and T-FS-A,
    \begin{enumerate}[label=(\alph*)]
    \item \(\Gamma ~;~ a \vdash x \typ \sigma \)
    \item \(H \vdash \Gamma ; L\)
    \end{enumerate}
    
  \item By 3., WF-Env, \(dom(\Gamma) \subseteq dom(L)\).
    
  \item By T-Var, \(\Gamma (x) = \sigma \wedge x \in dom(\Gamma)\).
    
  \item By 4. and 5., \(x \in dom(L)\).
    
  \item Define 
    \begin{enumerate}[label=(\alph*)]
    \item \(\pframe {L'} {t'} {P'} {l'} := F'\)
    \item \(L'' := \begin{cases}
        L'[y \mapsto L(x)] & \text{ if } l=y\\
        L' & \text{ otherwise }
      \end{cases}\), which is well-defined by 6.
    \item \(F'' := \pframe {L''} {t'} {P'} {l'}\)
    \end{enumerate}

  \item By E-Return1, E-Return2, 
    \(\freduce H F H {F'}\)
  \item By 8., \(\fsreduce H {F \circ F' \circ FS''} H {F'' \circ FS''}\)

  \item 9. concludes this case.
    
  \end{enumerate}
  
\item [-] Case \(t=\texttt{let}~x = \texttt{null} ~ \texttt{in}~t'\)
  \begin{enumerate}
  \item By E-Null, 
    \(\freduce H F H {\pframe {L[x\mapsto \texttt{null}]} t P l}\)
  \item By 1., E-Frame
    \(\fsreduce H {F \circ FS'} H {{\pframe {L[x\mapsto \texttt{null}]} t P l} \circ {FS'}}\)
  \item 2. concludes this case.
  \end{enumerate}

\item [-] Case \(t=\texttt{let}~x = \texttt{new}~C ~ \texttt{in }~t'\)
  \begin{enumerate}
  \item Define
    \begin{enumerate}[label=(\alph*)]
    \item \(o\) fresh
    \item \(\seq{f} := fields(C)\)
    \item \(H' := H [o \mapsto \obj C {\seq {f \mapsto \texttt{null}}}]\)
    \item \(L' := L[x \mapsto o]\)
    \item \(F' := \pframe{L'} {t'} P l\)
    \end{enumerate}
    
  \item By E-New, 
    \(\freduce H F {H'} {F'}\)
  \item By 1., E-Frame
    \(\fsreduce H {F \circ FS'} {H'} {{F'} \circ {FS'}}\)
  \item 3. concludes this case.
  \end{enumerate}

\item [-] Case \(t=\texttt{box}[C]\{x \Rightarrow t'\}\)
  \begin{enumerate}
  \item Define
    \begin{enumerate}[label=(\alph*)]
    \item \(o\) fresh
    \item \(\seq{f} := fields(C)\)
    \item \(H' := H [o \mapsto \obj C {\seq {f \mapsto \texttt{null}}}]\)
    \item \(p\) fresh
    \item \(L' := L[x \mapsto b(o, p)]\)
    \item \(P' := P \cup \{p\}\)
    \item \(F' := \pframe{L'} {t'} {P'} \epsilon\)
    \end{enumerate}
    
  \item By E-Box,
    \(\fsreduce H {F \circ FS'} {H'} {{F'} \circ \epsilon}\)
  \item 2. concludes this case.
  \end{enumerate}

\item [-] Case \(t={\texttt{let}~x = y~\texttt{in}~t'}\)
  \begin{enumerate}
  \item By the assumption
    \begin{enumerate} [label=(\alph*)]
    \item \(\vdash H \typ \star \)
    \item \(H \vdash F \circ FS'\)
    \end{enumerate}

  \item By 1.b), T-FS-NA and T-FS-A, 
    \(H \vdash F \typ \sigma\)
    
  \item By 2., T-Frame1,
    \begin{enumerate} [label=(\alph*)]
    \item \(\Gamma ~;~ a \vdash \texttt{let}~ x = y ~ \texttt{in} ~ t' \typ \sigma\)
    \item \(H \proves \Gamma ~;~ L\)
    \end{enumerate}

  \item By 3.b), WF-Env, \(dom(\Gamma) \subseteq dom(L)\)

  \item By 4.a), T-Var, T-Let, \(y \in dom(\Gamma)\)
    
  \item By 4, 5, \(y \in dom(L)\), and \(L(y)\) is defined.

  \item Define 
    \(L' := L[x \mapsto L(y)]\)

  \item 
    By E-Var, 7.,  \(\freduce H F H {\pframe {L'} {t'} P l}\)

  \item By 8., E-Frame, \(\fsreduce H {F \circ FS'} H {F' \circ FS'}\)

  \item 9. concludes this case.
    
  \end{enumerate}

\item [-] Case \(t={\texttt{let}~x = y.f~\texttt{in}~t'}\)

\begin{enumerate}
  \item By the assumption
    \begin{enumerate} [label=(\alph*)]
    \item \(\vdash H \typ \star \)
    \item \(H \vdash F \circ FS'\)
    \end{enumerate}

  \item By 1.b), T-FS-NA and T-FS-A,
    \(H \vdash F \typ \sigma\)
    
  \item By 2., T-Frame1,

    \begin{enumerate} [label=(\alph*)]
    \item \(\Gamma ~;~ a \vdash {\texttt{let}~x = y.f~\texttt{in}~t'} \typ \sigma \)
    \item \(H \vdash \Gamma ; L\)
    \end{enumerate}

  \item By 3.a), T-Let, T-Select, T-Var,
    \begin{enumerate} [label=(\alph*)]
    \item \(\Gamma(y) = C\)
    \item \(ftype(C, f) = D\)
    \end{enumerate}

  \item By 3.b), WF-Env,
    \begin{enumerate} [label=(\alph*)]
    \item \(dom(\Gamma) \subseteq dom(L)\)
    \item \(\forall x \in dom(\Gamma), H \vdash \Gamma ; L ; x\)
    \end{enumerate}

  \item By 4.a)., 5.a), \(y \in dom(L)\)

  \item By 5.b),  WF-Var, 4.a), and assuming \(L(y) = \texttt{null}\) 
    (otherwise the theorem conclusion holds),
    \begin{enumerate} [label=(\alph*)]
    \item \(L(y) = o\)
    \item \(typeof(H, o) \sub C\)
    \end{enumerate}

  \item By 4.b),
    \(f \in fields(C)\)

  \item By 7.b), \(o \in dom(H)\).

  \item Define \(\obj {C'} {FM} := H(o)\) (which is defined by 9.)
    
  \item By 1.a), 10, and Definition~\ref{def:well-typed-heap}, 
    \(dom(FM) = fields(C')\)
    
  \item By 10, \(typeof(H, o) = C'\)

  \item By  7.b), 11., 8., and Lemma~\ref{lem:subclass-field},
    \(f \in dom(FM)\).
  
  \item Define 
    \begin{enumerate} [label=(\alph*)]
    \item \(L' := L[x \mapsto FM(f)]\) (defined by 13.)
    \item \(F' := \pframe {L'} {t'} P l\)
    \end{enumerate}

  \item By 14.a-b), E-Select, \(\freduce H F H {F'}\)

  \item By 15, E-Frame, \(\fsreduce H {F \circ FS'} H {F'\circ FS}\)
  \item 16. concludes the proof.

  \end{enumerate}

\item [-] Case \(t = {\texttt{capture}(x.f, y) \{ z \Rightarrow t'}\}\)
  \begin{enumerate}

  \item By assumption
    \begin{enumerate} [label=(\alph*)]
    \item \(\vdash H \typ \star \)
    \item \(H \vdash F \circ FS'\)
    \item \(H~;~a \vdash FS \ok\)
    \end{enumerate}

  \item 
    By 1.b), T-FS-NA, and T-FS-A,
    \begin{enumerate} [label=(\alph*)]
    \item \(\Gamma ~;~ a \vdash {\texttt{capture}(x.f, y) \{ z \Rightarrow t'}
      \typ \sigma\)
    \item \(H \vdash \Gamma; L\)
    \item \(H \vdash \Gamma ; L; P\)
    \end{enumerate}

  \item By 2.c), \(\forall x \in dom(\Gamma), 
    L(x) = b(o, p) \wedge 
    \Perm Q \in \Gamma \Longrightarrow p=\gamma(Q) \in P
    \)
    
  \item By 2.a), T-capture, 
    \begin{enumerate} [label=(\alph*)]
    \item \(\Gamma ~;~a \vdash x \typ \GuaT Q C\)
    \item \(\Gamma ~;~a \vdash y \typ \GuaT {Q'} D\)
    \item \(\{\Perm Q, \Perm Q'\} \subseteq \Gamma\)
    \item \(D \sub ftype(C, f)\)
    \end{enumerate}

  \item By 2.b), \begin{enumerate} [label=(\alph*)]
    \item \(dom(\Gamma) \subseteq dom(L)\)
    \item \(\forall x \in dom(\Gamma), H \vdash \Gamma ; L ; x\)
    \end{enumerate}


  \item By 5.b), 4.a-b), assuming \(L(x) \neq \texttt{null} \wedge
    L(y) \neq \texttt{null}\) (otherwise the theorem conclusion holds),
    \begin{enumerate} [label=(\alph*)]
    \item \(L(x) = b(o, p)\)
    \item \(typeof(H, o) \leq C\)

    \item \(L(y) = b(o', p')\)
    \item \(typeof(H, o') \leq D\)
    
    \end{enumerate}

  \item By 4.a), 4.b), 4.c), 7.a), 7.c), 3,
    \begin{enumerate} [label=(\alph*)]
    \item \(p \in P\)
    \item \(p' \in P\)
    \end{enumerate}

  \item Define \(\obj {C'} {FM} := H(o)\)
    
  \item Define 
    \begin{enumerate} [label=(\alph*)]
    \item \(L' := L[z \mapsto L(x)]\)
    \item \(H' := H[o \mapsto \obj {C'} {FM[f \mapsto o']}]\)
    \item \(P' := P \setminus \{p'\}\)
    \item \(F' := \pframe {L'} {t'} {P'} {\epsilon}\)
    \end{enumerate}
    
  \item By E-Capture, \(\fsreduce H {F \circ FS'}
    {H'} {F' \circ \epsilon}
    \)
  \item 10. concludes this case

  \end{enumerate}

\item [-] Cases \(t = \texttt{swap}(x.f, y) \{ z \Rightarrow t'\}\),\\
  \(\texttt{let}~x = y.f = z ~\texttt{in}~ t'\),\\
  \(\texttt{let}~x = y.m(z) ~\texttt{in} ~t'\)
  are left out and are proven 
  analogously to the \texttt{capture} and \texttt{select}
  cases.
  
\end{itemize}
\end{proof}


%% file: lacasa_mechanized.tex
\section{Formulation of key definitions and 
theorems in Coq}\label{sec:coq}

Static and dynamic operational semantics of \CLCONE were mechanized in the Coq theorem proving system.
The Coq mechanization closely follows the definitions in this paper and 
is inspired by \cite{Mackay12}. 
\subsubsection*{Partial functions}
An important difference is the use
of explicit partial functions over finite domains defined as follows in Coq:
\begin{lstlisting}[numbers=left, numberstyle=\scriptsize\color{gray}\ttfamily]
Module Partial (T: InfiniteTypeWithDecidableEquality) 
               (B: Typ). 
  Definition B := BT.t.
  Notation A := T.t.

  Record PartFunc : Type :=
   mkPartFunc {
       func: A -> option B;
       domain: list A;
       fDomainCompat :
         forall valT: A,
           ~ In valT domain <-> func valT = None
     }.
  ...
\end{lstlisting}
Partial functions are defined in a module, which is instanciated with different
domain and range types for 
the heap, environment, typing context and field map:
\begin{lstlisting}[numbers=left, numberstyle=\scriptsize\color{gray}\ttfamily]
Inductive FM_Range_type :=
  | FM_null : FM_Range_type
  | FM_ref : Ref_type -> FM_Range_type.

Module FM_typeM <: Typ .
  Definition t := FM_Range_type.
End FM_typeM.

Module p_FM := Partial FieldNameM FM_typeM.

Definition FM_type := p_FM.PartFunc.

Inductive RTObject :=
 | obj : ClassName_type -> FM_type -> RTObject.

Module RTObject_typeM <: Typ .
  Definition t := RTObject.
End RTObject_typeM.

Module p_heap := Partial RefM RTObject_typeM.

Definition Heap_type := p_heap.PartFunc.

Inductive env_Range_type :=
| envNull : env_Range_type
| envRef : Ref_type -> env_Range_type
| envBox : Ref_type -> env_Range_type.

Module env_Range_typeM <: Typ.
  Definition t := env_Range_type.
End env_Range_typeM.

Module p_env := Partial VarNameM env_Range_typeM.

Definition Env_type := p_env.PartFunc.

Inductive typecheck_type :=
| typt_class : ClassName_type -> typecheck_type
| typt_box : ClassName_type -> typecheck_type
| typt_all : typecheck_type.

Inductive effect :=
| eff_ocap : effect
| eff_epsilon : effect.

Module tyM <: Typ.
  Definition t := typecheck_type.
End tyM.

Module p_gamma := Partial VarNameM tyM.

Notation Gamma_type := p_gamma.PartFunc.
\end{lstlisting}
\subsection*{Mutual recursion}
Several definitions in this paper utilize mutual recursion, which is not well supported in Coq.
\subsubsection*{Mutual recursion in syntax of \CLCONE}
The syntax of \CLCONE has mutual recursion between terms \verb|t = let x = e in t|
and expressions \verb|e = x.open{ y => t }|.
We solved this problem by defining a combined inductive type
\verb|ExprOrTerm| and defining predicates \verb|isExpr|, \verb|isTerm| for
use when an expression or term is expected:
\begin{lstlisting} [numbers=left, numberstyle=\scriptsize\color{gray}\ttfamily]
Inductive ExprOrTerm :=
| Null : ExprOrTerm
| Var : VarName_type -> ExprOrTerm
| FieldSelection : VarName_type -> FieldName_type
  -> ExprOrTerm
| FieldAssignment : VarName_type -> FieldName_type
 -> VarName_type -> ExprOrTerm
| MethodInvocation : VarName_type -> MethodName_type
 -> VarName_type -> ExprOrTerm
| New : ClassName_type -> ExprOrTerm
| Box : ClassName_type -> ExprOrTerm
| Open : VarName_type -> VarName_type -> ExprOrTerm
 -> ExprOrTerm
| TLet : VarName_type -> ExprOrTerm -> ExprOrTerm
 -> ExprOrTerm.

Fixpoint isTerm (e: ExprOrTerm) : Prop :=
  match e with
   | Var _ => True
   | TLet _ e t => (fix isExpr (e: ExprOrTerm) : Prop :=
                      match e with
                        | TLet _ _ _ => False
                        | Open _ _ t' => isTerm t'
                        | _ => True
                      end
                   ) e /\ isTerm t
   | _ => False
  end.

Definition isExpr (e: ExprOrTerm) : Prop :=
  match e with
    | TLet _ _ _ => False
    | Open _ _ t => isTerm t
    | _ => True
  end.
\end{lstlisting}

\subsubsection*{Mutual recursion in  T-FS}
The definitions and rules T-FS-NA, T-FS-NA2, T-FS-A, T-FS-A2
employ mutual recursion betwen 
\(H \vdash FS\) and \(H \vdash_x^\tau FS\).
We defined this rule in Coq by combining the two cases into
a single case \(H \vdash^a FS\), where \(a := \epsilon~|~(x, \tau)\):
\begin{lstlisting} [numbers=left, numberstyle=\scriptsize\color{gray}\ttfamily]
Inductive WF_FS : FS_ann_type -> Heap_type -> 
                 list (ann_frame_type) -> Type :=
  
| T_EmpFS : forall H, WF_FS None H nil
                            
| T_FS_NA : forall H F FS sigma,
              WF_Frame H (ann_frame F ann_epsilon) 
                sigma ->
              WF_FS None H FS ->
              WF_FS None H ((ann_frame F ann_epsilon)
                      :: FS)
                    
| T_FS_NA2 : forall H F FS sigma x tau,
               ( H , x , tau ## (ann_frame F 
                   ann_epsilon) @@ sigma ) ->
               WF_FS None H FS ->
               WF_FS (Some (x, tau)) H 
                 ((ann_frame F ann_epsilon) :: FS)
                     
| T_FS_A : forall H F FS x tau ,
              WF_Frame H (ann_frame F (ann_var x)) 
                 tau ->
              WF_FS (Some (x, tau)) H FS ->
              WF_FS None H ((ann_frame F 
                (ann_var x))  :: FS)
| T_FS_A2 : forall H F FS y sigma x tau,
              ( H, y, sigma ## (ann_frame F 
                  (ann_var x)) @@ tau) ->
              WF_FS (Some (x, tau)) H FS ->
              WF_FS (Some (y, sigma)) H 
                ((ann_frame F (ann_var x)) :: FS).
\end{lstlisting}

\subsection*{Reduction and typing}
Reductions and typing rules are defined as follows:

\begin{lstlisting}
Inductive TypeChecks : Gamma_type -> effect ->
         ExprOrTerm -> typecheck_type -> Type :=
| T_Null : forall gamma eff,
      TypeChecks gamma eff Null typt_all
                       
| T_Var : forall gamma eff x sigma,
   p_gamma.func gamma x = Some sigma ->
    TypeChecks gamma eff (Var x) sigma
| T_Field : forall gamma eff x f C,
   forall witn: fldP C f,
     p_gamma.func gamma x = Some (typt_class C) ->
     TypeChecks gamma eff (FieldSelection x f)
      (typt_class (ftypeP C f witn))
| T_Assign : forall gamma eff x f y C D,
   forall witn: fldP C f,
     p_gamma.func gamma y = Some (typt_class C) ->
     TypeChecks gamma eff (FieldSelection x f) 
              (typt_class D) ->
     subtypeP (typt_class C) (typt_class D) ->
     TypeChecks gamma eff (FieldAssignment x f y) 
            (typt_class C)
| T_New : forall gamma C eff,
    TypeChecks gamma eff (New C) (typt_class C)
| T_Open : forall gamma eff x C y t sigma,
    TypeChecks gamma eff (Var x) (typt_box C) ->
    TypeChecks (p_gamma.updatePartFunc 
                  p_gamma.emptyPartFunc
                  y (typt_class C)
               ) eff t sigma ->
    TypeChecks gamma eff (Open x y t) (typt_box C)
| T_Let : forall gamma eff e sigma x tau t,
  TypeChecks gamma eff e sigma ->
  TypeChecks (p_gamma.updatePartFunc gamma x sigma) eff
             t tau ->
  TypeChecks gamma eff (TLet x e t) tau.

 Inductive Reduction_FS : cfg_type -> cfg_type -> Type:=
 | E_StackFrame : forall H H' L L' t t' FS a,
     Reduction_SF ( # H , L , t ! ) ( # H' , L' , t' ! )  ->
     Reduction_FS (H, ( (ann_frame (sframe L t)  a) :: FS) )
                  (H', (ann_frame (sframe L' t') a) :: FS )
 | E_Return1 : forall H L x y F FS envVal,
    p_env.func L x = Some envVal ->
    Reduction_FS (H, 
       (ann_frame (sframe L (Var x)) (ann_var y)) :: F :: FS)
                 (H, updFrame F y envVal :: FS)
 | E_Return2 : forall H F FS,
    Reduction_FS (H, (ann_frame F ann_epsilon) :: FS)
                 (H, FS)
                 
 | E_Open : forall H L x1 x2 y t1 t2 ann FS o,
  isTerm t1 ->
  isTerm t2 ->
  p_env.func L x2 = Some (envBox o) ->
  In o (p_heap.domain H) ->
  Reduction_FS (H, (ann_frame (sframe L
              ( t_let x1 <- (Open x2 y t1) t_in t2 ))
       ann) :: FS)
   (H, (ann_frame (sframe 
       (p_env.emptyPartFunc +++ y --> (envBox o) )  t1)
             ann_epsilon)
    :: (ann_frame (sframe ( L +++ x1 --> (envBox o) ) t2)
                     ann) :: FS).
\end{lstlisting}

\subsection*{Heap properites and invariants}
\(\vdash H \typ \star \) and
\( H \vdash \Gamma, L\) were defined as follows:
\begin{lstlisting} [numbers=left, numberstyle=\scriptsize\color{gray}\ttfamily]
Definition Heap_obj_ok H C FM :=
  forall f o,
  forall f_witn : fldP C f, 
    p_FM.func FM f = Some (FM_ref o) ->
    { o_witn: In o (p_heap.domain H) &
              subclassP (heap_typeof H o o_witn)
                        (ftypeP C f f_witn)
    }.
Definition Heap_dom_ok H : Prop :=
  forall o C FM,
    p_heap.func H o = Some (obj C FM) ->
    fieldsP C (p_FM.domain FM).
Definition Heap_ok H : Prop :=
  forall o C FM,
    p_heap.func H o = Some (obj C FM) ->
    Heap_obj_ok H C FM.
Definition WF_Var (H: Heap_type) (Gamma: Gamma_type)
       (L: Env_type) (x: VarName_type):=
  (p_env.func L x = Some envNull) +
  {C_o |
   match C_o with
     | (C, o) => {witn |
     p_env.func L x = Some (envRef o) /\
     p_gamma.func Gamma x = Some (typt_class C) /\
     subtypeP (typt_class (heap_typeof H o witn)) 
        (typt_class C)
    }
   end} +
  {C_o |
   match C_o with
     | (C, o) => {witn |
      p_env.func L x = Some (envBox o) /\
      p_gamma.func Gamma x = Some (typt_box C) /\
      subtypeP (typt_class (heap_typeof H o witn)) 
   (typt_class C)
     }
   end}.
Definition WF_Env H Gamma L :=
  gamma_env_subset Gamma L *
  forall x sigma,
    p_gamma.func Gamma x = Some sigma ->
    WF_Var H Gamma L x.
  
\end{lstlisting}

\subsection*{Preservation theorems for \CLCONE}
We defined and proved simplified versions of 
Theorem~\ref{thm:core1-preservation}:
\begin{lstlisting}[numbers=left, numberstyle=\scriptsize\color{gray}\ttfamily]
Theorem single_frame_WF_ENV_preservation :
  forall H H' L L' t t' sigma ann,
    forall (heap_dom_ok: Heap_dom_okP H),
    WF_Frame' H (ann_frame (sframe L t) ann) sigma ->
    Heap_okP H ->
    Reduction_SF' ( # H, L, t !) ( # H', L', t' !) ->
    ((WF_Frame' H' (ann_frame (sframe L' t') ann) sigma) *
     (Heap_okP H') *
     (Heap_dom_okP H')
    ).
\end{lstlisting}

\begin{lstlisting} [numbers=left, numberstyle=\scriptsize\color{gray}\ttfamily]
Theorem single_frame_WF_ENV_preservation :
  forall H H' L L' t t' sigma ann,
    forall (heap_dom_ok: Heap_dom_okP H),
    WF_Frame' H (ann_frame (sframe L t) ann) sigma ->
    Heap_okP H ->
    Reduction_SF' ( # H, L, t !) ( # H', L', t' !) ->
    ((WF_Frame' H' (ann_frame (sframe L' t') ann) sigma) *
     (Heap_okP H') *
     (Heap_dom_okP H')
    ).

Notation "[ H ## FS ]" := (WF_FS P None H FS) (at level 0).

Theorem multiple_frame_WF_ENV_preservation :
  forall H H' FS FS',
    Reduction_FS' (H, FS)
                 (H', FS') ->
    Heap_okP H ->
    [ H ## FS ]  ->
    [ H' ## FS' ] * (Heap_okP H').

\end{lstlisting}
